\documentclass[11pt, oneside]{amsart}   	
\usepackage[letterpaper, margin=2cm]{geometry}
\usepackage[parfill]{parskip}    						
\usepackage{amssymb, amsthm, amsmath, mathtools, bm, bbm}
\usepackage{enumitem}
\usepackage[obeyspaces,hyphens,spaces]{url}
\usepackage{mdframed}
\usepackage{mathrsfs}
\usepackage{accents}

\SetLabelAlign{Center}{\hfil#1\hfil}
\SetLabelAlign{CenterWithParen}{\hfil(\makebox[1.0em]{#1})\hfil}

\usepackage{graphicx, overpic}
\usepackage[usenames,dvipsnames]{xcolor}

\graphicspath{{./Figures/}}

\allowdisplaybreaks
\numberwithin{equation}{section}

\usepackage{hyperref}
\hypersetup{
    bookmarks=true,         	
    unicode=false,          		
    pdftoolbar=true,        		
    pdfmenubar=true,        	
    pdffitwindow=false,     		
    pdfstartview={FitH},    		
   pdftitle={title},    					
   	pdfauthor={people},     	
    linktocpage=true,			
    pdfnewwindow=true,      	
    colorlinks=true,       			
    linkcolor=blue,          		
    citecolor=PineGreen,    	
    filecolor=magenta,      		
    urlcolor=cyan           		
}

\usepackage{tikz} 
\usetikzlibrary{shapes,snakes,calc,fit}
\usepackage{pgf}
\usepackage{pgfplots}
\usepackage{pgflibraryarrows}
\usepackage{pgflibrarysnakes}
\usetikzlibrary{decorations.text}
\usepgfmodule{shapes}
\usetikzlibrary{decorations.pathmorphing}
\usetikzlibrary{decorations.markings}
\usetikzlibrary{patterns}
\usetikzlibrary{automata}
\usetikzlibrary{positioning}
\usetikzlibrary{spy}
\usetikzlibrary{backgrounds}
\tikzset{partial ellipse/.style args={#1:#2:#3}{insert path={+ (#1:#3) arc (#1:#2:#3)} }}
\tikzset{->-/.style={decoration={ markings, mark=at position #1 with {\arrow{>}}},postaction={decorate}}}
\tikzset{-<-/.style={decoration={ markings, mark=at position #1 with {\arrow{<}}},postaction={decorate}}}

\theoremstyle{theorem}
\newtheorem{theorem}{Theorem}[section]
\newtheorem{lemma}[theorem]{Lemma}
\newtheorem{corollary}[theorem]{Corollary}
\newtheorem{remark}[theorem]{Remark}

\newtheorem{assumption}[theorem]{Assumption}

\newtheorem{prop}[theorem]{Proposition}
\newtheorem{thm}[theorem]{Theorem}

\newtheorem{note}[theorem]{Note}

\theoremstyle{definition}

\newtheorem{RHP}{Riemann--Hilbert problem}

\renewcommand{\(}{\left(}
\renewcommand{\)}{\right)}

\def\le{\left}
\def\ri{\right}
\def\d{{\rm d}}
\newcommand{\rvline}{\hspace*{-\arraycolsep}\vline\hspace*{-\arraycolsep}}

\newcommand{\od} [3] [ ] {\frac{\mathrm{d}^{#1}#2}{\mathrm{d}{#3}^{#1}} }
\newcommand{\pd} [3] [ ] {\frac{\partial^{#1}#2}{\partial{#3}^{#1}} }

\newcommand{\R}{\mathbb{R}}
\newcommand{\Z}{\mathbb{Z}}
\newcommand{\C}{\mathbb{C}}
\newcommand{\Id}{\operatorname{Id}}

\newcommand{\bigo}[1]{\mathcal{O} \left( #1 \right) }

\def\Re{\operatorname{Re}}  
\def\Im{\operatorname{Im}}
\newcommand{\Res}{\operatorname*{Res}}

\DeclareMathOperator{\dn}{dn}
\DeclareMathOperator{\nd}{nd}
\DeclareMathOperator{\sn}{sn}

\newcommand{\sgn}{\operatorname{sgn}}
\newcommand{\sech}{\operatorname{sech}}
\newcommand{\ol}{\overline}

\def\Tr{\mathop{{\rm Tr}}}
\DeclareMathOperator{\inside}{int}

\usepackage[backend=bibtex,maxbibnames=99,style=alphabetic]{biblatex}
\bibliography{SolitonGas.bib}

\begin{document}

\title[Soliton versus the gas]{Soliton versus the gas: \\
Fredholm determinants, analysis,  and the rapid oscillations behind the kinetic equation}

\author{Manuela Girotti}
\author{Tamara Grava}
\author{Robert Jenkins}
\author{Ken T-R McLaughlin}
\author{Alexander Minakov}

\address[Girotti]{Department of Mathematics and Computing Sciences, Saint Mary's University, 923 Robie St, Halifax, NS B3H 3C3}
\email{manuela.girotti@smu.ca}
\address[Grava]{SISSA, via Bonomea 265, 34136 Trieste, Italy and School of Mathematics, University of Bristol, UK}
\email{grava@sissa.it}
\address[Jenkins]{Department of Mathematics, University of Central Florida, 4393 Andromeda Loop N, Orlando, FL 32816}
\email{robert.jenkins@ucf.edu}
\address[McLaughlin]{Department of Mathematics, Tulane University, 6823 St Charles Ave, New Orleans, LA 70118}
\email{kmclaughlin@tulane.edu}
\address[Minakov]{Department of Mathematical Analysis, Univerzita Karlova,  Sokolovsk\'a 83 Praha 8, 186 75 Czech Republic}
\email{minakov@karlin.mff.cuni.cz}

\date{\today}

\begin{abstract}
We analyze the case of a {dense} modified Korteweg de Vries (mKdV) soliton gas and its large time behaviour in the presence of a single trial soliton. We show that the solution can be expressed in terms of Fredholm determinants as well as in terms of a Riemann--Hilbert problem.  We then show that the solution can be decomposed as the sum of the background gas solution (a modulated elliptic wave), plus a soliton solution: the individual expressions are however quite convoluted due to the interaction dynamics.  Additionally, we are able to derive the  local phase shift of the gas after the passage of the soliton, and we can trace the location of the soliton peak as the dynamics evolves.
Finally we show that the soliton peak, while interacting with the soliton gas, has an oscillatory velocity whose leading order average value satisfies  the  kinetic   velocity equation analogous to the one posited by V. Zakharov and G. El. 
\end{abstract}

\maketitle

\tableofcontents

\section{Introduction}
\label{sec-intro}

Solitons are fundamental, localized solutions of nonlinear evolution equations.  They can appear as single entities; traveling wave solutions that propagate without deformation.  They can also appear as ensembles, evolving as a collective and asymptotically decomposing into isolated solitons and (possibly) sub-ensembles of solitons.  Since the discovery of ensembles of solitons, the interpretation of them as particles has been a source of novel investigations.

This paper concerns the interaction of a single soliton with a dense ``gas" of solitons in the setting of the (focusing) modified KdV equation.  The mKdV equation, 
\begin{eqnarray}\label{mkdv}
q_{t} + 6 q^{2} q_{x} + q_{xxx} = 0 \ ,\quad x\in\mathbb{R},\;t\in\mathbb{R}_+,
\end{eqnarray}
has a soliton solution of the form
\begin{equation}
\label{soliton}
q(x,t) =\pm 2 \kappa \sech [2\kappa ( x  - 4 \kappa^2 t - x_0 ) ]
\end{equation}
where the quantity $4\kappa^2$ is the wave velocity, $\kappa \in \R_+$, and $x_0 $ is a phase shift. The solution with the positive hump  is the standard  \emph{soliton}, while the negative amplitude solution is called an \emph{anti-soliton}.

For general solutions on the line which decay  sufficiently quickly at infinity, it is well known that $q$ solves the mKdV equation if and only if there is a simultaneous solution ${\bm \Phi}={\bm \Phi}(k;x,t)$, $k\in\mathbb{C}$, to the following Lax pair equations:
\begin{equation}
\label{eq:MKdVLaxPair}
\begin{aligned}
{\bm \Phi}_{x} &= \begin{pmatrix}
- i k & q(x,t) \\
- q(x,t) & i k \\
\end{pmatrix}
{\bm \Phi}, \\
{\bm \Phi}_{t} &= 
\begin{pmatrix} - 4 i k^{3} + 2 i k q^{2} & 4 k^{2} q + 2 i k q_{x} - 2 q^{3} - q_{xx} \\
- 4 k^{2} q + 2 i k q_{x} + 2 q^{3} + q_{xx} & 4 i k^{3} - 2 i k q^{2} \\
\end{pmatrix} {\bm \Phi} \ . 
\end{aligned}
\end{equation}
The spectral data associated to \eqref{eq:MKdVLaxPair} is comprised of the continuous spectrum and its reflection coefficient $r(k)$ for $ k \in \mathbb{R}$ and the discrete spectrum.
Generic solitons and breathers correspond to simple eigenvalues,  $\pm i\kappa_j$,  $\kappa_j>0$ and $\pm \kappa_j$, $\pm \overline{\kappa_j}, $, $\Im \kappa_j>0$,  respectively.  In this manuscript we consider only the case of discrete spectrum associated to solitons  $\pm i \kappa_{j}, \ \kappa_{j} >0$, $j=1,\dots,N$, which, together with  the associated norming constants $\chi_{j}\in\R\backslash\{0\}$ (one for each eigenvalue) and  reflection coefficient $r(k)$, completely determines the function $q$.  

A one soliton solution with velocity $4\kappa^2$ is reflectionless with a discrete spectrum consisting of a single eigenvalue $i\kappa$. The norming constant $\chi$ and eigenvalue $i\kappa$ determine the phase shift $x_0$ in \eqref{soliton} via the formula
\begin{equation}
\label{x0}
x_0=\frac{1}{2\kappa}\ln\frac{2\kappa}{|\chi|}\in\R. 
\end{equation}
The sign of $\chi$ determines whether the solution is a soliton ($\chi>0$) or an anti-soliton ($\chi<0$).

Finally, as $q(x,t)$ evolves according to the mKdV equation, the spectral data  have a simple and integrable evolution.  The problem of reconstructing  the solution $q(x,t)$ to the mKdV equation at any time $t>0$, from the evolved spectral data, is referred to as the inverse problem.

The characterization of the spectral data and the formulation of the inverse problem are achieved via a detailed consideration of fundamental sets of solutions to equation (\ref{eq:MKdVLaxPair}).  This has been described in \cite{Wadati72,BealsCoifman84,ItsNovokshenov86} for initial data decaying at infinity and in \cite{GravaMinakov20} for step-like initial data. The inverse problem can be formulated as a Riemann--Hilbert (RH) problem.

The solution of the mKdV equation that represents an ensemble of  $N$ solitons (without radiative component) can be written in terms  of a Fredholm determinant  
\begin{equation}
q_{N}(x,t)= i \dfrac{\partial }{\partial x}\ln\det \le(\Id_{L^2(-\infty,x)} + \mathcal K_N\ri) - i \dfrac{\partial }{\partial x}\ln\det \le(\Id_{L^2(-\infty,x)} - \mathcal K_N\ri)\ , \label{intro:mKdV_sol_FD_N}
\end{equation}
where $  \mathcal K_N$ is an integral operator of finite rank, with kernel
\[
F_N(s,t)=-i \sum_{j=1}^N \chi_j e^{-i \theta_j( s,t)}\ , \quad  \theta_j(s,t) = x \kappa_j - 4t \kappa_j^3 \ . 
\]
Here, each soliton is characterized by the norming constant $\chi_j$ and the  point spectrum $i \kappa_j$, for $j=1,\dots, N$. 
The above formula is  derived in Section~\ref{sec-fredholm} formul\ae \  \eqref{q_fredh1}--\eqref{q_fredh2} from the work of Wadati \cite{Wadati72}.  

\vskip .2in

\paragraph{\bf Convergence to a soliton gas.}
In the present paper we are interested in purely solitonic solutions where the number of solitons $N$ goes to infinity, while $(x,t)$ lies in (arbitrarily large) compact sets of $\R\times \R_+$, and  their spectrum is confined in an interval $[i\eta_1,i\eta_2]$: we call such types of solution a \emph{dense} soliton gas.

Once the limit is taken, we consider the soliton gas on the whole real line $x\in \R$. Similarly to the analysis conducted in \cite{GirGraJenMcL}, at fixed values of $t\in \R_+$, such a soliton gas converges slowly  to an elliptic wave as $x \to +\infty$, while it converges rapidly to zero as $x\to -\infty$. 

The analysis establishing the existence of a solution in the $N \to \infty$ limit is fairly straightforward and could be carried out in a manner entirely similar to what was done for the case of the KdV equation in \cite{GirGraJenMcL}.  Here we take a different approach, and characterize the $N$-soliton solution in terms of Fredholm determinants, for which the $N \to \infty$ limit can be established. The additional characterization of the solution in terms of a RH problem follows from techniques from \cite{BertolaCafasso11}. 

Our main goals in this paper are: (1) to prove that the kinetic theory applies to our soliton gas;  and (2) to provide a very detailed description of the highly oscillatory interaction between a large trial soliton travelling through a soliton gas from which the averaged velocity of the trial soliton can be determined.  These two results were not considered in \cite{GirGraJenMcL}.

Although we do not consider any randomness in the initial configuration of solitons, it is certainly possible to introduce randomness into the $N$-soliton data, in such a way that the average behaviour of the soliton gas is captured by our analysis.

The concept of an infinite ensemble of solitons was already analysed by Zaitsev \cite{Zaitsev} and Boyd \cite{Boyd}. In these papers, it is shown that the sum of an infinite number of equally spaced and identical solitons coincides with the elliptic solution of the KdV equation. Furthermore, Gesztesy, Karwowski and Zhao  showed the existence of the infinite soliton limit for KdV when the point spectrum is bounded and the norming constants have sufficient decay \cite{Gesztesy92}.

The notion of a soliton gas was first put forth by Zakharov in 1971, for the case of the KdV equation,
\begin{eqnarray}
q_{t} + q q_{x} + q_{xxx} = 0. \ 
\end{eqnarray}
The fundamental calculation is to   prepare a two-soliton solution so that for $t \to - \infty$, one has
\begin{eqnarray*}
q(x,t) \approx 12 \kappa_{1}^{2}\, \mbox{sech}^{2}\left(\kappa_{1}(x-4 \kappa_{1}^{2}t)\right)+ 12 \kappa_{2}^{2}\, \mbox{sech}^{2}\left(\kappa_{2}(x-4 \kappa_{2}^{2}t)
\right) \ ,
\end{eqnarray*}
with $\kappa_{1} < \kappa_{2}$, and there are two isolated peaks, a taller one at (roughly) $x = 4 \kappa_{2}^{2} t$ and a smaller one at $x = 4 \kappa_{1}^{2}t$, with the taller one significantly further to the {left} of the smaller one.
Then, for $t \to + \infty$, one finds
\begin{eqnarray*}&&
q(x,t) \approx 12 \kappa_{1}^{2} \, \mbox{sech}^{2}\left[\kappa_{1}\left(x-4 \kappa_{1}^{2}t + \frac{1}{\kappa_{1}}\ln\left|\frac{\kappa_{2}+\kappa_{1}}{\kappa_{2}-\kappa_{1}}
\right|\right)\right]
+
 \\&&
\qquad \qquad \qquad 
+12 \kappa_{2}^{2}\, \mbox{sech}^{2}\left[\kappa_{2}\left(x-4 \kappa_{2}^{2}t-\frac{1}{\kappa_{2}}\ln\left|\frac{\kappa_{2}+\kappa_{1}}{\kappa_{2}-\kappa_{1}}\right|\right)\right] ,
\end{eqnarray*}
and one sees that the two peaks are again isolated, and travelling to the right, but now the taller peak has overtaken the smaller one, with both peaks having received a shift in the position of their centers (manifested by the logarithmic terms appearing above).  
\begin{figure}
\hspace*{\stretch{1}}
	\begin{overpic}[scale=0.4]{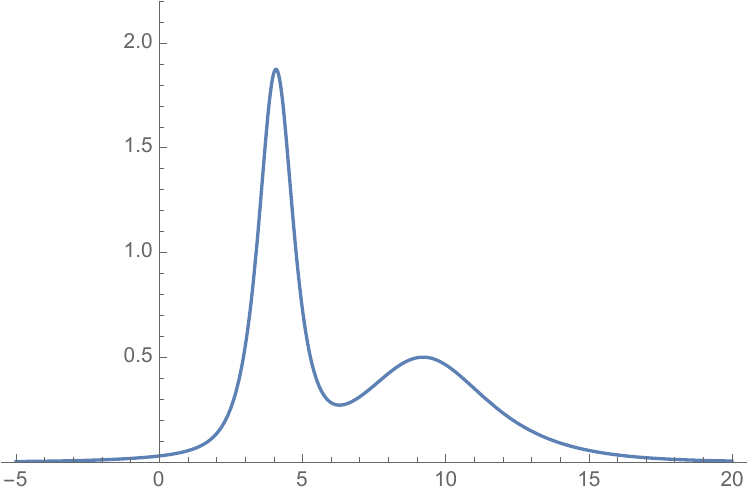}
	   \put(40,65){$t=2.358$}
	   \put(8,70){$\tiny {}_{q(x,t)}$}
	   \put(102,5){$\tiny {}_{x}$}
	\end{overpic}
\hspace*{\stretch{1}}
	\begin{overpic}[scale=.4]{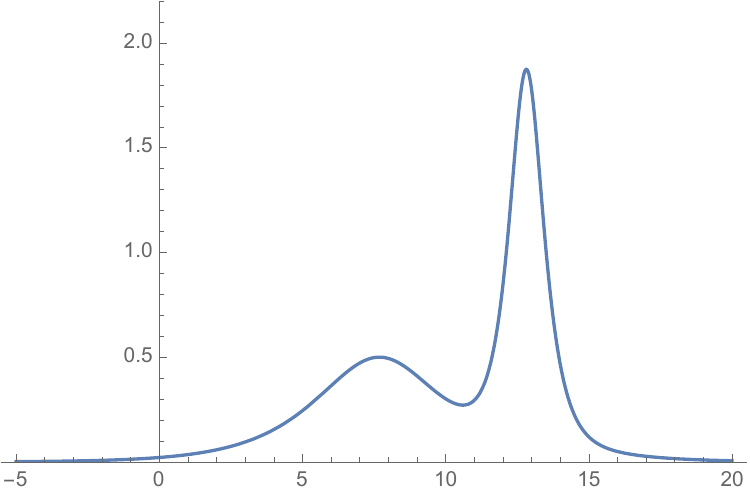}
	   \put(34,65){$t= 7.073$}
	   \put(8,70){$\tiny {}_{q(x,t)}$}
	   \put(102,5){$\tiny {}_{x}$}
	\end{overpic}
\hspace*{\stretch{1}}
\caption{
A two soliton solution to the mKdV equation, determined by RH problem \eqref{rhp:M}, with $N=2$, $\kappa_{1} = 1/4, \kappa_{2} = 1$, $\chi_{2} = 2$, and $\chi_{1} = \frac{25}{2^{1/4} 3^2 e^{5}}$.  The left plot is at $t=2.358$, and the right plot is at $t=7.073$.  Note that the smallest peak emerges shifted to the left, though the solitons themselves, if they were each propagating independent of the other, would move to the right with constant velocities.}
\label{fig:phase}
\end{figure}
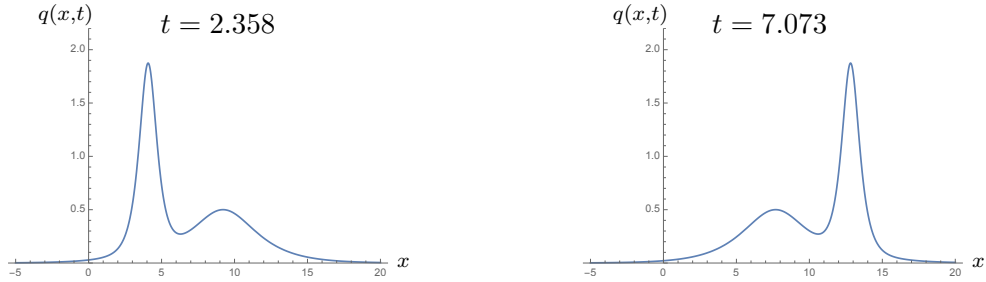

Zakharov considered ``the propagation of an individual soliton in a `gas' " \cite[p.540]{Z71}.  Presuming the soliton gas to be \emph{dilute}, the trial soliton interacts with each member of the gas in sequence, and continually accumulates  shifts in its position, which effectively alter its velocity.  He extracted the kinetic equation for the trial soliton velocity, $v=v(\kappa;x,t)$, where the soliton's initial speed (before interaction with the gas) was $4 \kappa^{2}$.  
\begin{eqnarray}\label{eq:00}
v(\kappa) = 4 \kappa^{2} + \frac{1}{\kappa} \int_{0}^{\infty} \ln\left|
\frac{s+\kappa}{s-\kappa}
\right| \left(\kappa^2-s^2\right) f(s)\,  \d s \ .
\end{eqnarray}
The quantity  $f=f(\kappa; x,t)$ is the distribution function for the soliton gas at time $t$,  with respect to the spectral parameter $\kappa$ and the spatial coordinate $x$, i.e. it is the density in the phase space, giving the number of solitons per unit interval of the spectrum and per unit interval of space.
Equation (\ref{eq:00}) was derived under the assumption that $f$ is small (i.e. the gas is dilute) relative to the characteristic value of $s$. 

\vskip .2in

\paragraph{\bf The interaction of a trial soliton with a dense soliton gas.}

In 2003 \cite{El03}, El derived an integral equation for the velocity   $v=v(\kappa; x,t)$  of the trial soliton with point spectrum $\kappa$  propagating in a \emph{dense} KdV soliton gas:
\begin{equation}
\label{eq:01}
v(\kappa) = 4 \kappa^{2} + \frac{1}{\kappa} \int_{0}^{\infty} \ln\left|
\frac{s+\kappa}{s-\kappa}
\right| \left(v(\kappa)-v(s)\right) f(s)\,  \d s \ ,
\end{equation}
where the soliton density  $f=f(\kappa; x,t)$  satisfies the continuity equation
\begin{equation}\label{eq:02}
f_{t} + \left( v\, f \right)_{x} = 0.
\end{equation}
Analytical developments concerning the small dispersion limit of the KdV equation in the '80s, as it turns out, are connected to the same velocity equations \eqref{eq:01}--\eqref{eq:02}. In \cite{LaxLevI,LaxLevII,LaxLevIII}, Lax and Levermore considered the behaviour of solutions of the KdV equation in the small dispersion limit, 
\[
q_{t} - 6 q q_{x} + \epsilon^{2} q_{xxx} = 0,  \qquad \ q(x,0) = q_{0}(x) \ ,
\]
in which $\epsilon \to 0$, while the initial data is assumed to be independent of $\epsilon$.  They showed that the (infinite) solitonic component of the spectral data associated to the initial data drives the dynamics, both before and after the formation of a shock in the underlying Hopf equation, which suggests that the small dispersion limit of the KdV equation is an instance of a regular soliton gas.  Indeed, the methods presented in this paper can be used to show that the Lax--Levermore minimizer yields densities that encode all aspects of the kinetic theory.

The kinetic theory for solitons  has been extended to the cases of focusing NLS equation \cite{ElKa05}, defocusing and resonant NLS equations \cite{CongyElRoberti21}, and to the case of breather gasses in the NLS equation  \cite{Tovbis2}.  Such configurations have been shown to exist and their properties studied in experimental and theoretical investigations in  \cite{LowmanHoefer,LowmanHoeferEl,MaidenFrancoWebbHoefer}. 
From these considerations, a soliton gas is thought of as a continuum accumulation of large numbers of solitons that are in some sense random, and the kinetic equations provide a description of the velocity of a trial soliton (averaged over realizations of the soliton gas).

Computational simulations for ensembles of well-separated solitons in the KdV equation were carried out in \cite{PelinovskySergeeva,SergeevaPelinovskyTalipova,DutPel}, and in the mKdV equation in \cite{PelShu}. In these works, asymptotic methods and Monte Carlo simulations were used to study statistical properties of solutions beyond the mean behaviour.  In particular in \cite{PelShu} for the case of ``bipolar soliton fields" for the modified KdV equation, the likelihood of large-amplitude waves increases during the dynamical evolution.

The above mentioned computational works are stochastic in nature.  In our work, once the $N \to \infty$ limit is taken, the entire evolution is completely determined by a spectral function $r(k)$, so there is no randomness in the soliton gas.  Nonetheless, surprisingly, our analysis proves that the same Zakharov--El kinetic equations are satisfied.  More precisely, we explicitly describe the oscillations of the soliton while interacting with the gas and the phase shift effect on the gas, and  we show that the soliton velocity \emph{averaged over one period} satisfies such kinetic equations. Incidentally, our gas appears to be in the same class of solutions considered in \cite{Tovbis2} as \emph{condensate gasses}.

Looking to the near future, our methods should permit the introduction of randomness into the finite-$N$ soliton setup. The analysis as $N \to \infty$ to prove that the average behavior of the soliton gas is as described in this paper, would be very interesting.  Of course, an analytical description of the statistical fluctuations is a major challenge, which (to the best of our knowledge) has not been  rigorously considered in the literature.

At this stage it is useful to provide some precision to the notion of what constitutes a soliton gas.  A natural definition (it seems to us) is that a solution to a nonlinear wave equation should be considered a \emph{regular} soliton gas if (1) it arises as the $N \to \infty$ limit of a solution containing $N$ solitons (or other collections of soliton-like solutions naturally indexed by $N$), and (2) the averaged velocity of a trial soliton interacting with the background solution satisfies the above kinetic equation \eqref{eq:01}.  Although it may appear to be a ``fait accompli", by this definition our solution constructed as $N \to \infty$ is indeed a soliton gas.

\vskip .2in

\paragraph{\bf Interpretation as the interaction of a particle with a complex medium.}

In this paper we also prove that the soliton itself has a nontrivial effect on the (background) soliton gas solution, and this provides an interesting connection to the interaction between a particle and a complex medium.  Indeed, one may interpret the soliton gas (the background into which the trial soliton evolves and interacts) as a prepared complex medium.  As the trial soliton (interpreted as a particle) interacts with this medium, it's own velocity becomes highly oscillatory, providing an effective acceleration. Meanwhile, as the trial soliton passes through the soliton gas (which is also highly oscillatory), it leaves in its wake a phase shift in the medium's oscillations, which we explicitly calculate (see \figurename~\ref{fig:density}, and formula \eqref{shift} below). 

In this light, it is interesting to compare our work to the Toda shock problem as considered by \cite{DeiftVenakidesOba} (motivated by numerical simulations in \cite{HolianFlaschkaMcLaughlin}), and further studied in \cite{BlochKodama1}, \cite{BlochKodama2}, and very recently in \cite{EgorovaMichorTeschl}.  In the Toda shock problem, the complex medium is a quiescent Toda lattice excited by a single particle driven with fixed velocity into the lattice. The dynamics reveals the generation of a two-periodic solution near the driven particle, and an evolving elliptic wave region further into the medium. By contrast, in our work the trial soliton (particle) is not driven with fixed velocity and hence, while it affects the medium, the medium as well affects its evolution.  For the Toda shock problem, this would correspond to giving to one particle a large initial velocity and let it interact with the quiescent medium. 

The interactions between solitons and  a dispersive shock wave or rarefaction wave 
 can be treated via inverse scattering, or also on a more physical approach via  soliton-mean field interactions.
In particular, the interaction between  solitons,  breathers   and  a dispersive shock wave  was considered in the present context of the mKdV equation   in  \cite{GravaMinakov20}  and on a physical level in \cite{vdSEH21}.  Similar settings have been considered at the physical level for the defocusing NLS in \cite{SHE18} and  for the focusing NLS equation in \cite{BNL19} and on a mathematical level in \cite{BM21}. A detailed account of the available literature can be found in \cite{ACEHL22}.

\section{Statement of the results}
In this paper we consider the mKdV equation \eqref{mkdv} for a collection of solutions analogous to those constructed for the KdV equation in \cite{ZakZakDya} and analyzed in \cite{GirGraJenMcL}.  The ``gas" of solitons is produced from a continuum of poles accumulating on the intervals $(-i\eta_{2}, - i\eta_{1}) \cup(i \eta_{1}, i \eta_{2})$, with $0 < \eta_{1} < \eta_{2}$, and with positive norming constants.  
We study the behaviour of this soliton gas while it interacts with a single soliton   with spectral data $\pm i \kappa_{0}$ with $\kappa_{0} > \eta_{2}$, so that the trial soliton's velocity is greater than the velocity of any member of the gas of solitons. Our main results are  the following.  
\begin{itemize}
\item  The  derivation of an expression for the solution of the soliton gas plus the trial soliton in terms of a Fredholm determinant (Theorem~\ref{theorem1} and Corollary~\ref{corollary1.2}).
\item The asymptotic analysis of the behaviour of the solution $q(x,t)$  for large times  (Theorem~\ref{theorem1.3}).
\item The derivation of the dynamical   properties of the trial soliton   (Theorem~\ref{theorem1.4}  or  Theorem~\ref{thm:xpeak}). In particular we  determine  the peak position  $x_{\rm peak}(t)$ of the trial soliton  and we show that its velocity 
$\dot x_{\rm peak}$ has an {\it oscillatory} behaviour   while it interacts with the soliton gas. The   leading order average velocity   $\bar{v}_{\rm sol} (\kappa_0)$  of the soliton peak satisfies the kinetic equation
\begin{eqnarray*}
\bar{v}_{\rm sol} (\kappa_0)=  4 \kappa_0^{2} + \frac{1}{ \kappa_0} \int_{\eta_1}^{\alpha} \ln{
\left|
\frac{\kappa_0-s}{\kappa_0+s}
\right|
}(v_{{\rm group}}(s) - \bar{v}_{\rm sol} (\kappa_0) ) \, \partial_x \rho(is) \,  \d s \ ,
\end{eqnarray*}
where the density $\rho$ is defined  below  in \eqref{rho.intro} and the group velocity $v_{\rm group} := -\tfrac{\rho_t}{\rho_x}$ is defined in \eqref{vgroup}. Here the parameter  $\alpha$ with $\eta_1<\alpha\leq \eta_2$ depends on $x$ and $t$ and it is defined in  equation \eqref{Whitham01} and \eqref{Whitham02}.
\end{itemize}
We note that this last equation is the kinetic velocity equation for the mKdV equation analogous  to the one posited by El and co-authors  in \cite{El03, ElKa05, Tovbis2},  thus showing that the solution \eqref{intro:mKdV_sol_FD}--\eqref{intro:mKdV_sol1_FD} represents indeed a soliton gas.   In this case, the continuity equation \eqref{eq:02} is automatically satisfied for $v = v_{\rm group} =  -\tfrac{\rho_t}{\rho_x}$ and $f = \rho_x$. However, we emphasize that the true peak velocity $\dot x_{\rm peak}$ does \emph{not} satisfy Zakharov--El's kinetic equation, due to the presence of the oscillatory terms.

In order to derive the soliton gas solution,  rather than repeating the RH construction employed in  \cite[Section 2]{GirGraJenMcL}, in which an $N$-soliton solution is characterized by the meromorphic RH problem   and then the singular limit $N \to \infty$ is taken through Riemann--Hilbert gymnastics, we start from the equivalent representation in terms of Fredholm determinants \eqref{intro:mKdV_sol_FD_N}.

We want to consider the limit $N\to\infty$ under the additional assumptions:
\begin{itemize}
\item The poles $\{i \kappa _{j}^{(N)} \}_{j=1}^{N}$ are sampled from a smooth positive density function $\varrho( k )$ so that $\int_{\eta_{1}}^{ \kappa _{j}} \varrho( \eta ) \, \d \eta = j / N $, for $j = 1, \ldots, N$.
\item The coefficients $\{\chi_j\}_{j=1}^N$ are real and positive, such that $N \varrho(\kappa_j) \chi_j$ are discretizations of a given function:
\begin{gather}
N \varrho(\kappa_j) \chi_j = \frac{ r( i \kappa_j)}{2\pi} \qquad j=1,\ldots, N, \label{cjasR0}
\end{gather}
where  $r(k)$ is a real-valued,  continuous,  non-vanishing function of $k$ for $k\in  (i \eta_1, i\eta_2)$.
\end{itemize}

\begin{thm}
\label{theorem1}
Let $\Sigma_1=(i\eta_1,i\eta_2)$, $\eta_2>\eta_1$,  be the interval where the solitons accumulate and let $r(k)$ be a real-valued  positive and  continuous  function on $\Sigma_1$ whose discretization gives the  norming constants $\chi_j$  of the finite $N$-soliton solution according to  \eqref{cjasR0}. 
Let $\mathcal{K}:L^2(\Sigma_1)\to L^2(\Sigma_1)$ be the integral operator $\mathcal{K}[f](k)=\int_{\Sigma_1}K(k,y)f(y)\d y$ with kernel 
\begin{gather}\label{thm2.6_kernel}
{K}(k,y) =\frac{\sqrt{r(k)}e^{-i \, \theta(k;x,t)} \, \sqrt{r(z)} e^{-i \, \theta(z; x,t)}}{2\pi i(k+z)}, \qquad k,z \in \Sigma_1
\end{gather}
where $ \theta(k;x,t) = 4t k^3 + x k$.
Then the   function
\begin{equation}
q(x,t) = i  \dfrac{\partial }{\partial x}\ln\det \le(\Id_{L^2(\Sigma_1)} + {\mathcal K}\ri) - i \dfrac{\partial }{\partial x}\ln\det \le(\Id_{L^2(\Sigma_1)} - {\mathcal K}\ri) ,\label{intro:mKdV_sol_FD}
\end{equation}
is the soliton gas solution of the mKdV equation \eqref{mkdv}. 
\end{thm}

The expression \eqref{intro:mKdV_sol_FD} for the mKdV soliton gas solution is strikingly similar to the Tracy--Widom Fredholm determinant formula \cite{TracyWidom96} for the solution of the integrated version of the \emph{defocusing} mKdV: $q_t - q_{xxx} + 6q^2q_x=0$.  This expression is also similar to the one considered in \cite{BB}  and in \cite{KL}  for solving the weak noise theory of the Kardar--Parisi--Zhang equation and more generally in the study of Fredholm determinants of a class of Hankel composition operators \cite{Bothner2022}.
For example 
\begin{equation}
\le(q(x,t)\ri)^2 = -  \dfrac{\partial }{\partial x}\ln\det \le(\Id_{L^2(\Sigma_1)} -{\mathcal K}^2\ri),
\label{intro:q2}
\end{equation}
 for $t=0$  corresponds to  the square of a Hankel operator.
 
In addition to the gas of solitons, the potentials we will consider in this paper also have an additional soliton that will interact with the gas. The Fredholm determinant derivation is analogous: the spectral parameters $i\kappa_j$'s will accumulate within the interval $\Sigma_1$ except for one point ($i\kappa_0$, with corresponding norming constant $\chi\in\R\backslash\{0\}$), which will lie on the imaginary axis, outside $\Sigma_1$.
\begin{corollary} 
\label{corollary1.2}
The function
\begin{equation}
q(x,t) = i  \dfrac{\partial }{\partial x}\ln\det \le(\Id_{L^2(\mathcal C)} + {\mathcal K}\ri) - i \dfrac{\partial }{\partial x}\ln\det \le(\Id_{L^2(\mathcal C)} - {\mathcal K}\ri) \label{intro:mKdV_sol1_FD}
\end{equation}
is the solution of the mKdV equation representing a soliton gas plus a regular soliton, where $\mathcal C = \Sigma_1 \cup \mathcal C_0$, $\mathcal C_0$ being a small circle around the pole $i\kappa_0$, not intersecting the real line, nor $\Sigma_1$; the operator $\mathcal K$ has kernel
\[
{K}(k,z) =\frac{\sqrt{\widetilde r(k)}e^{-i \, \theta(k;x,t)} \, \sqrt{\widetilde r(z)} e^{-i \, \theta(z; x,t)}}{2\pi i(k+z)}, \qquad k,z \in \mathcal C\ ,
\]
with $\widetilde r(k)$ is defined as $\widetilde r(k) = r(k)$ for $k \in \Sigma_1$ and $\widetilde r(k) = \frac{\chi}{k - i\kappa_0}$ for $k \in \mathcal C_0$ and $\chi\in\R\backslash\{0\}$.
\end{corollary}

From the Fredholm determinant formula, we can derive the following meromorphic RH problem, where in addition to jumps across the intervals $\Sigma_{1}$ and $\Sigma_{2}:= - \overline{\Sigma}_1 = (-i \eta_{2}, -i \eta_{1})$ oriented upwards (the gas), we are accounting for the presence of an extra pair of poles, at $\pm i \kappa_{0}$ (the soliton). 
\begin{RHP}[\bf Soliton + gas]\label{rhp:X}
Let $r(k)$ be a positive real valued  function defined on $\Sigma_1$. 
 Find a $2\times 2$ matrix valued function $\bm{X}(\, \cdot\,; x,t)$ with the following properties: 
	\begin{enumerate}[label=\arabic*.]
		\item $\bm{X}(k;x,t)$ is  analytic for $k \in \C \setminus (\Sigma_1 \cup \Sigma_2)\cup\{\pm i\kappa_0\}.$
		\item $\bm{X}(k;x,t) = \bm I + \bigo{k^{-1}}$ as $k \to \infty$,
		where $ \bm I $ is the $2\times 2$ matrix identity.
		\item For $k \in \Sigma_1 \cup \Sigma_2$, the boundary values $\bm{X}_\pm(k;x,t) = \bm{X}(k \mp 0; x,t)$ satisfy the jump relation
		\begin{gather}\label{Xjumps}
		  \bm{X}_+(k) = \bm{X}_-(k) 
		  \begin{dcases} 
		    \begin{bmatrix} 1 & 0 \\ i r(k) e^{-2i \theta(k;x,t)} & 1 \end{bmatrix}, & k \in \Sigma_1, \\
		    \begin{bmatrix} 1 &  i \,{\overline{r(\bar{k})}} e^{2i \theta(k;x,t)} \\ 0 & 1 \end{bmatrix}, & k \in \Sigma_2, \\
		  \end{dcases}
		  \shortintertext{where}
		  \theta(k;x,t) = 4t k^3 + x k.
		\end{gather}
		\item $\bm{X} (k;x,t)$ has simple poles at $k = \pm i \kappa_0$, with $\kappa_0 > \eta_2$, satisfying
		\begin{equation}\label{Xres}
		  \begin{gathered}
		    \Res_{k = i \kappa_0} \bm{X}(k;x,t) = \lim_{k \to i \kappa_0} \bm{X}(k;x,t) \begin{bmatrix} 0 & 0 \\ -i \chi e^{-2i \theta(k;x,t)} & 0 \end{bmatrix}, \\
		    \Res_{k = -i \kappa_0} \bm{X}(k;x,t) = \lim_{k \to -i \kappa_0} \bm{X}(k;x,t) \begin{bmatrix} 0 & -i \chi e^{2i \theta(k;x,t)} \\ 0 & 0 \end{bmatrix}.
		  \end{gathered}
		\end{equation}
	\end{enumerate}
\end{RHP}
The solution \eqref{intro:mKdV_sol1_FD} to the mKdV equation can be extracted via 
\begin{equation}\label{q_via_X}q(x,t) = 2i \lim_{k \to \infty} k \, \bm{X}(k;x,t)_{12}.
\end{equation}

\begin{prop}
\label{theorem_existence}
Given a  real  function $r \in L^2(\Sigma_1)$, the  Riemann--Hilbert problem \ref{rhp:X} is uniquely solvable for all $(x, t) \in \R^2$.
Moreover, the function $q(x,t)$ defined in \eqref{q_via_X} is a classical solution to the mKdV equation \eqref{mkdv}, which belongs to the class $C^\infty(\R_x \times \R_t)$. 
At time $t=0$ the initial data $q(x,0)$ has the following properties:
\begin{itemize}
\item $q(x,0)=\mathcal{O}(e^{-c|x|})$ as $x\to\, -\infty$ with $c>0$.
\item $q(x,0)= (\eta_2 + \eta_1) \dn \left( (\eta_2+\eta_1)(x-2(\eta_1^2+\eta_2^2)t-x^0),  m_1 \right) +\mathcal{O}(x^{-1}),	$
as $x\to\infty$, where $\dn(z,m_1)$ is the Jacobi elliptic function with modulus $m_1 = \frac{4\eta_2 \eta_1}{(\eta_2+\eta_1)^2}$ and the phase $x^0$ depends explicitly on $r(k)$.
\end{itemize}
\end{prop}

The proof of  the first part of the  above theorem is given in the Appendix~\ref{appB_solvability}. The second part of the theorem about the asymptotic properties of the initial data as $x\to\pm\infty$, follows the steps in \cite{GirGraJenMcL} and for this reason we omit it.  We remark that for $\eta_1\to0$  the initial data  becomes asymptotically step-like since $q(x,0)\to \eta_2+\mathcal{O}(x^{-1})$  as $x\to \infty$.

\begin{figure}
\centering
\hspace*{\stretch{1}}
	\begin{overpic}[scale=0.5]{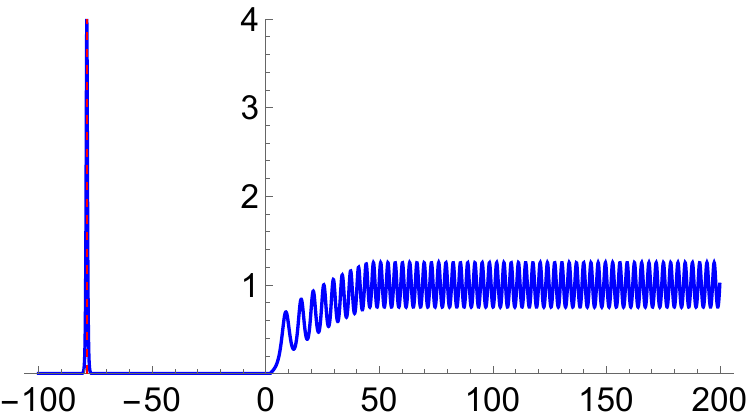}
	   \put(101,4){$\small x$}
	\end{overpic}
\hspace*{\stretch{1}}
	\begin{overpic}[scale=.5]{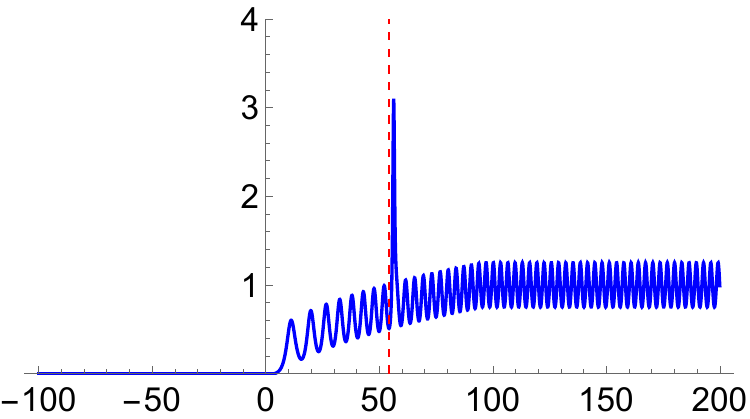}
	   \put(101,4){$\small x$}
	\end{overpic}
\hspace*{\stretch{1}} 
\\[15pt]
\hspace*{\stretch{1}}
	\begin{overpic}[scale=.5]{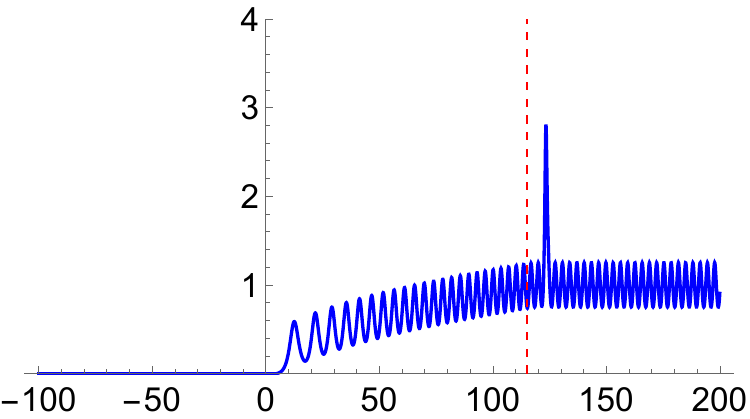}
	   \put(101,4){$\small x$}
	\end{overpic}
\hspace*{\stretch{1}}
	\begin{overpic}[scale=.5]{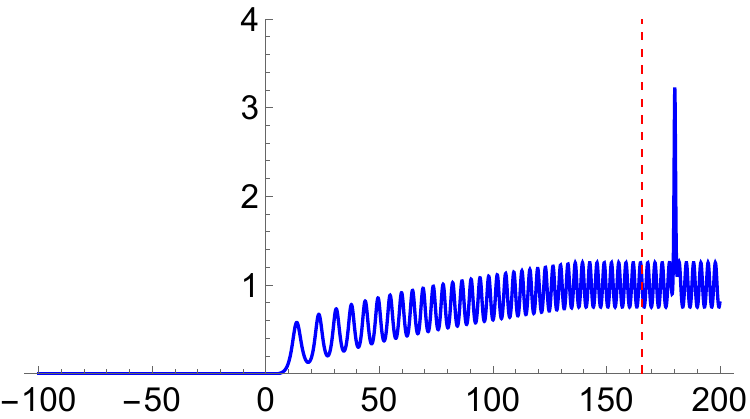}
	   \put(101,4){$\small x$}
	\end{overpic}
\hspace*{\stretch{1}}	
\caption{Four plots showing a soliton interacting with a soliton gas at different times. The blue curve is the  
leading order asymptotic behavior of the potential $q(x,t)$ determined by RH problem \ref{rhp:X}, and the dashed red line is the position of the global maximum of the same soliton if there were no soliton gas to interact with. 
Here $r\equiv 1$, $\eta_{1} = 0.25, \eta_{2}=1, \kappa_{0}=2$, and $\chi=4 e^{-800}$.}
\label{fig:video}
\end{figure}

\begin{figure}
\centering
	\hspace*{\stretch{1}}
	\begin{overpic}[scale=0.4]{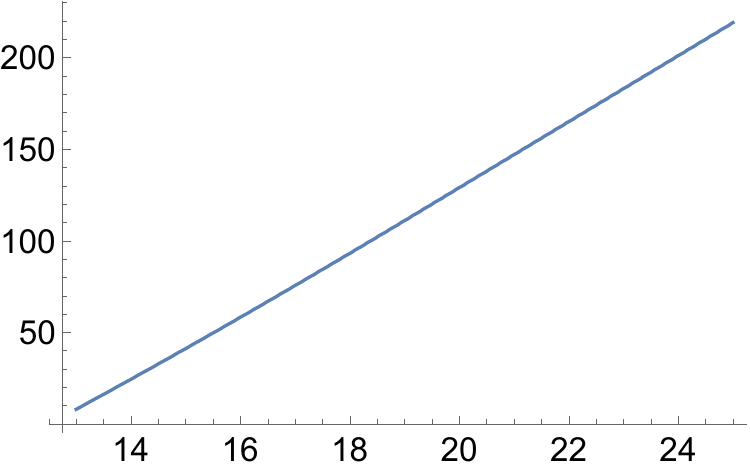}
	   \put(98,2){$\small {}_t$}
	   \put(-3,64){$\small {}^{x_{\mathrm{peak}}}$}
	\end{overpic}
	\hspace*{\stretch{1}}
	\begin{overpic}[scale=.4]{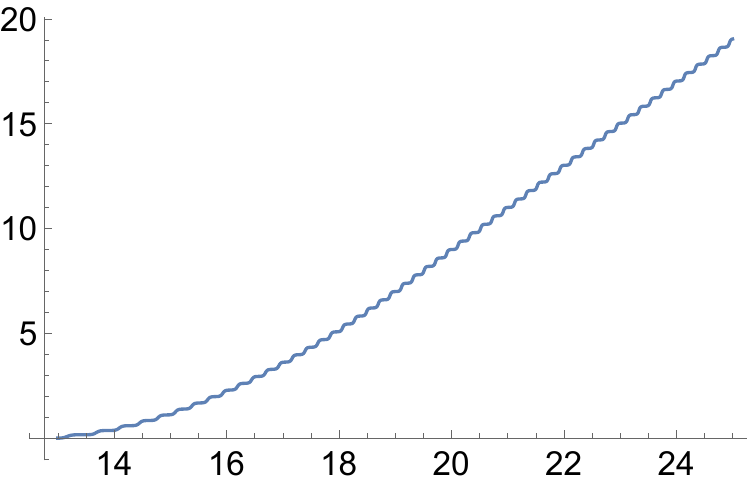}
	   \put(98,2){$\small {}_t$}
	   \put(-5,64){$\small {}^{x_{\mathrm{peak}}-x_{\mathrm{free}}}$}	
	\end{overpic}
	\hspace*{\stretch{1}}
	\begin{overpic}[scale=.4]{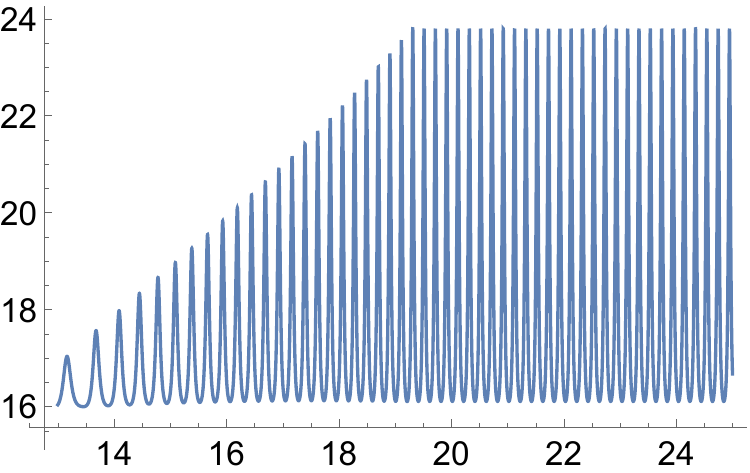}
	    \put(98,2){$\small {}_t$}
	    \put(-2,64){$\small {}^{ \od{x_{\mathrlap{\mathrm{peak}}}}{t} }$}
	\end{overpic}
	\hspace*{\stretch{1}}
\caption{Space-time plots of the position and velocity of the soliton peak.  The first plot shows the numerically computed location of the peak of the leading order asymptotic description.  The middle plot shows the difference $x_{{\rm max}}(t) - x_{{\rm free}}(t)$, where $x_{{\rm free}}$ is the position of the same soliton propagating in the vacuum (i.e. in the absence of the gas).  The third plot shows the velocity of the peak as a function of time.  The effective velocity of the soliton is the asymptotic average of this highly oscillatory velocity profile.  These plots were created in Mathematica using the leading order asymptotic behavior of the potential determined by the RH problem \ref{rhp:X}, with $r\equiv 1$, $\eta_{1} = 0.25, \eta_{2}=1, \kappa_{0}=2$, and $\chi=4 e^{-800}$.}
\label{fig:position_velocity}
\end{figure}
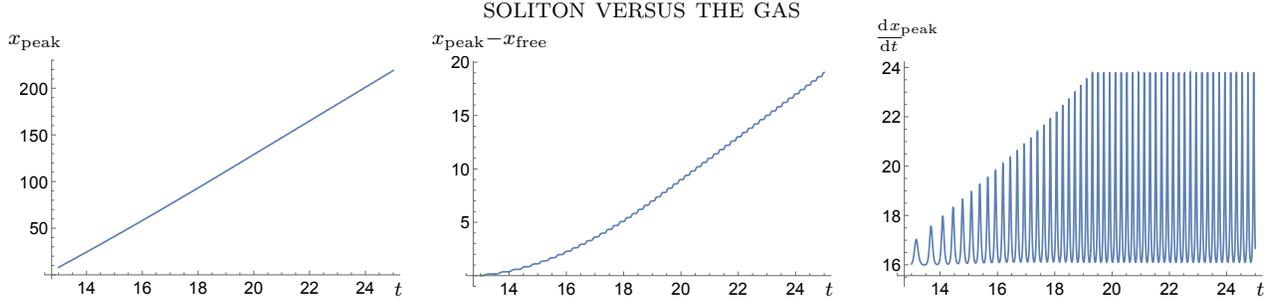

Next we analyse the  asymptotic behaviour of the solution $q(x,t)$  that depends on which direction in the $(x,t)$ plane one looks at: there is a region with exponential decay (when $x<4\eta_1^2 t$), a rarefaction wave region and
a   periodic travelling wave region (when $x>4\eta_1^2 t$)  up to terms of order $\mathcal{O}(1/t)$. To distinguish the latter two regions we need to introduce the speed $v_2$  of the leading front of the rarefaction region and one of the speeds of the   Whitham modulation equations  that describe the rarefraction wave connecting the zero solution to the elliptic solution
\begin{equation}
\label{v2}
v_2=\eta_2^2W\le(\tfrac{\eta^2_1}{\eta_2^2}\ri),\quad W(m) = \frac{4(1-m) K(m)}{E(m)} + 2(1+m),
\end{equation}
where  $K(m)=\int_{0}^{\frac{\pi}{2}}\tfrac{\d \theta}{\sqrt{1-m\sin^2\theta}}$  and $E(m)=\int_0^{\frac{\pi}{2}}\sqrt{1-m\sin^2\theta}\, \d\theta$  are  the complete integral of the first and second kind respectively.

\begin{theorem}
\label{theorem1.3}
Suppose that $r:\Sigma_1 \to \R$ is a continuous, positive function on $\Sigma_1$ with analytic extension to a neighbourhood of $\Sigma_1$ (cf. Assumption~\ref{Assumption:r}). Then the large-time asymptotics of the  soliton gas mKdV solution  with spectrum in the interval $[i \eta_1, i\eta_2] $  and with a larger soliton  with point spectrum $i\kappa_0>i\eta_2$ 
 (RH problem \ref{rhp:X})  is given by
\begin{equation}
\label{eq:qrep}
	q(x,t) =  q_{\mathrm{bg}}(x,t) + q_{\mathrm{sol}}(x,t) + \mathcal{O}(t^{-1})\ , 
\end{equation}
for $\frac{x}{t}>4\eta_1^2$.
The background gas $q_{\rm bg}$ is equal, at leading order,  to the periodic  travelling wave solution of the mKdV equation, namely
\begin{equation}
\label{eq:bgwps}
	q_{\rm bg}(x,t)	= (\alpha + \eta_1) \dn \left( (\alpha+\eta_1)(x-2(\eta_1^2+\alpha^2)t-x^{^{(\pm)}}),  m_1 \right) ,	
\end{equation}
where $\dn(z,m_1)$ is the Jacobi elliptic function with modulus $m_1 = \frac{4\alpha \eta_1}{(\alpha+\eta_1)^2}$, the phase shifts $x^{^{(\pm)}}$ are explicitly determined according to the position of the soliton (see (\ref{x.shifts.background})), and $\alpha = \alpha(x/t)$ is   such that
\begin{itemize}
\item for $ \tfrac{x}{t} \in (  4\eta_1^2,  v_2) $, $\alpha(x/t) \in (\eta_1, \eta_2)$ and satisfies the Whitham modulation equation
\begin{equation}\label{Whitham01}
	\frac{x}{t} = \alpha^2 W(\eta_1^2/\alpha^2),
\end{equation}
with  $v_2$ and  $W(m)$ as in \eqref{v2}. The above equation in uniquely solvable for $\alpha=\alpha(x/t)$;
\item for $\tfrac{x}{t} \geq v_2$, 
\begin{equation}\label{Whitham02}\alpha(x/t) \equiv \eta_2.
\end{equation}
\end{itemize}
The expression of the soliton  $q_{\mathrm{sol}}$ on the elliptic background   is  given in \eqref{q.sol.1}. When $\kappa_0 \gg \eta_2$, $q_{\mathrm{sol}}$ agrees at leading order with the soliton solution \eqref{soliton} on the zero background.
\end{theorem}

A separate fact that emerges from our analysis is that the  phase shifts, $x^{^{(\pm)}}$, in the background wave (\ref{eq:bgwps}) are different if the soliton is behind ($x^{^{(+)}}$) or in front of the wave ($x^{^{(-)}}$) when looking at a fixed direction $x/t$.
The phase shift that the  periodic background  $ q_{\mathrm{bg}}(x,t) $  experiences from the soliton interaction is (see Proposition~\ref{prop51} and Figure~\ref{fig:density})
\begin{equation}
\label{shift}
\begin{aligned}
	&x^{^{(+)}} - x^{^{(-)}} = \frac{2K(m)}{\alpha} \le( 1+   \int_{i\alpha}^{i\kappa_0}  \frac{\alpha}{i K(m)} \frac{\d k}{R(k)} \ri) \ , \\
	& m = \frac{\eta_1^2}{\alpha^2}\ ,\ \quad R(k)=\sqrt{(k^2+\alpha^2)(k^2+\eta_1^2)} \ .
\end{aligned}
\end{equation}

\begin{remark}
Theorem~\ref{theorem1.3} is proven under the assumption that $r(k)$ is positive, bounded, and nonvanishing on $\Sigma_1$, and it admits an analytic extension to a lens-shaped neighborhood of $\Sigma_1$. If instead, one assumes that $r(k) =|z-i\eta_k|^{\pm1/2} \tilde r(k)$, $k=1,2$ with $\tilde r$  positive, bounded and non-vanishing such that $r$ has analytic extension to an open neighborhood of $\Sigma_1$ except for square root branch cuts of the extension (cf. Assumption~\ref{Assumption:r2}) then the error rate for $x/t > v_2$ is much improved
\[	
	q(x,t) =  q_{\mathrm{bg}}(x,t) + q_{\mathrm{sol}}(x,t) + \mathcal{O}( e^{-c t}) 
\]
for some fixed $c>0$. 
\end{remark}

\begin{remark}
The Whitham modulation equations are the modulation equations for  the  wave parameters  $\beta_3>\beta_2>\beta_1$  of the 
elliptic solution of the mKdV equation
\begin{eqnarray*}&&
q_{\rm ell}(x,t)=-\beta_1-\beta_2-\beta_3
+
 \\&&
\qquad \qquad  
+\frac{2(\beta_2+\beta_3)(\beta_1+\beta_3)}{\beta_2+\beta_3-(\beta_2-\beta_1)\mathrm{cn}^2(\sqrt{\beta_3^2-\beta_1^2} (x-2(\beta_1^2+\beta_2^2+\beta_3^2)t)+x_0 \mid m)} \ ,
\end{eqnarray*}
where $\mathrm{cn}(u\mid m)$ is the Jacobi elliptic cosine function. The expression of $q_{\rm ell}(x,t)$  can be reduced to  $q_{\rm bg}(x,t)	$ in \eqref{eq:bgwps} performing the Landen transformation as in \eqref{q.bg1} and \eqref{Landen}.\\
 The Whitham modulation equations  take the form
\begin{equation*}
\dfrac{\partial}{\partial t}\beta_j+W_j(\beta_1,\beta_2,\beta_3)\dfrac{\partial}{\partial x}\beta_j=0,\quad j=1,2,3,
\end{equation*}
where 
the speeds $W_j=W_j(\beta_1,\beta_2,\beta_3)$ are 
\begin{align}
\label{speed}
&W_j(\beta_1,\beta_2,\beta_3)=2(\beta_1^2+\beta_2^2+\beta_3^2)+4\dfrac{\prod_{k\neq j}(\beta_j^2-\beta_k^2)}{\beta_j^2+\gamma}\ ,\\
\label{alphan}
&\gamma=-\beta^2_3+(\beta_3^2-\beta_1^2)\dfrac{E(m)}{K(m)}\ ,\quad m=\dfrac{\beta_2^2 -\beta_1^2}{\beta_3^2-\beta_1^2}\ .
\end{align}
In the case considered in  Theorem~\ref{theorem1.3}, one  has $\beta_1=0$, $\beta_2=\eta_1$ and $\beta_3=\alpha$ and one is looking for a self-similar solution in the form  $\alpha=\alpha(\frac{x}{t})$, which gives the equation $\frac{x}{t}=W_3(0,\eta_1,\alpha)$ that coincides with \eqref{Whitham01}.  
This modulation problem is different from the  Gurevich and Pitaevsky problem \cite{GP74}    or the rarefaction wave problem  \cite{Leach} that connect two constant backgrounds of different amplitude.  It is a generalization of the Riemann problem 
with a  zero background for $x<0$ and an elliptic background for $x>0$.
When  $\eta_1=0$, then  the initial data  is step-like, namely $q(x,0)= \eta_2+\mathcal{O}(x^{-1})$  as $x\to+\infty$ and one recovers the  standard  rarefaction wave $q(x,t)\simeq \sqrt{x/t}$. 
\end{remark}

The first stage of our analysis is quite similar to the asymptotic analysis in \cite{GirGraJenMcL}, relying on the construction of a $g$-function to control the exponentially large off-diagonal factors appearing in the jump relationships for $k \in \Sigma_{1} \cup \Sigma_{2}$.  The function $g$ is determined uniquely by a suitable collection of conditions,  and there are a number of different representations  \cite{G1,GravaTian} of such a function that are useful for different purposes.  
In connection with the kinetic theory, a representation  of the $g$ function  in terms of a logarithmic transform is important.  One represents $g$ in terms of a measure supported on an evolving set $ \Sigma_{1,\alpha} \cup \Sigma_{2,\alpha} := ( i \eta_{1}, i \alpha(x/t)) \cup ( - i \alpha(x/t), -i \eta_{1} ) $ as follows:
\begin{eqnarray}
g(k;x,t) = \int_{\Sigma_{1,\alpha} \cup \Sigma_{2,\alpha}} \ln{(k-s)} \rho_+(s;x,t) \, \d s \ .
\end{eqnarray}
The measure $\rho(s) \, \d s$ is given explicitly by
\begin{eqnarray} \label{rho.intro}
\rho (s;x,t)=-\frac{1}{\pi i} \left\{
12 t \frac{ s^4 + \frac{1}{2}( \eta_{1}^{2} + \alpha^{2})s^2 + c_{2}}{R(s)} + x \frac{ s^{2} + c_{0} }{R(s)}
\right\} \ , 
\end{eqnarray}
where the quantity $R(s)=\sqrt{(s^2+\alpha^2)(s^2+\eta_1^2)} $  is analytic in $\mathbb{C} \setminus \{\Sigma_{1,\alpha} \cup \Sigma_{2,\alpha}\}$ and the constants $c_{0}$ and $c_{2}$ depend on $x$ and $t$, and are uniquely determined by
\[
\int_{-i\eta_1}^{i\eta_1}\rho (s;x,t)\d s=0.
\]
In the definition of $g(k;x,t)$, the quantity $\rho_+(s;x,t)$  refers to the  left  boundary value $R_+(s)$ of the function $R(s)$ on the oriented contour $\Sigma_{1,\alpha} \cup \Sigma_{2,\alpha}$.

The quantity $q_{\mathrm{sol}}(x,t)$ in (\ref{eq:qrep}) above represents the contribution to the solution $q(x,t)$ from the poles in the Riemann-Hilbert problem \ref{rhp:X} (the soliton component).  To describe the dynamical evolution of the trial soliton and its interaction with the gas, we take the poles of the soliton to be located at $\pm i \kappa_{0}$ with $\kappa_{0} > \eta_{2}$, so that the trial soliton's velocity is greater than the velocity of any member of the gas of solitons. Additionally, we initiate the trial soliton's location $x_{0}$ to be in the quiescent region of the soliton gas solution, separated from the modulated cnoidal wave region by a large distance, so that $x_{0}$ is  essentially the asymptotic parameter.  This scenario can be visualized in Figures~\ref{fig:video} and \ref{fig:density}.

{
Within this setting, at large times, we are able to compute a collection of quantities to better understand the dynamics. From the $g$-function and the corresponding \emph{wave phase} $\varphi (k;x,t) := g(k;x,t) + kx + 4k^3t$ (see formula \eqref{phi.def}), we introduce 
\[ \kappa(k) := \frac{\partial \varphi (k)}{\partial x} \quad \text{(wave number)}\ , \qquad
 \omega(k) := -\frac{\partial \varphi(k)}{\partial t} \quad \text{(frequency)} \ , \]
 following an analogy with the classical theory of the wave equation solved with the help of the Fourier transform.
We then define the \emph{phase velocity} of the wave
\begin{gather}
v_{\rm phase}(k) := \frac{\omega(k)}{\kappa(k)} = - \frac{\varphi_t(k)}{\varphi_x(k)}\ ,
\end{gather}
which will allow us to calculate the velocity of the background elliptic wave $q_{\rm bg}$ (i.e. the velocity of the envelope)
\begin{gather}
v_{\rm bg} = v_{\rm phase}(i \eta_1) = 2(\eta_1^2+\alpha^2)  \ , 
\end{gather}
and we define the \emph{ group velocity} of the wave packet as
\begin{gather}
\label{vgroup}
v_{\rm group}(k) = - \frac{\rho_t(k)}{\rho_x(k)}  = -12 \frac{k^4 + \frac{1}{2}(\eta_1^2 +\alpha^2)k^2 + c_2}{k^2 + c_0}\ .
\end{gather}

We consider the region of space-time such that
\begin{equation}\label{into_modulated_elliptic_region0}
\begin{aligned}
&4 \eta_{1}^{2} + \varepsilon < \frac{x}{t} < v_{2} - \varepsilon, \\
& 4 \eta_{1}^{2} + \varepsilon < \frac{x_{0}}{t} + 4 \kappa_{0}^{2} ,
\end{aligned}
\end{equation}
for some $\varepsilon >0$, where we recall that $v_2$ is the   velocity  of the leading front of the rarefaction region  defined in \eqref{v2} and $x_0$ is the phase of the trial soliton.
As discussed in detail in Section \ref{sec-gas}, the first condition defines the region of space-time in which the gas behaves as a modulated elliptic wave, while the second ensures that enough time has passed so that the soliton initially at position $x_0\ll -1$ has traversed the quiescent region and has entered the modulated wave region of the soliton gas.  Then we identify two times $t_1$ and $t_2$ characterized by the soliton peak entering and leaving the modulated wave region respectively.  The first time, $t_{1}$, is  explicit:
\begin{equation}
\label{t1}
 t_1 := \frac{-x_0}{4(\kappa_0^2-\eta_1^2)},
\end{equation}
while $t_{2}$ is quite implicit (see Section \ref{subsec:soliton.peak}).

In Theorem \ref{thm:xpeak} we establish a number of dynamical properties regarding the location $x_{\rm peak}(t)$ of the peak of the soliton, which we summarize here.

In Theorem \ref{theorem1.4} we establish a number of dynamical properties regarding the location $x_{\rm peak}(t)$ of the peak of the soliton. 
\begin{theorem}[\bf Dynamical properties of the soliton peak]
\label{theorem1.4}
For $\chi>0$ and  $x_0 \ll -1$, there exists a  $\tilde{\kappa}>  \kappa_{\rm crit}$ (see \eqref{kappa.c.Q-1}) , such that for all $\kappa_0 > \tilde{\kappa}$ there exists a unique global maximum $x_\mathrm{peak}(t)$ of the  mKdV  solution $q(x,t)$ which identifies the position of the soliton peak for all $t>0$.

Moreover, $x_{\rm peak}(t)$ is strictly increasing, and satisfies (for some small positive $\varepsilon$): 
\begin{itemize}
\item[(i)] for $t\in (0, t_1(1-\varepsilon))$ with $t_1$ as in \eqref{t1},  $x_{\rm peak}(t) = x_0 + 4\kappa_0^2 t$; 
\item[(ii)] for $t >t_1(1 + \varepsilon)$,
\begin{equation}
\dot x_{\rm peak}(t) =  - \frac{2\varphi_{t}(i \kappa_{0}) - \partial_t \ln\Psi(x,t; \kappa_0, \eta_1) } { 2\varphi_{x}(i \kappa_{0})- \partial_x \ln\Psi(x,t; \kappa_0, \eta_1) }\Big\vert_{x = x_{\rm peak}(t)} + \bigo{ t^{-1}} \ ,
\end{equation}
 where $\Psi$  is defined by \eqref{Psi.speed}, and is an oscillatory function.
\end{itemize}

For $(x,t)$ such that $x>4\eta_1^2t$ and $t >t_1(1 + \varepsilon)$,  the average velocity of the trial soliton's position $x_{\rm peak}(t)$ as it traverses one period $T$ of the background oscillatory wave   is 
\[
\frac{ x_{\rm peak}(t+T) - x_{\rm peak}(t) }{T}= \bar{v}_{\rm sol} (\kappa_0) + \bigo{t^{-1}}\ ,
\]
where 
\begin{gather}
\bar{v}_{\rm sol} (\kappa_0)= -\frac{ \varphi_{t}(i\kappa_0)}{\varphi_{x}(i\kappa_0)} = v_{\rm phase}(i \kappa_0) =  4\kappa_0^2 \frac{K\le(\tfrac{\eta_1^2}{\alpha^2}\ri)}{\Pi \le(\tfrac{\eta_1^2}{\kappa_0^2}, \tfrac{\eta_1^2}{\alpha^2}\ri)}  + 2(\eta_1^2+\alpha^2)  
\end{gather}
  where  $\Pi(n,m) = \int_0^{\frac{\pi}{2}} \frac{\d \theta}{(1-n\sin^2\theta)\sqrt{1-m\sin^2\theta}}$ is the complete elliptic integrals of third kind.

Moreover, this average soliton velocity $\bar{v}_{\rm sol} (\kappa_0)$ satisfies  the integral equation
\[
\bar{v}_{\rm sol} (\kappa_0)=  4 \kappa_0^{2} + \frac{1}{ \kappa_0} \int_{\eta_1}^{\alpha} \ln{
\left|
\frac{\kappa_0-s}{\kappa_0+s}
\right|
}(v_{{\rm group}}(s) - \bar{v}_{\rm sol} (\kappa_0) ) \, \partial_x \rho(is) \,  \d s \ .
\]
where $v_{{\rm group}}$ is defined in \eqref{vgroup} and $\rho$ is the density in \eqref{rho.intro}.
\end{theorem}

In Figure~\ref{vsolVSvpeak} the velocity $\bar{v}_{\rm sol} (\kappa_0)$ and the  actual peak velocity  $\dot x_{\rm peak} $ of the soliton  are plotted. 
 One consequence of our analysis is to clarify that the solution of the Zakharov--El   kinetic equations  represents the leading order average velocity of the soliton peak.

\begin{figure}
\centering
\begin{tikzpicture}[
spy using outlines={magnification=7, rectangle, width=4.5cm, height=2.5cm, red, connect spies}
]
\node (n1) at (0,0) {
\begin{overpic}[scale=0.5]{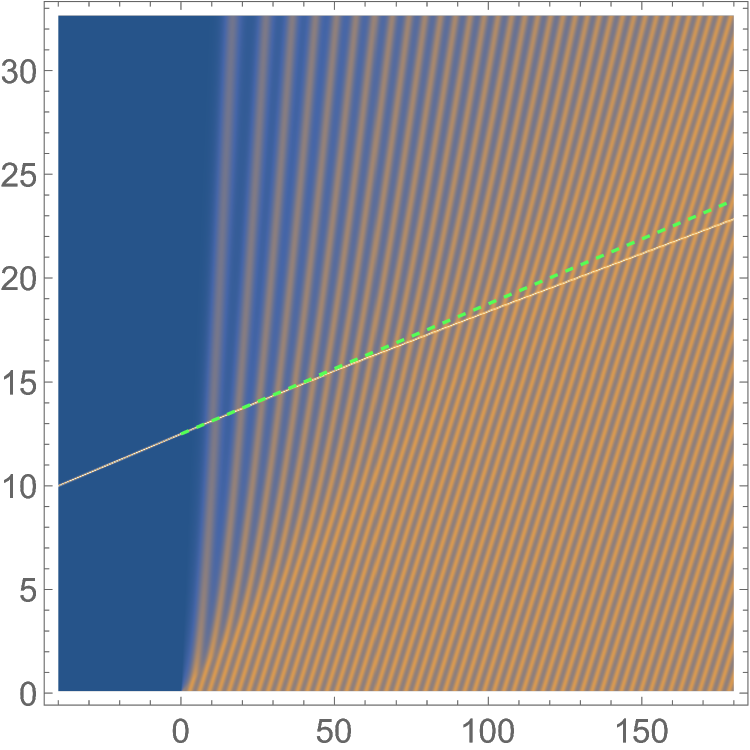}
	\put(2,99) {$\small t$ }
	\put(99,2) {$\small x$ }
\end{overpic}
};
\node (n1) at (-4,0) {
\begin{overpic}[scale=0.5]{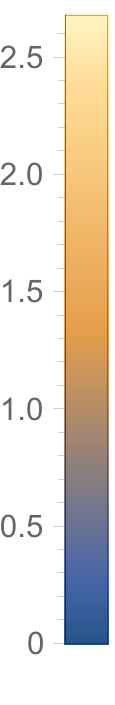}
\end{overpic}
};
\spy on (1,0.55) in node at (7,0.5);
\end{tikzpicture}
\caption{Color map plotting the leading order behavior of a trial soliton traveling through a soliton gas. Color indicates the amplitude of the solution $q(x,t)$ at a given point $(x,t)$. The soliton is accelerated by interaction; the dashed green line shows the position the trial soliton in a vacuum. The zoomed in region is included to show the affect of the soliton on the gas. Interaction with the trial soliton induces a phase shift in the soliton gas. This shift is given by \eqref{shift}. For this plot $r\equiv 1$, $\eta_1 = 0.25$, $\eta_2 = 1$, $\kappa_0 = 2$, and $\chi = 4e^{-800}$.  
}
\label{fig:density}
\end{figure}

\vskip .2in

\begin{figure}
\centering
\begin{overpic}[scale=.6]{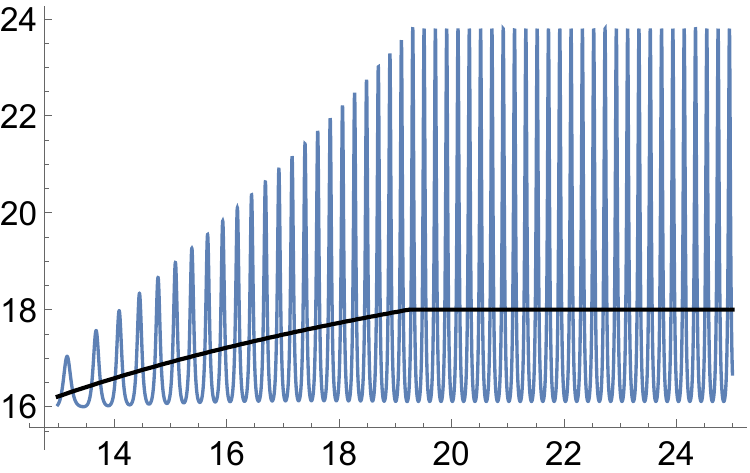}
	\put(99,2) {$\small t$ }
\end{overpic}
\caption{Comparison between $\bar{v}_{\rm sol}(t)$ (black curve) and $\dot x_{\rm peak}(t)$ (blue) as the soliton traverses the soliton gas.  The parameters are with $r\equiv 1$, $\eta_{1} = 0.25, \eta_{2}=1, \kappa_{0}=2$, and $\chi=4e^{-800}$. \label{vsolVSvpeak}}
\end{figure}

\vskip .2in

\paragraph{\bf Outline of the paper.}
In Section \ref{sec-fredholm} we present the soliton gas solution in terms of the Fredholm determinant, also in the presence of a trial soliton, and derive its corresponding RH problem and we prove Theorem~\ref{theorem1}
The large $t$ and $x$ asymptotic analysis of the RH problem~\ref{rhp:X}  is carried out in the subsequent sections: in Section~\ref{sec-setup} we introduce the $g$-function and perform the necessary preparations in order to perform a steepest descent analysis, which is then presented in Section~\ref{sec-model} with the introduction of the model problem and the proof of Theorem~\ref{theorem1.3}. 
The study of the interaction dynamics is conducted in Section~\ref{sec-gas} where we prove Theorem~\ref{theorem1.4}.
The solvability of the  RH problem~\ref{rhp:X} is proved in Appendix \ref{appB_solvability}, thus guaranteeing that the Fredholm determinant representation of the solution is well defined.   In Appendix \ref{app:KdV_Fredholm} we provide the analogous formula for the KdV Fredholm determinant solution of the soliton gas.

}

\section{Soliton gas limit using Fredholm determinants  and proof of Theorem~\ref{theorem1}}
\label{sec-fredholm}

In this section we will derive  the (free) soliton gas solution as a continuum limit of a finite number of solitons, in terms of Fredholm determinant,  thus proving Theorem~\ref{theorem1}. We will also briefly explain how to extend the procedure to a more general setting where a gas and a finite number of solitons are simultaneously considered
so that we prove Corollary~\ref{corollary1.2}.  Then we explain the derivation of the Riemann--Hilbert problem 1 from the Fredholm determinant solution. 

\vskip .2in

\paragraph{ \bf The mKdV $N$-soliton solution.}
We recall here the RH problem for an mKdV $N$-soliton solution:  find a $2\times 2$ matrix  $\bm M(k;x,t)$ such that
\begin{enumerate}[label=\arabic*.]
\item $\bm M$ is meromorphic in $\mathbb{C}$, with simple poles at $\{ i\kappa_j \}_{j=1}^{N}$ in $i\mathbb{R}_{+}$, and at the corresponding conjugate points $\{ -i \kappa_j \}_{j=1}^{N}$ in $i\mathbb{R}_{-}$;
\item $\bm M$ satisfies the residue conditions
\begin{equation}\label{rhp:M}
\begin{gathered}
\Res_{k= i\kappa _{j}} \bm M(k) = \lim_{k \to  i\kappa _{j}} \bm M(k) \begin{bmatrix} 0 & 0 \\ \displaystyle {-i\chi_j } e^{-2 i \theta(k;x,t) } & 0 \end{bmatrix} \\ \Res_{k=-i\kappa_{j}} \bm M(k) = \lim_{k \to -i\kappa_{j} } \bm M(k) \begin{bmatrix} 0 & \displaystyle {-i\chi_j}  e^{2 i \theta(k;x,t) }\\ 0 & 0\end{bmatrix} \ ,
\end{gathered}
\end{equation}
where $\theta(k,x,t) = 4tk^3 + xk$ and $\chi_j$ are  nonzero real constants;
\item $\displaystyle \bm M(k) = \bm I + \mathcal{O}\le(\frac{1}{k}\ri)$ as $k \rightarrow \infty$.
\end{enumerate}
The $N$-soliton potential $q_N(x,t)$ is determined from $\bm M$ via
\begin{equation}
\label{sol01}
q_N(x,t)=2i\lim\limits_{k\to\infty}k \bm M(k;x,t)_{12},
\end{equation}
or alternatively
\begin{equation}
\label{sol02}
\le( q_N(x,t) \ri)^2=2i \, \partial_x\le (\lim_{k\to\infty}k (\bm M(k;x,t)_{11}-1)\ri).
\end{equation}
We are looking for a solution in the form
\begin{gather}\label{Xsol}
\bm M(k;x,t)=
\begin{pmatrix}
 1+\sum_{\ell=1}^N\frac{i \alpha_\ell(x,t)}{k-i \kappa_\ell} & \sum_{\ell=1}^N\frac{i \beta_\ell(x,t)}{k+i \kappa_\ell}
\\
\sum_{\ell=1}^N\frac{i \beta_\ell(x,t)}{k-i \kappa_\ell} & 1-\sum_{\ell=1}^N\frac{i \alpha_\ell(x,t)}{k+i \kappa_\ell}
\end{pmatrix},
\end{gather}
that respects the symmetry
\[
\bm M(k)=\overline{\bm M(-\overline{k})}=\begin{pmatrix}
0&-1\\1&0\end{pmatrix}\bm M(-k)\begin{pmatrix}
0&1\\-1&0\end{pmatrix}\ .
\]
Plugging the above ansatz into the residue conditions gives the system of equations
\begin{equation}
\label{linear_eq1}
\begin{pmatrix}
\bm I_N  & \rvline & \bm A\\
\hline
 -\bm A & \rvline &
\bm I_N
\end{pmatrix} 
\begin{bmatrix}
\tilde{\boldsymbol{\alpha}}\\
\tilde{\boldsymbol{\beta}}
\end{bmatrix}
=
\begin{bmatrix}
\boldsymbol{0}\\
-\tilde{\boldsymbol{\chi}}
\end{bmatrix}
\end{equation}
where  $\bm I_N$ is the $N\times N$ dimensional matrix and 
\begin{equation}
\label{def_A}
\bm A_{j\ell}=\sgn(\chi_j)\frac{\sqrt{|\chi_j|}\sqrt{|\chi_\ell|}}{\kappa_j+\kappa_\ell}e^{-i (\theta_j(x,t)+\theta_\ell(x,t))},\quad \tilde{\boldsymbol{\alpha}}_j=\frac{\alpha_j}{\sqrt{|\chi_j|}}e^{i\, \theta_j(x,t)}, \quad 
\end{equation}
and 
\[
\tilde{\boldsymbol{\beta}}_j=\frac{\beta_j}{\sqrt{|\chi_j|}}e^{i \,\theta_j(x,t)}, \quad \tilde{{\bm \chi}}_j=\sgn(\chi_j)\sqrt{|\chi_j|} e^{-i\, \theta_j (x,t)}\,,
\]
with $\theta_j(x,t) := \theta(i \kappa_j; x,t)$, $\forall \ j=1,\ldots, N$.

\begin{note}\label{solVSantisol}
Notice that  $\sgn(\chi_j)<0$  for the  anti-soliton while   $\sgn(\chi_j)>0$ for the  soliton.
 The case where there are only anti-solitons can be recovered  from the case with only solitons
  by sending  the matrix $\bm{A}\to-\bm{A}$ and $\boldsymbol{\chi}\to -\boldsymbol{\chi}$.
\end{note}
From the solvability of the RH problem for $\bm M$ (see \cite{Wadati72}) we can conclude that the matrix ${\bm I}_N+\bm A^2$ is invertible, and we can recover the solution of mKdV by solving the above linear system of equations \eqref{linear_eq1}:
\[
\begin{bmatrix}
\tilde{\boldsymbol{\alpha}}\\
\tilde{\boldsymbol{\beta}}
\end{bmatrix}
=
\begin{pmatrix}
\bm I_N -\bm A(\bm I_N + \bm A^2)^{-1} \bm A  & \rvline &-\bm A(\bm I_N+\bm A^2)^{-1}\\
\hline
 (\bm I_N+\bm A^2)^{-1}\bm A & \rvline &
 (\bm I_N+\bm A^2)^{-1}
\end{pmatrix} 
\begin{bmatrix}
\boldsymbol{0}\\
-\tilde{\boldsymbol{\chi}}
\end{bmatrix}
\]
that implies
\begin{align*}
\alpha_j&= \sum_{\ell=1}^N(\bm A (\bm I_N+\bm A^2)^{-1})_{j\ell }\sgn(\chi_\ell)\sqrt{|\chi_\ell|} \sqrt{|\chi_j|}e^{-i(\theta_\ell(x,t)+\theta_j(x,t)) },\\
 \beta_j&=-\sum_{\ell=1}^N(\bm I_N+\bm A^2)^{-1}_{j\ell }\sgn(\chi_\ell)\sqrt{|\chi_\ell|} \sqrt{|\chi_j|}e^{-i(\theta_\ell(x,t)+\theta_j(x,t)) }.
\end{align*}

From the expression of $\bm M(k;x,t)$ in formula \eqref{sol02} we have
\begin{equation}
\label{Sol}
q_N(x,t)=-2\sum_{j=1}^N \beta_j(x,t),\qquad \le( q_N(x,t)\ri) ^2=-2\, \partial_x\le(\sum_{j=1}^N\alpha_j(x,t)\ri).
\end{equation}

When there are only solitons,  namely $\chi_j>0$ for $j=1,\dots, N$  we can write the solution \eqref{Sol} in terms of logarithmic derivatives of matrix determinants.
\begin{prop}
The $N$-soliton solution of the mKdV equation takes the form
\begin{equation}
\label{sol1}
q_N(x,t)=i \dfrac{\partial }{\partial x}\ln\det \le(\bm I_N -i \bm A\ri)-i \dfrac{\partial }{\partial x}\ln\det \le(\bm I_N +i \bm A\ri)
\end{equation}
or alternatively
\begin{equation}
\label{sol2}
\le( q_N(x,t) \ri)^2=- \dfrac{\partial^2}{\partial x^2}\ln\det \le(\bm I_N +\bm A^2\ri),
\end{equation}
where the matrix $\bm A$ is defined in \eqref{def_A} with $\chi_j>0$ for $j=1,\dots, N$.
\end{prop}
\begin{proof}

Using the relations
\[
 \dfrac{\partial }{\partial x} \bm A_{\ell  j}= \sqrt{\chi_\ell} \sqrt{\chi_j}e^{-i(\theta_\ell(x,t)+\theta_j(x,t)) } \ , \qquad
   \dfrac{\partial }{\partial x} \bm A^2 =   \bm A \le( \dfrac{\partial }{\partial x} \bm A\ri) +\le(\frac{\partial }{\partial x} \bm A \ri) \bm A\ ,
\]
we can write $\sum_{j=1}^N\alpha_j$ in the form
\[
\sum_{j=1}^N\alpha_j=\mbox{Tr}\left(\bm A (\bm I_N+\bm A^2)^{-1}\dfrac{\partial }{\partial x}\bm A\right)=\frac{1}{2}\dfrac{\partial }{\partial x}\ln\det \le(\bm I_N+\bm A^2 \ri)\ ,
\]
where we used the identity $\dfrac{\partial }{\partial x}\ln\det \bm M=\mbox{Tr}\left( \bm M^{-1}\dfrac{\partial }{\partial x}\bm M\right)$, valid for any invertible matrix $\bm M$ and the fact that 
$\bm A$ is symmetric.  Therefore, from \eqref{Sol} we obtain
\begin{equation}
\label{q2}
(q_N(x,t))^2=- \dfrac{\partial^2}{\partial x^2}\ln\det \le(\bm I_N +\bm A^2\ri) \ .
\end{equation}
Regarding the expression for $q_N(x,t)$,  we have
\[
q_N(x,t)=2\, \mbox{Tr}\left((\bm I_N +\bm A^2)^{-1} \dfrac{\partial }{\partial x} \bm A \right)
\]
and noticing that
\[
(\bm I_N +\bm A^2)^{-1}=(\bm I_N  +i \bm A)^{-1}(\bm I_N -i \bm A)^{-1}=(\bm I_N +i \bm A)^{-1}+i (\bm I_N +i \bm A)^{-1}(\bm I_N -i \bm A)^{-1}\bm A\ ,
\]
where we used the identity $(\bm I_N -i \bm A)^{-1}=\bm I_N +i (\bm I_N -i \bm A)^{-1}\bm A$, we can obtain
\begin{eqnarray*}&&
\mbox{Tr}\left((\bm I_N +\bm A^2)^{-1}\dfrac{\partial }{\partial x}\bm A\right)=\mbox{Tr}\left((\bm I_N +i \bm A)^{-1}\dfrac{\partial }{\partial x}\bm A\right)
+
\\&&
\qquad \qquad \qquad 
+i \mbox{Tr}\left((\bm I_N +i \bm A)^{-1}(\bm I_N -i \bm A)^{-1}\bm A\dfrac{\partial }{\partial x}\bm A\right) \ .
\end{eqnarray*}
In conclusion, we can write $q_N(x,t)$ in the following way
\begin{align*}
q_N(x,t)&=-2i \dfrac{\partial }{\partial x}\ln\det \le(\bm I_N +i \bm A\ri)+i \dfrac{\partial }{\partial x}\ln\det\le( \bm I_N +\bm A^2\ri) \\
&=i \dfrac{\partial }{\partial x}\ln\det \le(\bm I_N -i \bm A\ri)-i \dfrac{\partial }{\partial x}\ln\det \le(\bm I_N +i \bm A\ri).
\end{align*}
\end{proof}
We can now recast the matrix $-i \bm A:\C^N\to \C^N$  as the composition of two operators
$\mathcal{A}_N:L^2(-\infty,x )\to\C^N$ and $\mathcal{B}_N:\C^N\to L^2(-\infty, x)$, where
\[
\big(\mathcal{A}_N[f]\big)_j=\int^x_{-\infty}\sqrt{\chi_j} e^{-i\, \theta_j(s,{\color{black}t}) }f(s)\, \d s,\quad {\mathcal B}_N[v](s,t)=-i \sum_{j=1}^N \sqrt{\chi_j} v_je^{-i\, \theta_j(s,{\color{black}t}) }
\]
so that
\[
( {\mathcal A}_N\circ{ \mathcal B}_N)_{j\ell}= \frac{\sqrt{\chi_j} \sqrt{\chi_\ell}e^{-i(\theta_j(x,t)+\theta_\ell(x,t))}}{i (\kappa_j+\kappa_\ell)}\ .
\]
Using the identity $\det \le(\bm I_N  \pm \mathcal A_N \circ \mathcal B_N\ri) = \det \le(\Id_{L^2(-\infty,x)} \pm \mathcal B_N \circ \mathcal A_N\ri)$, we obtain that the pure $N$-soliton solution of the mKdV equation is equal to
\begin{equation}
\label{q_fredh1}
q_N(x,t)= i \dfrac{\partial }{\partial x}\ln\det \le(\Id_{L^2(-\infty,x)} + \mathcal K_N\ri) - i \dfrac{\partial }{\partial x}\ln\det \le(\Id_{L^2(-\infty,x)} - \mathcal K_N\ri)
\end{equation}
where $  \mathcal K_N:={\mathcal B}_N\circ{ \mathcal A}_N$ is an integral operator 
\begin{equation}
\label{q_fredh2}
{\mathcal K}_N[f](y,t)=\int_{-\infty}^x F_N(y+s,{\color{black}2t})f(s) \, \d s, \quad \text{with kernel }
F_N(s,t)=-i \sum_{j=1}^N \chi_j e^{-i\, \theta_j(s,t)}.
\end{equation}

\vskip .2in

\paragraph{\bf The infinite soliton limit with positive $\chi_j$'s} \label{infinitesolitonlimit}
We want to consider the limit $N\to\infty$ under the additional assumptions:
\begin{itemize}
\item The poles $\{i \kappa _{j}^{(N)} \}_{j=1}^{N}$ are sampled from a smooth positive density function $\varrho( k )$ so that $\int_{\eta_{1}}^{\kappa _{j}} \varrho( \eta ) \, \d \eta = j / N $, for $j = 1, \ldots, N$.
\item The coefficients $\{\chi_j\}_{j=1}^N$ are real and positive, such that $N \varrho(\kappa_j) \chi_j$ are discretizations of a given function:
\begin{gather}
N \varrho(\kappa_j) \chi_j = \frac{ r( i \kappa _j)}{2\pi}  \qquad j=1,\ldots, N, \label{cjasR1}
\end{gather}
where  $r(k)$ is a real-valued,  continuous, non-vanishing function of $k$ for $k\in  \Sigma_1 = (i \eta_1, i\eta_2)$.
\end{itemize}

\begin{prop}
The following limit holds uniformly for $(s,t)$ in compact subsets of $ \R\times \R_+$:
\begin{eqnarray*}
\lim_{N\to +\infty} \sum_{j=1}^{N}\chi_j e^{-i \, \theta_j(s,{2t})} 
&=&\lim_{\mathclap{N\to +\infty}}  \ \ \frac{1}{N} \sum_{j=1}^{N} r( i \kappa_j) e^{-i \, \theta(i\kappa_j; s,{2t})} \frac{ \kappa_j -\kappa_{j-1}}{2\pi}\\
&=& \int_{\Sigma_1} r(\zeta)e^{-i \, \theta(\zeta; s,{2t})} \frac{\d\zeta}{2\pi i } \ .
\end{eqnarray*}
\end{prop}
\begin{proof}
Using the relation \eqref{cjasR1}, the claim follows easily from the convergence of the Riemann sums to the Riemann integral.
\end{proof}

\begin{remark}
The anti-soliton gas case can be analyzed following the same arguments as above (up to a sign change, see Note \ref{solVSantisol}). The mixed (or bipolar) case is more involved and we do not consider it here. 
\end{remark}

Therefore, in the limit we obtain an integral operator $ \mathcal K:L^2(-\infty,x )\to L^2(-\infty,x )$
\begin{gather}
\mathcal  K [g](y,t)=\int_{-\infty}^x F(y+s,{\color{black}2t})g(s)\, \d s, \quad
\text{with kernel} \quad
F(s,t):=-i \int_{\Sigma_1} r(\zeta)e^{-i \, \theta(\zeta; s,t)} \frac{\d\zeta}{2\pi i}\ . \label{Fc}
\end{gather}

\begin{prop}
The finite-$N$ Fredholm determinant (and its derivatives  w.r.t. $x$) converge in the limit:
$$\det \le(\Id_{L^2(-\infty,x)} \pm \mathcal K_N \ri) \to \det \le(\Id_{L^2(-\infty,x)} \pm \mathcal K \ri), \qquad \text{as } N\to \infty.$$
\end{prop}
\begin{proof}
Note that both the finite kernel $\mathcal K_N$ and the infinite kernel $\mathcal K$ are trace class, as they are products of two Hilbert-Schmidt operators, and that the finite kernel $F_N$ (and its $x$-derivative)  converges uniformly to the limiting kernel $F$. Then standard results on the convergence of operators in trace-class norm and continuity of the Fredholm determinant with respect to the trace-class norm topology (see \cite[Chapter 2]{Simon_trace}) imply the convergence of the respective Fredholm determinants as $N\to \infty$.  
\end{proof}

As in the $N$-soliton case, we notice that $\mathcal K$ can be written as composition of two operators $\mathcal K = \mathcal B \circ \mathcal A$, where $\mathcal A:L^2(-\infty,x)\to L^2(\Sigma_1)$ and $\mathcal B: L^2(\Sigma_1)\to L^2(-\infty,x)$,
\[
\mathcal A [g](\zeta)=\sqrt{r(\zeta)}\int_{-\infty}^xe^{-i \, \theta(\zeta; s,{\color{black}t})} g(s)\, \d s,\quad \mathcal B [f](x,t)=-\int_{\Sigma_1}\sqrt{r(k)}e^{-i \, \theta(k; x,{\color{black}t})}f(k)\dfrac{\d k}{2\pi} \ .
\]
Via the same identity $\det\le( \Id_{L^2(-\infty,x)} \pm \mathcal B \circ \mathcal A\ri)=\det \le(\Id_{L^2(\Sigma_1)} \pm \mathcal A \circ \mathcal B\ri)$, 
the  infinite-soliton solution can be written as in Theorem~\ref{theorem1} (with abuse of notation we denoted by $\mathcal K$ also the resulting operator $\mathcal A \circ \mathcal B$).

\begin{remark}
It is evident that the operator defined in \eqref{def_A} is the discretized version of \eqref{thm2.6_kernel}. However, convergence as $N\to \infty$ is not so straightforward, as the two operators are acting on different Hilbert spaces; we bypassed this difficulty by shifting the setting into integral operators acting on the same Hilbert space $L^2(-\infty,x)$.
\end{remark}

It remains to show that such expression still satisfies the mKdV equation.  For the purpose we use the  Riemann-Hilbert formulation of the problem.
We recognize in \eqref{thm2.6_kernel}--\eqref{intro:mKdV_sol_FD} the same setting as the one thoroughly analysed in \cite{BertolaCafasso11} and we will closely follow their arguments. 
It is possible to associate to the operators $\mathcal K$ and $\mathcal K^2$ the following two RH problems, with the same jump matrix $\bm J$, but different asymptotic behaviour at infinity. 
\begin{RHP}[\bf Integral operator ${\mathcal K}^2$]\label{RHPK2}
Find a meromorphic matrix-valued function ${ \bm \Xi} : \C \setminus \{\Sigma_1 \cup \Sigma_2\} \to \R^{2\times 2} $ such that
\begin{align*}
&{\bm \Xi}_+(k) = {\bm \Xi}_-(z) \bm J(k) \qquad k \in \Sigma_1 \cup \Sigma_2 \\
&{\bm \Xi}(k) = \bm I + \frac{{\bm \Xi}_1}{k} + \mathcal{O}\le(\frac{1}{k^2}\ri) \qquad k\to \infty
\end{align*}
with $\Sigma_1$ and $\Sigma_2 $ both oriented upwards.
\end{RHP}
The  jump matrix reads
\begin{gather}
\bm J (k)  = \begin{bmatrix} 1& - r(k) e^{-2i \, \theta(k; x,t)} \mathbbm{1}_{{\Sigma_1}}(k) \\  r(-k)e^{2i \, \theta(k; x,t)}\mathbbm{1}_{{ \Sigma_2}}(k) & 1 \end{bmatrix} 
\end{gather} 
and satisfies the symmetry $\bm J(-k) = { \sigma}_1 \bm J(k) { \sigma_1}$, with ${ \sigma}_1 =  \begin{bmatrix} 0&1\\1&0 \end{bmatrix}$. Here $\mathbbm{1}_A$ is the characteristic function of the set $A$.  We notice that $r(-k)=\overline{r(\bar{k})}$.

\begin{RHP}[\bf Integral operator $ {\mathcal K}$]
Find a meromorphic matrix-valued function ${\bm \Gamma} : \C \setminus \{\Sigma_1 \cup \Sigma_2\} \to \R^{2\times 2} $ such that
\begin{align*}
&{\bm \Gamma}_+(k) = {\bm \Gamma}_-(k) \bm J(k) \qquad k \in \Sigma_1 \cup \Sigma_2 \\
&{\bm \Gamma}(k) = \begin{bmatrix}1&1\\ -i k & i k \end{bmatrix} \le[ \bm I+ \frac{{ \bm \Gamma}_1}{k} + \mathcal{O}\le(\frac{1}{k^2}\ri) \ri] \qquad k\to \infty \\
& {\bm \Gamma}_1 = \begin{bmatrix} a_1&0\\0&-a_1\end{bmatrix} \\
&{\bm \Gamma}(-k) = { \bm \Gamma} (k) \begin{bmatrix} 0&1\\1&0 \end{bmatrix} 
\end{align*}
\end{RHP}

Furthermore, the following relationship holds between ${\bm \Gamma}$ and ${\bm \Xi}$:
\[
{ \bm \Gamma}(k) = \begin{bmatrix} 1 & 1 \\ -i k +2i \Xi_{1,12} & i k + 2i \Xi_{1,12} \end{bmatrix} {\bm \Xi}(k)\ ,
\]
which implies 
\begin{equation}
\label{a1_relation}
a_1={\bm \Xi}_{1,11}- {\bm \Xi}_{1,12}.
\end{equation}

The log-derivative of the respective Fredholm determinants for ${\mathcal K}$ and ${\mathcal K}^2$ can be written as (Theorem 4.1 and 4.2 in \cite{BertolaCafasso11})
\begin{align*}
& \frac{\partial}{\partial x} \ln \det \le(\Id_{L^2(\Sigma_1)} - {\mathcal K}^2\ri) = \int_{\Sigma_1\cup \Sigma_2} \Tr \le( { \bm \Xi}_-^{-1} {\bm \Xi}_-' \partial_x \bm J \bm J^{-1}\ri) \frac{\d k}{2\pi i} \\
& \frac{\partial}{\partial x} \ln \det \le( \Id_{L^2(\Sigma_1)} + {\mathcal K}\ri) = \frac{1}{2} \int_{\Sigma_1\cup \Sigma_2} \Tr \le( {\bm \Gamma}_-^{-1} {\bm \Gamma}_-' \partial_x \bm J \bm J^{-1}\ri) \frac{\d k}{2\pi i}, 
\end{align*}
thus yielding
 \begin{gather*}
q(x,t) =i \Res_{k=\infty} \Tr \Big( {\bm \Xi}^{-1}(k) {\bm \Xi}'(k)  \partial_x \bm T(k)\Big) \d k-i \Res_{k=\infty} \Tr \Big( {\bm \Gamma}^{-1}(k) {\bm \Gamma}'(k)  \partial_x \bm T(k)\Big) \d k
\end{gather*}
where $\bm T(k) = i \, \theta(k; x,t) { \sigma}_3$.  Straightforward calculations show that
\[
\Res_{k=\infty} \Tr \Big( {\bm \Gamma}^{-1}(k) {\bm \Gamma}'(k)  \partial_x \bm T(k)\Big) \d k = 2i a_1
\]
and  similarly
\[
\Res_{k=\infty} \Tr \Big( {\bm \Xi}^{-1}(k) {\bm \Xi}'(k)  \partial_x \bm T(k)\Big) \d k= 2i {\bm \Xi}_{1,11}
\]
so that according to \eqref{a1_relation} we obtain
\begin{equation}
q(x,t) = i\le(2i {\bm \Xi}_{1,11}-2i a_1\ri)= -2 {\bm \Xi}_{1,12}
\end{equation}

By setting
\begin{gather}
\bm X = \begin{bmatrix}0&e^{i \frac{\pi}{4}} \\ e^{-i \frac{\pi}{4}} &0 \end{bmatrix} \, {\bm \Xi} \,  \begin{bmatrix}0& e^{i \frac{\pi}{4}} \\ e^{-i \frac{\pi}{4}}&0 \end{bmatrix}\ , \label{X_Xi_RHP}
\end{gather}
we can recognize the RH problem \ref{rhp:X} (without the extra poles at $\pm i \kappa_0$)   with $r(-k)=\overline{r(\bar{k})}$.
 The solution to the RH problem \ref{rhp:X}  exists for any $x\in\R$ and positive time $t>0$, thanks to Proposition~\ref{theorem_existence}.  
We conclude that the soliton-gas solution of the mKdV equation can be calculated via 
\[q(x,t) = 2i \lim_{k\to \infty} k {\bm X}_{12}= i  \dfrac{\partial }{\partial x}\ln\det \le(\Id_{L^2(\Sigma_1)} + {\mathcal K}\ri) - i \dfrac{\partial }{\partial x}\ln\det \le(\Id_{L^2(\Sigma_1)} - {\mathcal K}\ri) ,\] 
thus we have concluded the proof of Theorem~\ref{theorem1}.
We also remark that the formula \eqref{intro:q2} can be easily obtained with the same reasoning from the expression \eqref{q2} for finite solitons.

\vskip .2in

\paragraph{\bf Extension to the soliton+gas setting.} The limiting procedure can be easily extended to the case of $N+M$ solitons, where $N$ of them are suitably rescaled to become a soliton gas and the remaining $M$ are \emph{not} rescaled. The linear algebra manipulations will still be the same, while the limiting operator $\mathcal K: L^2(-\infty,x) \to L^2(-\infty,x)$ will have the expression
\begin{align*}
\mathcal K [f](y,t)= \int_{-\infty}^x F(y+s,2t) f(s) \, \d s, \quad
\text{with kernel} \quad
 F(s,t) 
= -i \int_{\mathcal C} \widetilde r(\zeta) e^{-i \theta(\zeta; s,t)} \frac{\d\zeta}{2\pi i} ,
\end{align*}
where $\mathcal C = \Sigma_1 \cup \bigcup_{j=N+1}^{N+M} \mathcal C_j$, with $\mathcal C_j$ small disjoint circles (oriented counterclockwise), each one surrounding the poles $\{i \kappa_j\}_{j=N+1}^{N+M}$ and not intersecting the real line nor $\Sigma_1$, and the function $\widetilde r(\zeta)$ is defined as $\widetilde r(\zeta) = r_1(\zeta)$ for $\zeta \in \Sigma_1$ (described in \eqref{cjasR1}) and $\widetilde r(\zeta) = \frac{\chi_j}{\zeta - i\kappa_j}$ for $\zeta \in \mathcal C_j$ ($\forall \ j=N+1,\ldots, N+M$). The rest of the calculations remain almost unchanged. When deriving the limiting RH problem, the jumps on the $\mathcal C_j$'s can be easily converted back to residue conditions as in the RH problem~\ref{rhp:X} (case $M=1$).

\section{Setup of the asymptotic problem}
\label{sec-setup}

As previewed in Section \ref{sec-intro}, we consider here a gas of mKdV solitons with positive velocity initially supported on a right half-line interacting with a distinguished soliton traveling faster than the gas. 
Spectrally, the gas is described by a jump across a contour $\Sigma_1 $ and its complex conjugate $\Sigma_2 $. We choose to orient both contours upward. To the gas spectrum we add a pair of discrete spectral points $\pm i \kappa_0$, where the condition $\kappa_0 > \eta_2$ ensures the distinguished soliton travels faster than the gas. The solution of this problem is encoded into  RH problem~\ref{rhp:X} which will be our principle object of interest.

To see that  RH problem~\ref{rhp:X} initially encodes data supported on a right half-line, notice that when $t=0$, the phase function $\theta(k; x,t)$ reduces to $\theta(k; x,0) = x k$, and clearly for $x \ll -1$ we have $e^{-2i \theta(k; x,0)} = \bigo{ e^{2\eta_1 x}}$ for $k \in \Sigma_1$. As the jumps are exponentially near identity for $x \ll -1$, the solution of the Riemann-Hilbert problem, up to exponential corrections, is the one encoded by only the residues conditions at $\pm i\kappa_0$ (the distinguished soliton component). This justifies the claim that the gas is initially supported on a right half-line. 

\begin{remark}
Without the poles, this problem is similar to that previously studied in \cite{GirGraJenMcL}. In \cite{GirGraJenMcL} a soliton gas for the KdV equation was studied supported on a \text{left} half-line, whereas in this paper we study the modified KdV equation, and have constructed a soliton gas supported on a \text{right} half-line. It is straightforward in either the KdV or modified KdV equation to construct a soliton gas solution which is supported on ``\text{the other}" half-line.  The end result of those manipulations is that the signs of $x$ and $t$ in the phase function $\theta(k;x,t)$ appearing in the Riemann--Hilbert problem are flipped. This explains why the signs in the exponential factors containing $\theta(k;x,t)$ in the present paper are different from those in \cite{GirGraJenMcL}.  The fact that we can freely change the signs of $x$ and $t$ takes advantage of the fact that if $q(x,t)$ is a solution of mKdV (resp. KdV) then $\pm q(-x,-t)$ (resp. $q(-x,-t)$) generates a second solution.  Note, however, that the dynamics of the KdV equation for $t>0$ with initial data $q_{0}(-x)$ will be markedly different from $q_{0}(x)$.
\end{remark}

On its own, the discrete spectral data $( i \kappa_0, \chi) \in \C^+ \times \R \setminus \{0\}$ encodes a soliton solution \eqref{soliton} of \eqref{mkdv} with phase shift \eqref{x0}. 
In order to use asymptotic methods to study the interaction of the distinguished soliton and the soliton gas we choose $x_0 \ll -1$, and use this initial position as our large parameter in subsequent calculations. At the level of the scattering data that means we take 
\begin{equation}\label{chi0}
	\chi = \chi(x_0, \kappa_0) = 2\kappa_0 \varsigma  e^{-2\kappa_0 x_0}  \qquad \varsigma = \sgn(\chi) \in \{ -1, +1\}.
\end{equation}

In what follows we assume that $r(k)$ is real and non-vanishing on $\Sigma_1$. 
Note also that in contrast with analogous RH problems that appear in the long-time asymptotics for step-like problems \cite{KM2010,GravaMinakov20}, where the behavior of a function analogous to $r(k)$ at the edge points is of square-root type, here we do not have that restriction, and $r(k)$ might have any type of edge behavior.

\subsection{Steepest descent preparations}
\label{sec-g.function}

The first step in the asymptotic analysis of the RH problem~\ref{rhp:X} is to construct a scalar function $g(k) = g(k;x,t)$ which controls the terms with exponential growth in the jumps \eqref{Xjumps}.
Given $t\geq 0$, if $x < 4 \eta_1^2 t$, then $\Im \theta(k;x,t) <0$ for $k \in i \R_+$; therefore, the jumps are exponentially close to the identity matrix and the small norm theory guarantees that the solution is quiescent in this domain (unless the soliton is scaled to be present here).   Our analysis is inspired from the analysis in  \cite{GravaMinakov20}, \cite{EGKT}.

For $x > 4\eta_1^2 t$, two growing bands $(i\eta_1,i\alpha)\cup(-i\alpha,-i\eta_1)$ emerge from the endpoints $\pm i \eta_1$, subsets of $\Sigma_1$ and $\Sigma_2$ respectively, on which the exponential terms are asymptotically large. In this setting, we introduce a suitable scalar function $g(k;x,t)$ through the transformation:
\begin{equation}\label{T}
	\bm{T}(k) =  \bm{X}(k) e^{-i g(k; x,t) \sigma_3} f(k)^{\sigma_3} \ .
\end{equation}
Such a $g$-function will need to satisfy the following conditions:
\begin{enumerate}[label=\textrm{(\roman*}), align = Center] 
	\item \label{g1}$g$ is analytic in $\C \setminus [-i\alpha, i\alpha]$.
	\item $g_+(k) + g_-(k) + 8k^3 t + 2k x = 0$, within each band $k \in \Sigma_{1,\alpha} \cup \Sigma_{2,\alpha}$.
	\item \label{g3}$ g_+(k) - g_-(k) =- \Omega(x,t)$, in the gap $k \in (-i\eta_1, i\eta_1)$.
	\item \label{g4}$g(k) = g_0(x,t) + \bigo{k^{-1}}$ as $k \to \infty$. 
	\item \label{g5} near the endpoints $k=\pm i \eta$ (where $\eta$ is either $\eta_1, \alpha$ or $\eta_2$), 
	\[
		g(k) + 4k^3 t + k x = \bigo{ ( z \mp i \eta )^{p/2}}, \qquad k \to \pm i \eta,
	\] 
	with $p =1$ if the endpoint is stationary (i.e., for $\eta = \eta_j$ or $j=1,2$) and $p=3$ if the endpoint is allowed to vary with $(x,t)$ (i.e., for $\eta = \alpha$).
	\end{enumerate}
For now we will not specify the function $f(z)$ in \eqref{T}, except to remark that we assume it satisfies some jump conditions, to be chosen later, along the same contour $[-i\alpha, i\alpha]$.

With such a $g$-function, the transformation \eqref{T} results in a new RH problem for $\bm{T}(k)$ with jumps given by
	\begin{equation}\label{T jump}
	\begin{gathered}
		\bm{T}_+(k) = \bm{T}_-(k) \bm{J}_T(k), \\
		\bm{J}_T(k) = \begin{dcases} 
			\begin{bmatrix}
				e^{-i (g_+(k) - g_-(k) )} \frac{f_+(k)}{f_-(k)} & 0  \\
				i r(k) f_+(k) f_-(k) & e^{i (g_+(k) - g_-(k) )} \frac{f_-(k)}{f_+(k)}
			\end{bmatrix},
			& k \in \Sigma_{1,\alpha}, \\
			\begin{bmatrix}
				e^{-i (g_+(k) - g_-(k) )} \frac{f_+(k)}{f_-(k)} & i \frac{\overline{r(\bar k)} }{f_+(k) f_-(k)} \\
				0 & e^{i (g_+(k) - g_-(k) )} \frac{f_-(k)}{f_+(k)}
			\end{bmatrix},
			& k \in \Sigma_{2,\alpha}, \\	
			e^{i {\Omega}(x,t)  \sigma_3} \left(\frac{ f_+(k)}{f_-(k)} \right)^{\sigma_3}, & k \in i[-\eta_1, \eta_1].
		\end{dcases}
	\end{gathered}
	\end{equation}

\begin{figure}
\centering
\vspace*{5pt}
	\begin{overpic}[scale=0.4]{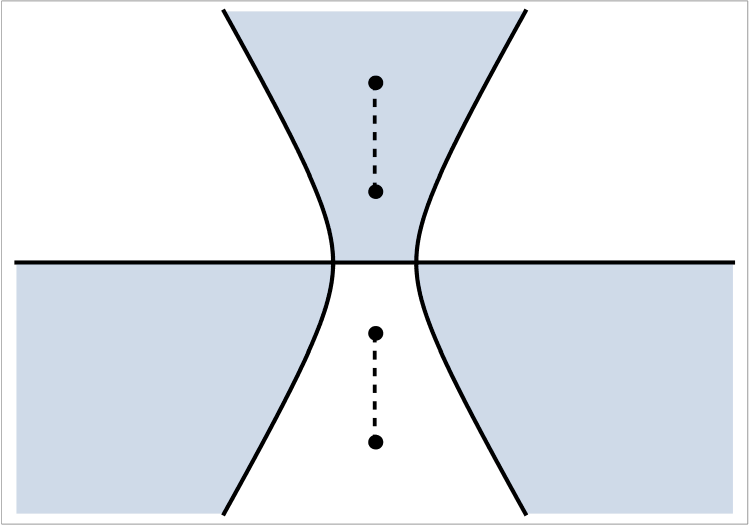}
	   \put(40,75){$v<0$}
	\end{overpic}
	\quad
	\begin{overpic}[scale=.4]{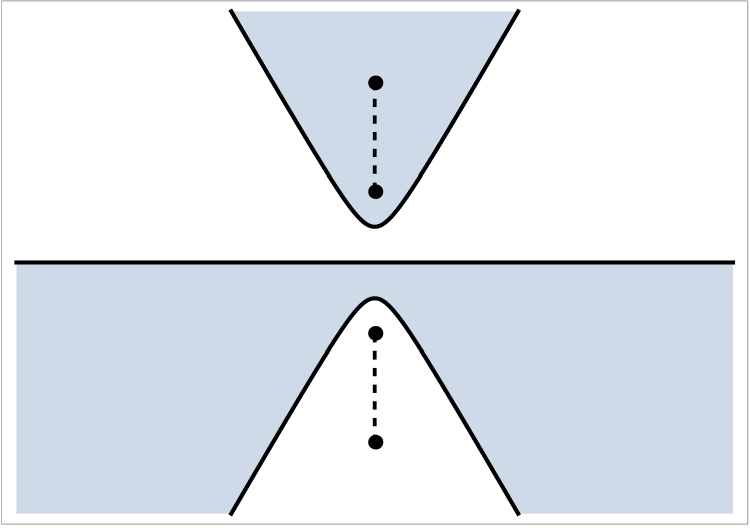}
	   \put(34,75){$0<v<4\eta_1^2$}
	\end{overpic}
	\quad
	\begin{overpic}[scale=.4]{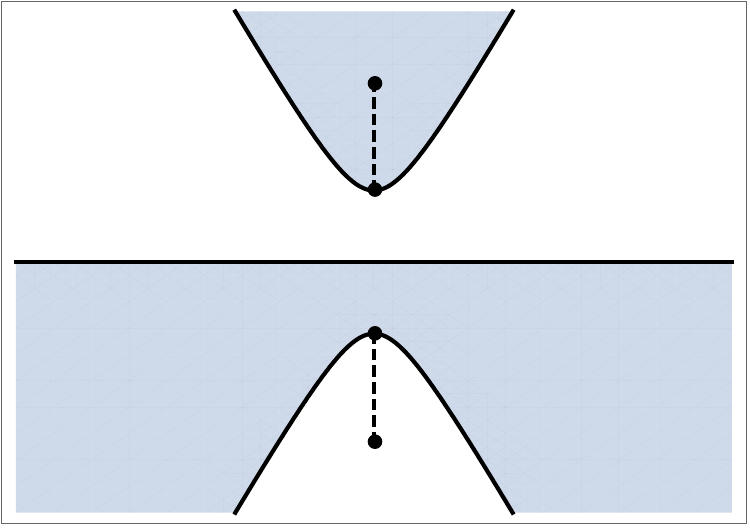}
	   \put(37,75){$v=4\eta_1^2$}
	\end{overpic}
	\\[5pt]
	\begin{overpic}[scale=.4]{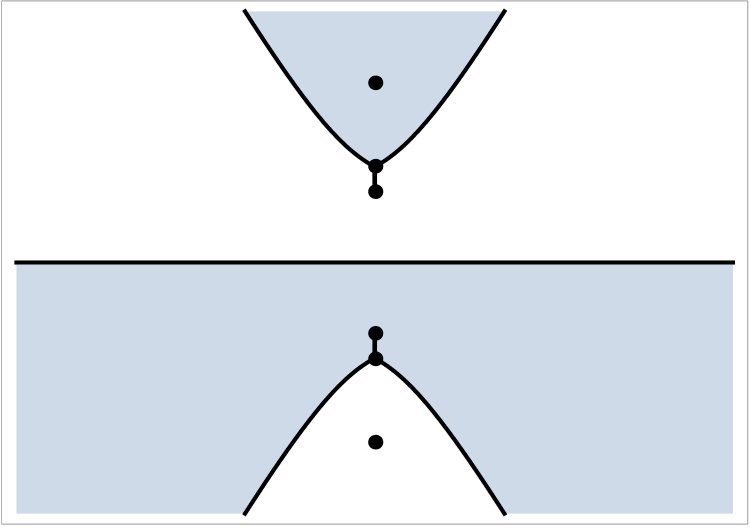}
	   \put(27,-6){$4\eta_1^2<v< v_2$}
	\end{overpic}
	\quad
	\begin{overpic}[scale=.4]{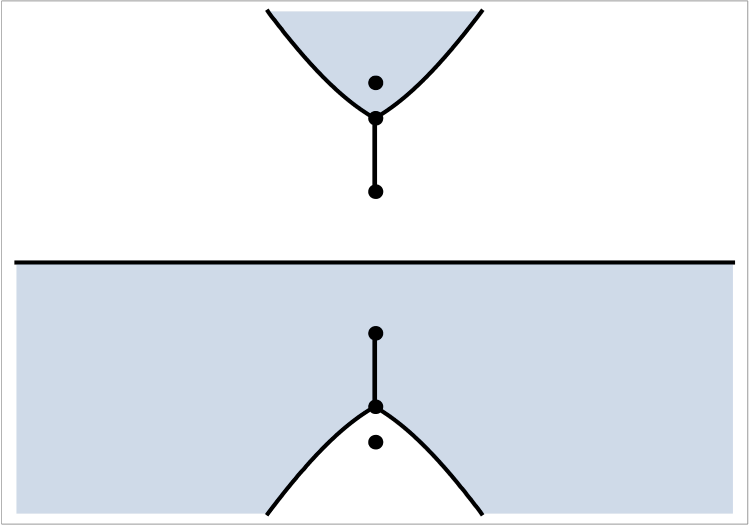}
	   \put(27,-6){$4\eta_1^2<v< v_2$}
	\end{overpic}
	\quad
	\begin{overpic}[scale=.4]{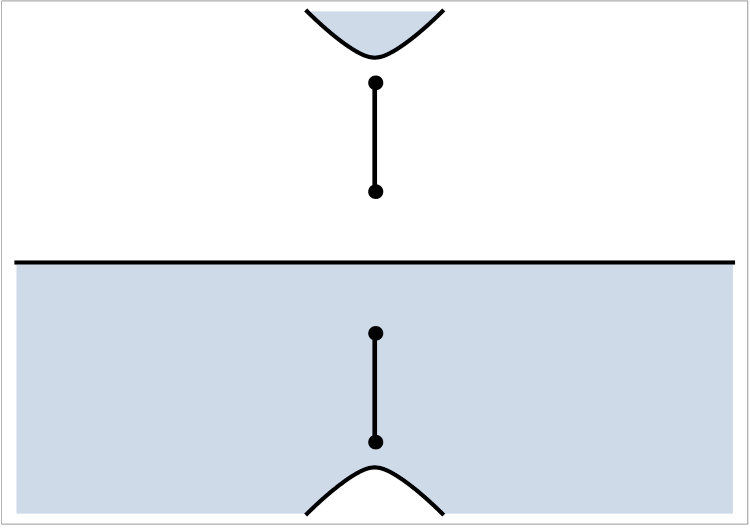}
	   \put(40,-6){$v> v_2$}
	\end{overpic}
\caption{
The signature table for $\Im \varphi$. White (resp. shaded) regions show where $\Im \varphi>0$ (resp. $\Im \varphi<0$). 
As $v = \tfrac{x}{t}$ increases past $4\eta_1^2$, a band of non-analyticity $(i\eta_1, i\alpha)$ emerges from $\eta_1$; the bands grow until $v$ reaches $v_2$ where $\alpha = \eta_2$. For $v > v_2$, the band  $(i\eta_1, i\eta_2)$ is fixed.
}
\label{fig:phase}
\end{figure}

With the following theorem we identify the exact value of the endpoints of the bands $\pm i \alpha$ and we prove the existence and explicit expression of the $g$ function.
\begin{theorem}\label{thm:g}
Let $\eta_1, \eta_2$ be the endpoints of the jump contour defining  RH problem~\ref{rhp:X}.  These values determine a unique positive quantity $v_2> 4\eta_1^2$ satisfying the equation
\begin{align}\label{Whitham W}
	v_2 = \eta_2^2 \, W \! \left( \eta_1^2/\eta_2^2  \right), \quad W(m) = \frac{4(1-m) K(m)}{E(m)} + 2(1+m),
\end{align}
where $K(m)$ and $E(m)$ are the complete elliptic integrals of the first and second kind respectively. 
Then for every $x > 4\eta_1^2 t$ there exists a function $g(k)$ satisfying conditions \ref{g1}-\ref{g5} above, with $\alpha = \alpha(x/t) \in (\eta_1, \eta_2]$, such that:
\begin{itemize}
\item For $ \tfrac{x}{t} \in (  4\eta_1^2,  v_2) $, $\alpha(x/t) \in (\eta_1, \eta_2)$ and it satisfies the self-similar Whitham evolution equation
\begin{equation}\label{Whitham}
	\frac{x}{t} = \alpha^2 W(\eta_1^2/\alpha^2) .
\end{equation}
Moreover, the evolution of $\alpha(v)$ is hyperbolic: $\frac{\d\alpha}{\d v} >0$, $\lim_{v \to 4\eta_1^2} \alpha(v)= \eta_1$ and $\lim_{v \to v_2} \alpha(v) = \eta_2$.

\item For $\tfrac{x}{t} \geq v_2$, $\alpha(x/t) \equiv \eta_2$, i.e., the jump of the function $g$ is supported along the entire length of the original jump contour $\Sigma_1 \cup \Sigma_2.$
\end{itemize}

Finally, let
\begin{equation}\label{phi.def}
	\varphi(k;x,t) := g(k;x,t) + k x + 4 k^3 t;
\end{equation}
then, for any $x > 4\eta_1^2 t$, we have 
\begin{equation}
	\varphi(k;x,t) = R(k) \left[ 4t k - \frac{{\Omega}}{2\pi i} \int\limits_{-i\eta_1}^{i\eta_1} \frac{\d s}{R(s)(s-k)} \right],
\end{equation}
where 
\begin{equation}\label{omega.1}
	\Omega = \Omega(x,t) = \frac{\pi \alpha }{K(m)} (x-2(\eta_1^2+ \alpha^2) t ), \quad m = \frac{\eta_1^2}{\alpha^2}
\end{equation}
and the function
\begin{equation}\label{Root}
	R(k) = \sqrt{ (k^2 + \eta_1^2)(k^2 + \alpha^2)}
\end{equation}
is analytic away from $\Sigma_{1,\alpha} \cup \Sigma_{2,\alpha}$ and normalized such that $R(k) = k^2 +\bigo{1}$ as $k\to \infty$. 
\end{theorem}

\begin{remark}\label{rem:g.left}
	For $x < 4 \eta_1^2 t$, we have no need for a $g$-function ($g(k;x,t) \equiv 0$), so we set 
\[
\varphi(k;x,t) = \theta(k;x,t) \qquad  \text{for } x < 4 \eta_1^2 t\ .
\]
This allows us to use the same notation, $\varphi(k;x,t)$, for the phase function in all regions of spacetime. 
\end{remark}

\begin{proof}[Proof of Theorem~\ref{thm:g}]
Using the Sokhotski–Plemelj formula we can write $g$ as 
\begin{equation}
\label{eq:gformula}
	g(k) = \frac{R(k)}{2\pi i} \left[ \int\limits_{\smash{\ \Sigma_{1,\alpha} \cup \Sigma_{2,\alpha}}} \frac{-2x s - 8t s^3}{R_+(s)} \frac{\d s}{s-k} - \int\limits_{-i\eta_1}^{i\eta_1} \frac{{\Omega}}{R(s)} \frac{\d s}{s-k} \right],
\end{equation}
which satisfies conditions \ref{g1}-\ref{g3} for any choice of $\Omega$ and $\alpha$. The asymptotic expansion of $g(k;x,t)$ at $\infty$ gives 
\begin{equation}
\begin{aligned}
	g(k) &= g_1 k + g_0 + \bigo{k^{-1}}, \qquad k \to \infty, \\
	g_1 &= \frac{1}{2\pi i} \int\limits_{\mathclap{\ \Sigma_{1,\alpha} \cup \Sigma_{2,\alpha}}} \frac{2xs + 8ts^3}{R_+(s)} \, \d s + \frac{{\Omega}}{2\pi i} \int\limits_{-i\eta_1}^{i\eta_1}
	   \frac{\d s}{R(s)},  \\
	g_0 &= \frac{1}{2\pi i} \int\limits_{\mathclap{\ \Sigma_{1,\alpha} \cup \Sigma_{2,\alpha}}} \frac{s(2xs + 8ts^3)}{R_+(s)} \, \d s + \frac{{\Omega}}{2\pi i} \int_{-i\eta_1}^{i\eta_1} 
	\frac{s }{R(s)}\, \d s  \equiv 0\qquad \text{(by symmetry)}.
\end{aligned}
\end{equation}
We now  require that $g$ satisfies \ref{g4}, i.e. $g_1\equiv 0$. This yields one real equation which determines ${\Omega}$ as
\begin{equation}
	{\Omega} = \frac
	{x-2(\eta_1^2+ \alpha^2) t }	
	{ \frac{1}{2\pi i} \int_{-i\eta_1}^{i\eta_1} \frac{\d s}{R(s)} } ,
\end{equation} 
which is equivalent to \eqref{omega.1}.

If $\alpha$ is stationary, then the description is complete. 
In the modulation zone (i.e. $\alpha \neq \eta_1,\eta_2$),  condition \ref{g5} determines the motion of the moving branch points $\pm i \alpha$. 
\color{black} Fix $s_0 > \max \{ |k|, \alpha \}$. \color{black} Using the residue theorem we can write 
\begin{equation}\label{g.2}
	g(k) + xk + 4t k^3 = \frac{R(k)}{2\pi i} \left[ \, \oint\limits_{{\color{black} |s| = s_0}} \frac{x s + 4t s^3}{R(s)} \frac{\d s}{s-k} -  \int\limits_{-i\eta_1}^{i\eta_1} \frac{{\Omega}}{R(s)} \frac{\d s}{s-k} \right],
\end{equation}
where the loop integral is positively oriented. 
A necessary and sufficient condition for $g$ to satisfy \ref{g5} is then
\[
	\frac{1}{2\pi i}\oint\limits_{{\color{black} |s| = s_0}} \frac{x s + 4t s^3}{R(s)} \frac{\d s}{s-i \alpha} - \frac{1}{2\pi i} \int\limits_{-i\eta_1}^{i\eta_1} \frac{{\Omega}}{R(s)} \frac{\d s}{s-i \alpha} = 0\ ;
\]
the integrals can be evaluated exactly, giving
\[
	4t i \alpha - \frac{{\Omega}\cdot i \eta_1}{2\pi i \eta_1 \alpha\cdot (-i \alpha) }  \int_{-1}^1 \frac{\d u}{\sqrt{(1-u^2)(1-m\, u^2) } (1- \sqrt{m} u) } = 0,
	\]
	that simplifies to 
	\begin{equation}
	\label{whitham.00}
	4t\alpha - \frac{{\Omega}}{\pi \alpha^2} \frac{E(m)}{1-m} = 0
\end{equation}
where $m = \eta_1^2/\alpha^2$ and $E(m)$ is the complete elliptic integral of second kind:
\[
E(m) := \int_0^1 \frac{\sqrt{1- m s^2}}{\sqrt{ 1-s^2}} \d s\ .
\]
Using \eqref{omega.1}, equation \eqref{whitham.00} is equivalent to \eqref{Whitham},
which gives the modulation equations determining the motion of the branch point $\alpha$ as a function of $v=x/t$. We note that \eqref{Whitham} is equivalent to the condition that 
\[
	x - 2(\eta_1^2+\alpha^2) t = \frac{4t(\alpha^2-\eta_1^2) K(m)}{E(m)} > 0
\]
which is strictly positive since $\alpha > \eta_1$ (and thus $0 < m < 1$ and ${\Omega} \in \R_+$).

Notice that \eqref{Whitham} reduces to   
\begin{align*}
	\frac{x}{t} = 4 \eta_1^2
	\qquad \text{as} \quad
	\alpha \to \eta_1,
\end{align*}
which is consistent with our initial assumption that the bands emerge out of the points $\pm i \eta_1$ for $x> 4\eta_1^2 t$.  We can implicitly differentiate the Whitham equation \eqref{Whitham} with respect to $v = \tfrac{x}{t}$ to get 
\[
	2 \alpha \alpha_v = \left[ W(m) - m W'(m) \right]^{-1} \ .
\]
Since $\alpha>0$ by construction, this is enough to show that $\alpha_v>0$: indeed, $W(m)>0$ for $0<m<1$  and 
\[
	W'(m) = -2 \frac{(K(m)-E(m)) ((1+m) E(m) - (1-m) K(m))}{m E(m)^2} < 0\ , 
\]
where the last inequality follows from the known inequality $\frac{K(m)}{E(m)}< \frac{1}{\sqrt{1-m}}$ \cite[\S19.9.8]{DLMF}.

Therefore, there exists a critical value $v_2$, with $v_2 > 4 \eta_1^2$, such that 
\[
	v_2 = \eta_2^2 \, W(\eta_1/\eta_2),  \qquad {\text{i.e.,} \qquad \alpha(v_2) = \eta_2}.
\]
\end{proof}

\subsection{The $g$-function in terms of abelian integrals}\label{g_function_abelian}
The $g$-function defined above can also be expressed in terms of abelian integrals associated to the two-sheeted genus-one Riemann surface 
\begin{equation}\label{surface}
	\mathfrak{X} = \{ (k,\eta) \in \C^2\, :\, \eta^2 =R(k)^2 = (k^2 + \eta_1^2)(k^2 + \alpha^2) \},
\end{equation}
where the first sheet of $\mathfrak{X}$ is identified by the fact that $R(k) >0$ for $\Im k = 0$. 
Denote by $\infty^+$ ($\infty^-$) the pre-image of $k=\infty$ on the first (second) sheet of $\mathfrak{X}$. 
We fix a canonical homology basis on $\mathfrak{X}$ by choosing $\mathcal B$ to encircle $\Sigma_1$ clockwise on the first sheet, and $\mathcal A$ to pass from the positive side of $\Sigma_2$ to $\Sigma_1$ on sheet 1 and from the negative side of $\Sigma_1$ to $\Sigma_2$ on sheet 2. See \figurename~\ref{fig:homology}.

\begin{figure}[th]
\centering
\scalebox{.9}{
\begin{tikzpicture}[>=stealth]
\path (0,0) coordinate (O);

\coordinate (TL) at (-2,5);
\coordinate (TR) at (7,5);
\coordinate (BL) at (-4,3);
\coordinate (BR) at (5,3);
\coordinate (INF1) at (5.4,4.3);

\coordinate (shift) at (0,-2.3);
\coordinate (TL2) at ($  (TL) + (shift)  $);
\coordinate (TR2) at ($  (TR) + (shift)  $);
\coordinate (BL2) at ($  (BL) + (shift)  $);
\coordinate (BR2) at ($  (BR) + (shift)  $);
\coordinate (INF2) at ($ (INF1) + (shift) $);

\coordinate (eta1) at (1.5,4);
\coordinate (eta2) at (4,4);
\coordinate (alpha) at (3.2,4);

\coordinate (-eta1) at (0.5,4);
\coordinate (-eta2) at (-2,4);
\coordinate (-alpha) at (-1.2,4);

\draw (TL) -- (TR) -- (BR) -- (BL) -- cycle;
\node[ label={[label distance= -0.3cm, below, xshift= -0.1cm]$\times$}]  at (INF1) {$\infty^+$};
\draw (TL2) -- (TR2) -- (BR2) -- (BL2) -- cycle;
\node[ label={[label distance= -0.3cm, below, xshift= -0.1cm]$\times$}]  at (INF2) {$\infty^-$};

\draw[->-=0.5] (eta1) -- (alpha);
\draw[->-=0.5] (-alpha) -- (-eta1);

\draw[->-=0.5] ($(eta1)+(shift)$) -- ($ (alpha) +(shift) $);
\draw[->-=0.5] ($(-eta1)+(shift)$) -- ($ (-alpha) +(shift) $);

\foreach \pos/\label in {eta1/i\eta_1, eta2/i\eta_2, alpha/i\alpha,-eta1/-i\eta_1,-eta2/-i\eta_2,-alpha/-i\alpha}{
\node[circle,fill=black, inner sep=0pt,minimum size=3pt,label=below:{\tiny $\label$}] at  (\pos) {};
\node[circle,fill=black, inner sep=0pt,minimum size=3pt,label=below:{\tiny $\label$}] at  ($ (\pos)+(shift) $)  {};
}

\foreach \pos in {eta1,alpha,-eta1,-alpha}{
\draw[dashed, black!30] (\pos) -- ($ (\pos) + (shift) $);
}

\draw[->- = .25, red] ($ 0.3*(-eta1)+0.7*(-alpha) $) .. controls + (70:.7cm) and + (70:.7cm) .. ($ 0.3*(eta1)+0.7*(alpha)  $);
\draw[->- = .25, red]  ($ 0.3*(eta1)+0.7*(alpha) +(shift) $) .. controls + (-110:0.7cm) and + (-110:.7cm) .. ($ 0.3*(-eta1)+0.7*(-alpha)  + (shift) $);
\draw[red!30, dashed] ($ 0.3*(-eta1)+0.7*(-alpha)  $) -- ++ (shift);
\draw[red!30, dashed] ($ 0.3*(eta1)+0.7*(alpha)  $) -- ++ (shift);
\node[above, red] at (1,4.5) {\small $\mathcal A$};
\draw[-<- = .25, blue] ( $0.5*(eta1)+0.5*(alpha) $) ellipse (1.3cm and .6cm);
\node[blue, above] at (3.25,4.4) {\small $\mathcal B$};

\end{tikzpicture}
}
\caption{The homology basis for the Riemann surface 
$\mathfrak{X}$ associated with 
$R(k) = \sqrt{ (k^2+\eta_1^2)(k^2 +\alpha^2)}$.
}
\label{fig:homology}
\end{figure}
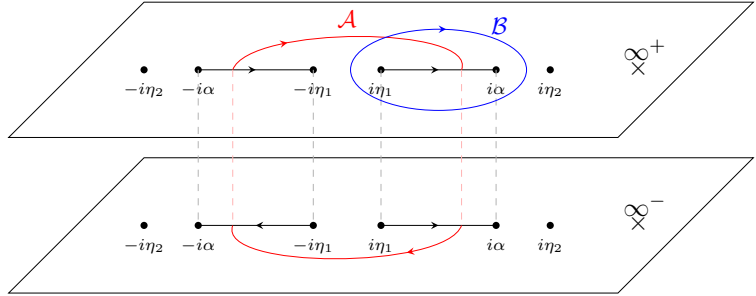

Using the representation \eqref{g.2} of $g(k;x,t)$,  we have that 
\begin{equation}\label{varphi}
	\varphi(k;x,t) = g(k;x,t) + x k + 4 t k^3 = \frac{R(k)}{2\pi i} \left[ \, \oint\limits_{{\color{black} |s| = s_0}} \frac{x s + 4t s^3}{R(s)} \frac{\d s}{s-k} -  \int\limits_{-i\eta_1}^{i\eta_1} \frac{{\Omega}}{R(s)} \frac{\d s}{s-k} \right]
\end{equation}
 satisfies the relations $\varphi_+(k) + \varphi_-(k) = 0$ for $k \in \Sigma_{1,\alpha} \cup \Sigma_{2,\alpha}$, $\varphi_+(k) - \varphi_-(k) = -{\Omega}<0$ for $k \in [-i\eta_1,  i\eta_1]$, and $\varphi(k) = 4tk^3 + x k + \bigo{k^{-1}}$ as $k \to \infty$. These relations show that $\varphi$ can also be represented in terms of abelian integrals 
\[
	\varphi(k;x,t) =   t \varphi_2(k) + x \varphi_0(k) .
\]
Here $\varphi_0$ and $\varphi_2$ are given by
\[
	\varphi_0(k) = \int_{i\alpha}^k \d\varphi_0, \qquad \varphi_2 = \int_{i\alpha}^k \d\varphi_2(k), 
\]
where
\begin{align*}
	\d \varphi_0(\zeta) &= \frac{ \zeta^2 + c_0}{R(\zeta)} \d\zeta = \left[ 1 + \bigo{\zeta^{-2}} \right] \d\zeta,  
	\\
	\d \varphi_2(\zeta) &= 12 \frac{ \zeta^4 + \tfrac{1}{2}(\eta_1^2 + \alpha^2) \zeta^2 + c_2}{R(\zeta)} \d\zeta = \left[ 12 \zeta^2 + \bigo{\zeta^{-2}} \right] \d\zeta,
\end{align*}
as $\zeta \to \infty^+$, and the constants $c_0$ and $c_2$ are chosen to ensure that $\oint_{\mathcal A} \d\varphi_j = 0$, $k=0,2$. 
This gives
\begin{equation}
\begin{aligned}
	&c_0 = - \left( \int_0^{i\eta_1} \frac{\d\zeta}{R(\zeta)} \right)^{-1} 
	\left(   \int_0^{i\eta_1} \frac{\zeta^2}{R(\zeta)} \, \d\zeta \right) 
	= \alpha^2 \left( 1 - \frac{ E(m)}{K(m)} \right) \\
	&c_2 = - 12 \left( \int_0^{i\eta_1} \frac{\d\zeta}{R(\zeta)} \right)^{-1} 
	\left(   \int_0^{i\eta_1} \frac{\zeta^4+\tfrac{1}{2}(\eta_1^2+\alpha^2)\zeta^2}{R(\zeta)} \, \d \zeta \right) 
	= \frac{\alpha^4}{6} \left( (1+m) \frac{E(m)}{K(m)} - (1-m) \right)
	\end{aligned}
\end{equation}
Therefore $\varphi(k;x,t)$ can be alternatively represented as
\begin{align}
	\label{varphi.3a}
	\varphi(k;x,t) = t \varphi_2(k) + x \varphi_0(k) 
	&=  \int_{i \alpha}^k \frac{12 t ( \zeta^4 + \tfrac{1}{2}(\eta_1^2 + \alpha^2) \zeta^2 + c_2) + x ( \zeta^2 + c_0)}{R(\zeta)} \, \d\zeta, \\
	\label{varphi.3b}
	&=  \int_{i \alpha}^k \frac{( \zeta^2 + \alpha^2) (12 t (\zeta^2 - \tfrac{1}{2}(\alpha^2-\eta_1^2) ) + x )}{R(\zeta)} \, \d\zeta \ .
\end{align} 
\begin{remark}\label{Dalpha=0}
Note that \eqref{Whitham} is equivalent to the condition that $\zeta^2+\alpha^2$ is a factor of the numerator of the integrand in \eqref{varphi.3a}; this additionally implies that $\partial_\alpha \varphi(k;x,t) \equiv 0$.
\end{remark}

For $k \in [-i\eta_1, i \eta_1]$ we have
\begin{equation}\label{omega.phi.ointB}
	\varphi_+(k) - \varphi_-(k) =  \left( - \oint_{\mathcal B}	 \d\varphi_0 \right) x + \left( -\oint_{\mathcal B} \d\varphi_2 \right) t = -{\Omega}\ ,
\end{equation}
therefore
\[
	\oint_{\mathcal B} \d\varphi_0 = \frac{\pi \alpha}{K(m)} \ ,
	\qquad 
	\oint_{\mathcal B} \d\varphi_2 = -\frac{2\pi \alpha}{K(m)}(\eta_1^2 +\alpha^2)\ ,
\]
by comparison with \eqref{omega.1}.

\subsection{Preparing the problem to open lenses}\label{sect:preparation_opening_lenses}
The structure of the $g$-function described in Theorem~\ref{thm:g} separates the $(x,t)$ half-plane into three sectors
\begin{equation}
	\label{sectors}
	\begin{gathered}
	   \mathcal{S}_L = \left\{ (x,t) \in \R \times \R_+ \, :\,  \frac{x}{t} < 4\eta_1^2 \right\}, \\
	   \mathcal{S}_M = \left\{ (x,t) \in \R \times \R_+ \, :\,  4\eta_1^2 < \frac{x}{t} < v_2  \right\}, \\ 
	   \mathcal{S}_R = \left\{ (x,t) \in \R \times \R_+ \, :\,  \frac{x}{t} > v_2  \right\}. \\ 
	\end{gathered}
\end{equation} 
We further subdivide these sectors by writing
\begin{gather}
	\begin{aligned} 
	   \mathcal{S}_j^{^{(+)}} &= \left\{ (x,t) \in \mathcal{S}_j \,:\, \ln \left| \frac{ \chi }{2\kappa_0} \right| + 2\Im \varphi(i\kappa_0;x,t) > 0 \right\},  \\
	   \mathcal{S}_j^{^{(-)}} &= \left\{ (x,t) \in \mathcal{S}_j \,:\, \ln \left| \frac{ \chi }{2\kappa_0} \right| + 2\Im \varphi(i\kappa_0;x,t) < 0 \right\},
	\end{aligned}
	\qquad j \in \{ L,M,R\},
\shortintertext{and}
	\label{before.after}
	\mathcal{S}^{^{(\pm)}} = \mathcal{S}_L^{^{(\pm)}} \cup \mathcal{S}_M^{^{(\pm)}} \cup \mathcal{S}_R^{^{(\pm)}} 	,
\end{gather}
which captures whether the coefficients $\chi e^{i \varphi(i\kappa_0;x,t)}$, which will appear in the residue condition \eqref{Tres}, are asymptotically large or small respectively. Physically, the set $\mathcal{S}^{^{(+)}}$ is the set of points in front of the soliton or before the soliton passes by, while $\mathcal{S}^{^{(-)}}$ corresponds to points in spacetime behind  the soliton or after the soliton has passed by. 

Our next step is to properly define the function $f(k)$ in the definition \eqref{T} of $\bm{T}(k;x,t)$. Motivated by the jump conditions \eqref{T jump} we choose $f$ as follows:
\begin{subequations}\label{f}
\begin{equation}
	f(k; x,t) = 
	\begin{dcases}
	  1, & (x,t) \in \mathcal{S}_L^{^{(-)}}, \\
	  \tfrac{k-i\kappa_0}{k+i\kappa_0}, & (x,t) \in \mathcal{S}_L^{^{(+)}}, \\
	  f^{^{(-)}}(k;x,t), & (x,t) \in \mathcal{S}_M^{^{(-)}} \cup \mathcal{S}_R^{^{(-)}},  \\
	  f^{^{(+)}}(k;x,t), & (x,t) \in \mathcal{S}_M^{^{(+)}} \cup \mathcal{S}_R^{^{(+)}}, 
	\end{dcases}
\end{equation}
where
\begin{gather}
\begin{aligned}
	&f^{^{(-)}}(k;x,t) = \exp \left\{ \frac{R(k)}{2\pi i}\left[  
	   \int_{\Sigma_{1,\alpha}} \frac{ -\ln r(s)  }{R_+(s) (s-k)}  \, \d s 
	   + \int_{\Sigma_{2,\alpha}} \frac{ \ln \overline{r( \bar s)}  }{R_+(s) (s-k)} \, \d s  \right.\right. \\
	   &\left. \left. \qquad \qquad \qquad \qquad  
	   + \int_{-i\eta_1}^{i \eta_1} \frac{ i\Delta^{(-)}}{R(s)(s-z)}\, \d s
	\right] \right\},
\\	
	&f^{^{(+)}}(k;x,t) = \left( \frac{ k - i\kappa_0}{k + i \kappa_0} \right) \exp \left\{ \frac{R(k)}{2\pi i}\left[  
	   \int_{\Sigma_{1,\alpha}} \frac{ -\ln r(s)  -2 \ln \left( \frac{ i\kappa_0-s}{ i \kappa_0+s} \right) }{R_+(s) (s-k)} \, \d s 
	     \right. \right. \\ 
	     &\left. \left.
	 \qquad \qquad \qquad \qquad  + \int_{\Sigma_{2,\alpha}} \frac{ \ln \overline{r( \bar s)} -2\ln \left( \frac{ i\kappa_0- s }{ i \kappa_0+s} \right) }{R_+(s) (s-k)}  \, \d s 
	   + \int_{-i\eta_1}^{i \eta_1} \frac{ i\Delta^{(+)}}{R(s)(s-z)} \, \d s
	\right] \right\},
	\end{aligned}
\end{gather}
\end{subequations}
and 
\begin{equation}\label{Delta}
	\Delta =  
	\begin{dcases} 
	0, & (x,t) \in \mathcal{S}_L ,\\
	\Delta^{^{(-)}} = -i \left( \int_{0}^{i\eta_1} \frac{ \d s }{R(s)}  \right)^{-1} \left( \int_{\Sigma_{1,\alpha} } \frac{ \log r(s)}{R_+(s)} \, \d s \right),  
	& (x,t) \in \mathcal{S}_M^{^{(-)}} \cup \mathcal{S}_R^{^{(-)}} ,\\
	\Delta^{^{(+)}} =  -i \left( \int_{0}^{i\eta_1} \frac{ \d s }{R(s)}  \right)^{-1} \left( \int_{\Sigma_{1,\alpha} } \frac{ \log r(s)+ 2 \log \left( \frac{i\kappa_0-s }{i \kappa_0+s}\right) }{R_+(s)} \, \d s \right),
	& (x,t) \in \mathcal{S}_M^{^{(+)}} \cup \mathcal{S}_R^{^{(+)}} .
	\end{dcases}
\end{equation}
The following proposition is an immediate consequence of the Sokhotski-Plemelj formula.
\begin{prop}\label{prop:f}
	For any $(x,t) \in \mathcal{S}_M \cup \mathcal{S}_R$ the scalar function $f(k;x,t)$ defined by \eqref{f} satisfies the following properties:
	\begin{enumerate}[label=(\arabic*)]
	  \item\label{f1} $f(k;x,t)$ is meromorphic for $k \in \C \setminus [-i\alpha,i\alpha]$.
	  \item\label{f2} $f(k;x,t)$ satisfies the jump relations
	  \begin{gather}
	  	f_+(k;x,t) f_-(k;x,t) = 
		\begin{dcases}
			\tfrac{1}{r(k)}, & k \in \Sigma_{1,\alpha}, \\
			\overline{r(\bar k)}, & k \in \Sigma_{2,\alpha},
		\end{dcases} \\
		\frac{f_+(k;x,t)}{f_-(k;x,t)} = e^{i\Delta},  \qquad k \in (-i\eta_1, i\eta_1),
	  \end{gather}
	  where $\Delta = \Delta(x,t)$ given by \eqref{Delta} is real-valued.
	  \item\label{f3} $f(k;x,t) = 1+ \bigo{k^{-1}}$ as $k \to \infty$. 
	  \item\label{f4} $f(k;x,t)$ is bounded and nonzero in $k \in \C \setminus [-i\alpha,i\alpha]$. 
	  \item\label{f5} For $(x,t) \in \mathcal{S}_L^{^{(-)}} \cup \mathcal{S}_R^{^{(-)}}$, $f(k;x,t)$ is holomorphic and nonzero in $k \in \C \setminus [-i\alpha,i\alpha]$. 
	  For  $(x,t) \in \mathcal{S}_L^{^{(+)}} \cup \mathcal{S}_R^{^{(+)}}$, $f(k;x,t)$ has a simple zero at $k=i\kappa_0$, a simple pole at $-i\kappa_0$, and is otherwise holomorphic and nonzero for $k \in \C \setminus [-i\alpha,i\alpha]$.
	 {\color{black}
	 \item\label{f5.5} The behavior of $f(k;x,t)$ at each endpoint of the jump contour is determined by the local behavior of $r(k)$. As $k \to \partial \Sigma_{1,\alpha}$, $f(k;x,t)^2 r(k)$  is bounded and nonzero; the same is true of $f(k;x,t)^2 \overline{r(\bar k)}^{-1}$ as $k \to \partial \Sigma_{2,\alpha}$.  
	 }
	  \item\label{f6} For all $k$, $f$ satisfies the symmetries $\overline{ f(\bar k;x,t) } = f(k)^{-1}$ and $f(-k;x,t) = f(k;x,t)^{-1}$. In particular $f(k;x,t)$ is real-valued for any $k \in i \R\setminus [-i\alpha,i \alpha]$. 

	\end{enumerate}
\end{prop}

\begin{remark}
	Locally, the functions $f$ and $\Delta$ depend on $x$ and $t$ through the slowly evolving parameter $x/t$, i.e. $f(k;x,t) = f(k; x/t)$ and $\Delta(x,t) = \Delta(x/t)$. Globally they depend on $x$ and $t$ independently, as their values change along the boundaries between regions $\mathcal{S}_{j}^{^{(\pm)}},\ j \in \{L,M,R\}$. 
\end{remark}

Collecting the properties of the $g$-function \ref{g1}--\ref{g5} and of the function $f$ \ref{f1}--\ref{f6}, the resulting RH problem for ${\bm T}(k;x,t)$ is given by
\begin{RHP}\label{rhp:T}
	Find a $2\times2$ matrix-valued function $\bm{T}(k;x,t)$ with the following properties
	\begin{enumerate}[label=\arabic*.]
	\item $\bm{T}(k;x,t)$ is meromorphic for $k \in \C \setminus [-i\eta_2, i\eta_2]$.
		\item For $k \in i [-\eta_2, \eta_2]$, the boundary values $\bm{T}_\pm(k; x,t) = \bm{T}(k \mp 0; x,t)$ satisfy the jump relation
		\begin{gather}\label{Tjumps}
		  \bm{T}_+(k) = \bm{T}_-(k) \bm{J}_T(k; x,t), \\
		\bm{J}_T(k; x,t) =   \begin{dcases} 
		    \begin{bmatrix} 1 & 0 \\ i r(k) f(k)^2 e^{-2i \varphi(k;x,t)} & 1 \end{bmatrix}, 
		    & k \in  (i\alpha,i\eta_2), \\
		    \begin{bmatrix} r(k)^{-1} f_-(k)^{-2} e^{2i \varphi_-(k; x,t)} & 0 \\ i  & r(k)^{-1} f_+(k)^{-2} e^{2i \varphi_+(k; x,t)} \end{bmatrix}, 
		    & k \in\Sigma_{1,\alpha}, \\
		    e^{i({\Omega}+ \Delta)\sigma_3}, & k \in [-i\eta_1, i\eta_1], \\
		    \begin{bmatrix} \overline{r(\bar k)}^{-1} f_+(k)^2 e^{-2i \varphi_+(k; x,t)} &  i   \\ 0 & \overline{r(\bar k)}^{-1} f_-(k)^2 e^{-2i \varphi_-(k; x,t)}  \end{bmatrix}, 
		    & k \in \Sigma_{2,\alpha}, \\
		    \begin{bmatrix} 1 &  i \overline{r(\bar k)} f(k)^{-2} e^{2i \varphi(k; x,t)} \\ 0 & 1 \end{bmatrix} & k \in ( -i\eta_2, -i\alpha). \\
		  \end{dcases}
		\end{gather}
		\item $\bm{T} (k; x,t)$ has simple poles at $k = \pm i \kappa_0$, with $\kappa_0 > \eta_2$, satisfying
		\begin{itemize}
		 \item For $(x,t) \in \mathcal{S}^{^{(-)}}$:
		  \begin{subequations}\label{Tres}
		  \begin{equation}\label{Tres-}
		  \begin{aligned}
		    \Res_{k = i \kappa_0} \bm{T}(k;x,t) &= \lim_{k \to i \kappa_0} \bm{T}(k;x,t) 
		    \begin{bmatrix} 0 & 0 \\ -i \chi f(i\kappa_0;x,t)^2 e^{-2i \varphi(k;x,t)} & 0 \end{bmatrix}, \\
		    \Res_{k = -i \kappa_0} \bm{T}(k;x,t) &= \lim_{k \to -i \kappa_0} \bm{T}(k;x,t) 
		    \begin{bmatrix} 0 & -i \chi f(i\kappa_0;x,t)^{2} e^{2i \varphi(k;x,t)} \\ 0 & 0 \end{bmatrix}.
		  \end{aligned}
		  \end{equation}
		\item For $(x,t) \in \mathcal{S}^{^{(+)}}$:
		   \begin{equation}\label{Tres+}
		   \begin{aligned}
		     \Res_{k = i \kappa_0} \bm{T}(k;x,t) &= \lim_{k \to i \kappa_0} \bm{T}(k;x,t) 
		     \begin{bmatrix} 0  & i \chi^{-1} f'(i\kappa_0;x,t)^{-2} e^{2i \varphi(k;x,t)} \\ 0 & 0 \end{bmatrix}, \\
		      \Res_{k = -i \kappa_0} \bm{T}(k;x,t) &= \lim_{k \to -i \kappa_0} \bm{T}(k;x,t) 
		     \begin{bmatrix} 0 &  0 \\ i \chi^{-1} f'(i\kappa_0;x,t)^{-2} e^{-2i \varphi(k;x,t)}  & 0 \end{bmatrix}.
		  \end{aligned}
		  \end{equation}
		  \end{subequations}
		\end{itemize}
	\end{enumerate}
\end{RHP}

\subsection{Reduction to model problems}

The transformation  $\bm{X}(k;x,t) \mapsto \bm{T}(k;x,t)$ results in a RH problem which has jumps that are exponentially near identity on the intervals $ (i\alpha, i\eta_2)$ and $(-i\eta_2, -i\alpha)$, oscillatory jumps on the bands $\Sigma_{1,\alpha} \cup \Sigma_{2,\alpha}$ and a constant (in $k$) diagonal jump on the gap interval $(-i\eta_1 , i\eta_1)$ between the two bands. Whenever the bands $\Sigma_{1,\alpha} \cup \Sigma_{2,\alpha}$ are non-empty intervals we introduce one further transformation which `opens lenses'  away from the $\Sigma_{1,\alpha}$ and $\Sigma_{2,\alpha}$, which have the effect of deforming the oscillatory jumps onto new contours where they are exponentially decaying. 

\subsubsection{Quiescent background }
For $(x,t) \in \mathcal{S}_L$ (i.e. $v:= \frac{x}{t}<4\eta_1^2$), we have  that $g \equiv 0$ in this sector of space-time (see Remark~\ref{rem:g.left}). Effectively, $\alpha \equiv \eta_1$ here, so $\Sigma_{1,\alpha} = \Sigma_{2,\alpha} = \emptyset$ and $\Omega = \Delta = 0$. The remaining jumps satisfy the following estimate:

\begin{prop}\label{prop:left.estimate}
For $(x,t) \in \mathcal{S}_L$ the jump matrix $\bm{J}_T(k)$ satisfy the estimate
\begin{equation}
	\left\| \bm{J}_T(k; x,t) - \bm I \right\| = \bigo{e^{-2 t \eta_1 (4 \eta_1^2 - v) } } ,  \qquad (x,t) \in \mathcal{S}_L, \quad k \in \Sigma_1 \cup \Sigma_2,
\end{equation}
where the implicit constant is bounded and independent of $(x,t) \in \mathcal{S}_L$.
\end{prop}

From the uniform estimate above and standard estimates for Cauchy singular integrals, we can  conclude that the solution $\bm{T}(k;x,t)$ of the RH problem~\ref{rhp:T} takes the form 
\begin{equation}
	\bm{T}(k; x,t) = \left[ \bm I + \bigo{e^{-2 t \eta_1 (4 \eta_1^2 - v) } } \right] \bm{X}_\mathrm{sol}(k; x,t),  \quad (x,t) \in \mathcal{S}_L,
\end{equation}
and $\bm{X}_\mathrm{sol}(k; x,t)$ is the solution of the RH problem~\ref{rhp:X} with $r \equiv 0$ and $N=1$, given by \eqref{Xsol}. 
As a consequence, the solution of the original RH problem~\ref{rhp:X} satisfies
\begin{equation}
	\bm{X}(k; x,t) = \left[ \bm I + \bigo{e^{-2 t \eta_1 (4 \eta_1^2 - v) } } \right] \bm{X}_\mathrm{sol}(k; x,t),  \quad \text{as }t \to \infty \text{ with } (x,t) \in \mathcal{S}_L\ .
\end{equation}

\subsubsection{Opening lenses in the support of the soliton gas for $x>4 \eta_1^2 t$: $r(k)$ bounded and nonzero }
For $(x,t) \in \mathcal{S}_M \cup \mathcal{S}_R$ the bands $\Sigma_{1,\alpha}$ and $\Sigma_{2,\alpha}$ are non-empty. 
The oscillatory jumps $\bm{J}_T(k; x,t)$ for $k \in \Sigma_{1,\alpha} \cup \Sigma_{2,\alpha}$ admit the following factorizations:
\begin{multline}\label{Tjump.factor.1}
	\bm{J}_T(k; x,t)   
	= \begin{bmatrix} r(k)^{-1} f_-(k)^{-2} e^{2i \varphi_-(k; x,t)} & 0 \\ i  & r(k)^{-1} f_+(k)^{-2} e^{2i \varphi_+(k; x,t)} \end{bmatrix}  \\
	\qquad = \begin{bmatrix} 1 & -i r(k)^{-1} f_-(k)^{-2} e^{2i \varphi_-(k; x,t)} \\ 0 & 1 \end{bmatrix} 
	   \begin{bmatrix} 0 & i \\ i & 0 \end{bmatrix}
	   \begin{bmatrix} 1 &  -i r(k)^{-1} f_+(k)^{-2} e^{2i \varphi_+(k; x,t)}  \\ 0 & 1 \end{bmatrix}, \quad k \in \Sigma_{1,\alpha},
\end{multline}
\begin{multline}\label{Tjump.factor.2}
	\bm{J}_T(k,x,t)   
	=  \begin{bmatrix} \overline{r(\bar k)}^{-1} f_+(k)^2 e^{-2i \varphi_+(k; x,t)} &  i   \\ 0 & \overline{r(\bar k)}^{-1} f_-(k)^2 e^{-2i \varphi_-(k; x,t)}  \end{bmatrix} \\
	\qquad = \begin{bmatrix} 1 & 0 \\  -i \overline{r(\bar k)}^{-1} f_-(k)^2 e^{-2i \varphi_-(k; x,t)} & 1 \end{bmatrix} 
	   \begin{bmatrix} 0 & i \\ i & 0 \end{bmatrix}
	   \begin{bmatrix} 1 & 0   \\ -i \overline{r(\bar k)}^{-1} f_+(k)^2 e^{-2i \varphi_+(k; x,t)} & 1 \end{bmatrix}, \quad k \in \Sigma_{2,\alpha}.
\end{multline}

\begin{figure}
\begin{center}
\begin{tikzpicture}
\begin{scope}[local bounding box=scope1]
\def\d{1}
\def\h{2};
\coordinate (ieta1) at (0, \d);
\coordinate (ieta2) at ($(ieta1)+(0,\h)$);
\def\angle{40};
\def\xo{ -cot(\angle) * \h/2  };
\def\radius{ \h  / 2 /  sin(\angle) };

\draw[dashed] (ieta1) to (ieta2);
\fill[cyan, fill opacity = 0.4, domain=-\angle:\angle] plot ({ \xo + (\radius)*cos(\x)}, {\d +\h/2 + \radius * sin(\x)});
\draw[thick, domain=-\angle:\angle] plot ({\xo + \radius*cos(\x)}, {\d+\h/2 + \radius * sin(\x)});
\draw [very thick, domain = -90:-90+\angle] plot({0.2*cos(\x)},{ \d+\h + 0.2*sin(\x)});
\draw [very thick, domain = 90:90-\angle] plot({0.2*cos(\x)}, {\d + 0.2*sin(\x)});
\node [label = { [label distance = -2.5, above, xshift = -0.2] \tiny$i\eta_2$}] at (0, \d+\h) {};
\node [label = {[below, xshift = 6]\tiny$\beta$}] at (0.05,\d+\h){};
\node [label = {[below]\tiny $i\eta_1$}] at (0,\d){};
\node [label = {[below, yshift = 7, xshift = 6]\tiny$\beta$}] at (0.05,\d){};

\node at ({\radius+\xo+0.4}, 2) { \tiny $\Omega_{\beta,r}$};
\draw[->] ({\radius+\xo+0.1}, {\d+\h/2}) -- ({\radius+\xo-0.15}, {\d+\h/2});

\fill[cyan, fill opacity = 0.4, domain=180-\angle:180+\angle] plot ({ -\xo + (\radius)*cos(\x)}, {\d +\h/2 + \radius * sin(\x)});
\draw[thick, domain=180-\angle:180+\angle] plot ({-\xo + \radius*cos(\x)}, {\d+\h/2 + \radius * sin(\x)});

\node at ({-\radius-\xo-0.4}, 2){\tiny $\Omega_{\beta,l}$};
\draw [->] ({-\radius-\xo-0.1}, {\d+\h/2}) -- ({-\radius-\xo+0.15}, {\d+\h/2});
\path []  (0,0.4) -- (.7,.4) node[midway,below,font=\small]{};

\node at (0,-1) {$(a)$ $r$ bounded, nonzero at $\eta_1, \eta_2$};

\end{scope}
\begin{scope}[shift={($(scope1)+(7cm, -1.03cm)$)}]
\def\d{1};
\def\h{2};
\coordinate (ieta1) at (0, \d);
\coordinate (ieta2) at ($(ieta1)+(0,\h)$);
\def\rad{0.5};

\draw [dashed] ($(ieta1)-(0,\rad)$) --($(ieta2)+(0,\rad)$);
\draw [very thick, fill=cyan, fill opacity = 0.4] ($(ieta1)+(\rad,0)$) -- ($(ieta2)+(\rad,0)$) arc (0:180:\rad) -- ($(ieta1)-(\rad,0)$) arc(180:360:\rad); 
\draw [very thick] (ieta1) -- ($(ieta1)-(0,\rad)$);
\draw [very thick] (ieta2) -- ($(ieta2)+(0,\rad)$);  
\draw[|-|] ($ (ieta1) -(0,1.5*\rad)$) --($ (ieta1) -(\rad,1.5*\rad)$) node[midway,below,font=\tiny]{$h$};

\node at (ieta1) [circle,fill,inner sep=1pt, label={[right, xshift=-.1,yshift=-.9]:\tiny$i\eta_1$}] {};
\node at (ieta2) [circle,fill,inner sep=1pt, label={[right, xshift=-.1,yshift=-.9]:\tiny$i\eta_2$}] {};

\draw[<-] (\rad/2,\d+\h/2)--(1.5*\rad,\d+\h/2) node[right,] {\tiny$\Omega_{h,r}$};
\draw[<-] (-\rad/2,\d+\h/2)--(-1.5*\rad,\d+\h/2) node[left,] {\tiny$\Omega_{h,l}$};

\node at (0,-1) {$(b)$ $r(k) (k-i \eta_k)^{\pm1/2} = \bigo{1}$};
\end{scope}
\end{tikzpicture}

\end{center}
\caption{The domain of the analytic extension of $r(k)$ away from $\Sigma_1$ depends on the behavior of $r(k)$ near the endpoints of $\Sigma_1$.    }
\label{fig:r}
\end{figure}
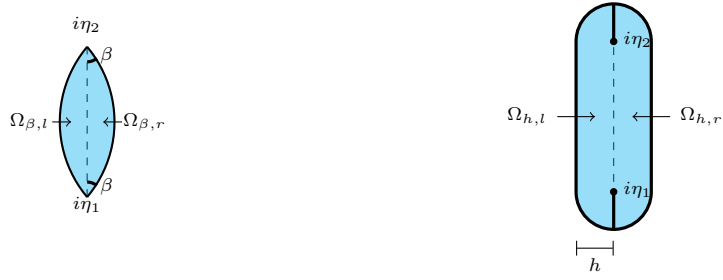

To deform these jumps off the imaginary axis we make the following assumption on the function $r(k):$
\begin{assumption}\label{Assumption:r}
Assume that the function $r(k)$ admits an analytic continuation off the imaginary axis:
\begin{align}
	&\hat{r}(k) \text{ analytic in } \Omega_{\beta,r}\cup \Omega_{\beta,l},
	&&\hat{r}(k)\! \big\vert_{\scriptstyle k \in [i\eta_1, i\eta_2]} = r(k),
\end{align}
where $\Omega_{\beta,r}$ is defined as
\begin{gather}\label{Omega_beta,rl}
\Omega_{\beta,r} := \left\{k \in \mathbb C \, :\, \ \Re k\geq 0 \text{ and } \left|\Im k-\frac{\eta_1+\eta_2}{2}\right|\leq \sqrt{\frac{(\eta_2-\eta_1)^2}{4}-\frac{\cos\beta}{\sin\beta}(\eta_2-\eta_1)\Re k - (\Re k)^2} \right\},
\end{gather}
with some $0<\beta<\frac{\pi}{2}$ and where $\Omega_{\beta, l}$ is defined by symmetry, $\Omega_{\beta, l}=\left\{k: -\ol k\in\Omega_{\beta,r}\right\}$ (see \figurename~\ref{fig:r}).
\end{assumption}

Moreover, this extension should preserve the symmetries of the RH problem, i.e, $\overline{\hat{r}(-\ol k)} =\hat{r}(k)$. We therefore open lens by introducing a pair of contours, together labeled $\mathcal{C}_1$, both starting at $-i\eta_1$ and ending at $i\alpha$ to the left and right of $\Sigma_{1,\alpha}$ respectively, such that $\mathcal{C}_1$ lies entirely in the domain $\Omega_{\beta,r}\cup \Omega_{\beta,l}$ where $\hat{r}(k)$ is analytic. We also introduce lens contours $\mathcal{C}_2$ on $\Sigma_{2,\alpha}$ by symmetry. See \figurename~\ref{fig:lenses}.

Using the factorizations in \eqref{Tjump.factor.1}--\eqref{Tjump.factor.2}, we then define
\begin{equation}\label{S}
	\bm{S}(k; x,t) = 
	\begin{dcases}
	\bm{T}(k;x,t) \begin{bmatrix} 1 & -i \hat{r}(k)^{-1} f(k)^{-2} e^{2i \varphi(k;x,t)} \\ 0 & 1 \end{bmatrix}, 
	& k\in \text{lens right of $\Sigma_{1,\alpha}$}, \\ 
	\bm{T}(k;x,t) \begin{bmatrix} 1 & i \hat{r}(k)^{-1} f(k)^{-2} e^{2i \varphi(k;x,t)} \\ 0 & 1 \end{bmatrix},
	& k\in \text{lens left of $\Sigma_{1,\alpha}$}, \\ 
	\bm{T}(k;x,t) \begin{bmatrix} 1 & 0 \\  \frac{-i f(k)^2 e^{-2i \varphi(k;x,t)}}{\overline{\hat r(\ol k)}} & 1 \end{bmatrix},
	& k\in \text{lens right of $\Sigma_{2,\alpha}$}, \\ 
	\bm{T}(k;x,t) \begin{bmatrix} 1 & 0 \\  \frac{i f(k)^2 e^{-2i \varphi(k;x,t)}}{\overline{\hat r(\ol k)}} & 1 \end{bmatrix},
	& k\in \text{lens left of $\Sigma_{2,\alpha}$}, \\
	\bm{T}(k; x,t), & \text{elsewhere .}
	\end{dcases}
\end{equation}
It follows easily that $\bm{S}$ satisfies the following RH problem 
\begin{RHP}\label{rhp:S}
	Find a $2\times2$ matrix-valued function $\bm{S}(k; x,t)$ such that
	\begin{enumerate}[label=\arabic*.]
		\item $\bm{S}(k; x,t)$ is meromorphic for $k \in\C \setminus \Gamma_S$, $\Gamma_S =  [-i\eta_2,i\eta_2] \cup \mathcal{C}_1 \cup \mathcal{C}_2$.
		\item $\bm{S}(k; x,t) = \bm I + \bigo{k^{-1}}$ as $k \to \infty$. 
		\item For $k \in \Gamma_S$ the boundary values $\bm{S}_\pm(k;x,t)$ satisfy the jump relation $\bm{S}_+(k; x,t) = \bm{S}_-(k; x,t) \bm{J}_S(k; x,t)$, where the values of $\bm{J}_S(k; x,t)$ are shown for $k \in \Gamma_S$ in \figurename~\ref{fig:lenses}.
		\item $\bm{S}(k; x,t)$ has simple poles at $k = \pm i \kappa_0$ and no other poles. The poles satisfy the same residue conditions \eqref{Tres} as $\bm{T}(k; x,t)$.
	\end{enumerate}
\end{RHP}

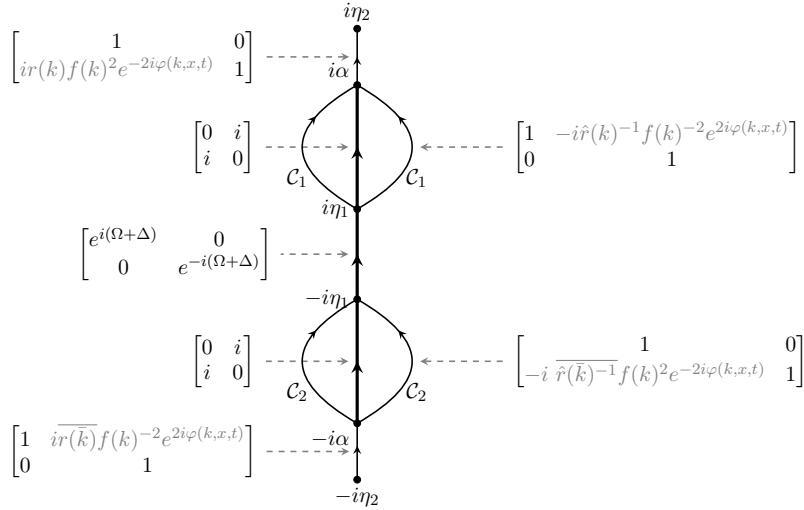
\begin{figure}
\centering
\scalebox{.75}{
\begin{tikzpicture}[>=stealth]op
\path (0,0) coordinate (O);

\coordinate (eta1) at (0,.8);    \coordinate (eta1c) at ($-1*(eta1)$);
\coordinate (eta2) at (0,4);       \coordinate (eta2c) at ($-1*(eta2)$);
\coordinate (alpha) at (0,3);   \coordinate (alphac) at ($-1*(alpha)$);

\draw[->- = .5,thick] (alpha)--(eta2) 
  node[pos=0.5, pin={[pin distance=15mm, pin edge={<-,dashed,thick}]180: $\begin{bmatrix} 1 & 0 \\ {\color{gray}i r(k) f(k)^2 e^{-2i \varphi(k,x,t)} } & 1 \end{bmatrix} $}] {}; 

\draw[->- = .5,ultra thick] (eta1)--(alpha) 
  node[pos=0.5, pin={[pin distance=15mm, pin edge={<-,dashed,thick}]180: $\begin{bmatrix} 0 & i \\ i & 0 \end{bmatrix}$}] {};  

\draw[->- = .5,ultra thick] (alphac)--(eta1c) 
  node[pos=0.5, pin={[pin distance=15mm, pin edge={<-,dashed,thick}]180: $\begin{bmatrix} 0 & i \\ i & 0 \end{bmatrix}$}] {};

\draw[->- = .5, ultra thick] (eta1c)--(eta1) 
  node[pos=0.5, pin={[pin distance=12mm, pin edge={<-,dashed,thick}]180: $\begin{bmatrix} e^{i({\Omega}+\Delta)} & 0 \\ 0 & e^{-i({\Omega}+\Delta)} \end{bmatrix} $}] {}; 

\draw[->- = .6,thick] (eta2c)--(alphac)
   node[pos=0.5, pin={[pin distance=15mm, pin edge={<-,dashed,thick}]180: $\begin{bmatrix} 1 &  {\color{gray}i \overline{r(\bar k)} f(k)^{-2} e^{2i \varphi(k,x,t)}} \\ 0 & 1 \end{bmatrix}$}] {};

\draw[->- = .7,thick] (eta1) .. controls + (30:1.5cm) and + (-30:1.5cm) .. (alpha) 
  node[pos=.25, right] {$\mathcal{C}_1$}
  node[pos=0.5, pin={[pin distance=15mm, pin edge={<-,dashed,thick}]0: $ \begin{bmatrix} 1 & {\color{gray} -i \hat{r}(k)^{-1} f(k)^{-2} e^{2i \varphi(k,x,t)} } \\ 0 & 1 \end{bmatrix}$ }] {};  
\draw[->- = .7,thick] (eta1) .. controls + (150:1.5cm) and + (-150:1.5cm) .. (alpha)
  node[pos=.25, left] {$\mathcal{C}_1$};

\draw[->- = .7,thick] (alphac) .. controls + (30:1.5cm) and + (-30:1.5cm) .. (eta1c) 
  node[pos=.25, right] {$\mathcal{C}_2$}
  node[pos=0.5, pin={[pin distance=15mm, pin edge={<-,dashed,thick}]0: $ \begin{bmatrix} 1 & 0   \\ {\color{gray}-i \ \overline{\hat{r}(\bar k)^{-1}} f(k)^2 e^{-2i \varphi(k,x,t)} } & 1 \end{bmatrix}$ }] {}; 
\draw[->- = .7,thick] (alphac) .. controls + (150:1.5cm) and + (-150:1.5cm) .. (eta1c) 
  node[pos=.25, left] {$\mathcal{C}_2$};

\draw[fill] (eta2) circle [radius=0.06] node[above] {$i\eta_2$};
\draw[fill] (eta1) circle [radius=0.06] node[left] {$i\eta_1$};
\draw[fill] (alpha) circle [radius=0.06] node[above left] {$i\alpha$};
\draw[fill] (eta2c) circle [radius=0.06] node[below] {$-i\eta_2$};
\draw[fill] (eta1c) circle [radius=0.06] node[left] {$-i\eta_1$};
\draw[fill] (alphac) circle [radius=0.06] node[below left] {$-i\alpha$};

\end{tikzpicture}
}
\caption{ The system of contours $\Gamma_S$ defining the lens opening transformation $\bm{T} \mapsto \bm{S}$ and the resulting jump matrix $\bm{J}_S$ on these contours. The entries shown in gray in the jump matrices 
are all exponentially small.}
\label{fig:lenses}
\end{figure}

We finally need the following lemma, which will guarantee that the off-diagonal entries in the jumps along the lenses $\mathcal C_1 \cup \mathcal C_2$ are exponentially small:
\begin{lemma}\label{lem:lenses}
For any $(x,t)$ with $x> 4\eta_1^2 t$ the following inequalities are satisfied
\begin{equation}\label{lens inequalities}
  \begin{aligned}
     \Im \varphi(k;x,t) < -c t,& \quad&&k \in K \text{ a compact subset of } (i\alpha, i\eta_2], \\
     \Im \varphi(k;x,t) > c t,&  \quad &&k \in \hat K \text{ a compact subset of } \mathcal{C}_1 \setminus \{ i \eta_1, i \alpha \},   
  \end{aligned}
\end{equation}
for some constant $c\in \R_+$. By the symmetry $\varphi(\bar k; x,t) = \overline{\varphi(k; x,t)}$, the reverse inequalities hold on compact subsets of $\mathcal{C}_2 \setminus \{-i \eta_1, -i\alpha\}$ and $[-i\eta_2, -i\alpha)$. 
\end{lemma}

\begin{proof}
The representation of  $\varphi(k; x,t)$ in \eqref{varphi.3b} and the fact that $\oint_{\mathcal A} \d \varphi = 2\int_{-i\eta_1}^{i\eta_1} \varphi' \, \d k = 0$ imply that 
\begin{equation}
	\varphi'(k; x,t)= \frac{12 t (k^2+\alpha^2) (k^2+ p^2) }{R(k)}
\end{equation}
for some $p = p(x,t) \in (0, \eta_1)$. It follows immediately that 
\begin{align*} 
	\varphi'(k) < 0,& && k \in (i\alpha, i\infty), \\
	i \varphi'_+(k) = -i \varphi'_-(k) > 0,& && k \in \Sigma_{1,\alpha}.
\end{align*}
These two conditions imply the first and second inequalities in \eqref{lens inequalities} respectively. 
\end{proof}

{\color{black}
\subsubsection{  Opening lenses in the support of the soliton gas for $x>4 \eta_1^2 t$: square root behavior in $r(k)$ at endpoints} 
Thus far we have restricted our attention to the situation in which $r(k)$ is strictly positive and bounded on $\Sigma_1$. It's reasonable, however,  to admit mild zeros or singular behaviour at the endpoints $i \eta_1$ and $i\eta_2$. Specifically, one can consider $r(k) = |k - i \eta_j|^{\beta} \tilde r(k)$ 
for $|\beta| <1$ ($j=1,2$), and the analysis goes through essentially unchanged except that the local Bessel parametrices at the fixed endpoints (cf. Section~\ref{sect_loc_par_eta1}) have to be slightly adjusted. 

The special case when $r(k) = |z-i \eta_j|^{\pm1/2} \tilde r(k)$ for $\tilde r$ locally bounded and non-zero is of particular interest. 
In this setting, the lens opening factorizations can be slightly adjusted so that local parametrices near the fixed end points $\pm i \eta_1$ and $\alpha = i \eta_2$ for $x/t > v_2$ are not needed. This has the effect that for $x/t> v_2$, the outer model is uniformly accurate and, as a result, the error bounds improve from polynomial to exponential decay in $t$.  In this case the potential falls into a class of potentials considered in 
\cite{Teschls}.

First, we need to modify Assumption~\ref{Assumption:r} on the analytic extension of $r(k)$ away from $\Sigma_1$. 
\begin{assumption}\label{Assumption:r2}
When $r(k)|k-i\eta_j|^{\pm 1/2}$ is bounded and nonzero on $\Sigma_1$, we assume that $r$ admits an analytic continuation off of $\Sigma_1$:
\begin{gather}
	\hat{r}(k) \text{ analytic in } \Omega_{h,r}\cup \Omega_{h,l},
	\qquad \qquad
	\hat{r}(k)\! \big\vert_{\scriptstyle k \in [i\eta_1, i\eta_2]} = r(k), \\
	\hat r_+(k) + \hat r_-(k) = 0, \qquad k \in [i\eta_2, i\eta_2+h] \cup [i\eta_1-h, i\eta_1]
	\label{rhat.jump}
\end{gather}
where $\Omega_{h,r}$ is defined as
\begin{gather}\label{Omega_h,rl}
\Omega_{h,r} := \left\{k \in \mathbb C \, :\, \ \Re k \in (0,h] \text{ and } \eta_1-\sqrt{h^2-\Re k^2} \leq \Im k \leq \eta_2+ \sqrt{h^2-\Re k^2} \right\},
\end{gather}
with some $0 < h < \eta_1$ and where $\Omega_{h, l}$ is defined by symmetry, $\Omega_{h, l}=\left\{k: -\ol k\in\Omega_{h,r}\right\}$ (see \figurename~\ref{fig:r}).
\end{assumption}
\begin{figure}
\centering
\scalebox{.6}{
\begin{tikzpicture}[>=stealth]op
\path (0,0) coordinate (O);

\coordinate (eta1) at (0,.8);    \coordinate (eta1c) at ($-1*(eta1)$);
\coordinate (eta2) at (0,4);       \coordinate (eta2c) at ($-1*(eta2)$);
\coordinate (alpha) at (0,3);   \coordinate (alphac) at ($-1*(alpha)$);

\coordinate (eta1_low) at (0,.3);    \coordinate (eta1_lowc) at ($-1*(eta1_low)$);
\coordinate (eta2_up) at (0,4.7);       \coordinate (eta2_upc) at ($-1*(eta2_up)$);
\coordinate (ic) at (0,3.6);         \coordinate (icc) at ($-1*(ic)$);


\draw[->- = .5, ultra thick] (eta1)--(eta2) 
  node[pos=0.5, pin={[pin distance=15mm, pin edge={<-,dashed,thick}]180: $\begin{bmatrix} 0 & i \\ i & 0 \end{bmatrix}$}] {};  

\draw[->- = .55, ultra thick] (eta2c)--(eta1c) 
  node[pos=0.5, pin={[pin distance=12mm, pin edge={<-,dashed,thick}]180: $\begin{bmatrix} 0 & i \\ i & 0 \end{bmatrix}$}] {}; 

\draw[->- = .55, ultra thick] (eta1c)--(eta1) 
  node[pos=0.5, pin={[pin distance=12mm, pin edge={<-,dashed,thick}]180: $\begin{bmatrix} e^{i({\Omega}+\Delta)} & 0 \\ 0 & e^{-i({\Omega}+\Delta)} \end{bmatrix} $}] {};

\draw[->- = .7, thick,gray] (eta1_low) .. controls + (0:1.5cm) and + (0:1.5cm) .. (eta2_up) 
  node[pos=.75, right] {$\mathcal{C}_1$}
  node[pos=0.5, pin={[pin distance=15mm, pin edge={<-,dashed,thick}]0: $ \begin{bmatrix} 1 & {\color{gray} -i \hat{r}(k)^{-1} f(k)^{-2} e^{2i \varphi(k,x,t)} } \\ 0 & 1 \end{bmatrix}$ }] {};
  
\draw[->- = .7,thick,gray] (eta1_low) .. controls + (180:1.5cm) and + (180:1.5cm) .. (eta2_up)
  node[pos=.75, left] {$\mathcal{C}_1$};

\draw[->- = .3,thick,gray] (eta2_upc) .. controls + (0:1.5cm) and + (0:1.5cm) .. (eta1_lowc) 
  node[pos=.25, right] {$\mathcal{C}_2$}
  node[pos=0.5, pin={[pin distance=15mm, pin edge={<-,dashed,thick}]0: $ \begin{bmatrix} 1 & 0   \\ {\color{gray}-i \ \overline{\hat{r}(\bar k)^{-1}} f(k)^2 e^{-2i \varphi(k,x,t)} } & 1 \end{bmatrix}$ }] {}; 

\draw[->- = .3,thick,gray] (eta2_upc) .. controls + (180:1.5cm) and + (180:1.5cm) .. (eta1_lowc) 
  node[pos=.25, left] {$\mathcal{C}_2$};

\draw[fill] (eta2) circle [radius=0.06] node[above] {$i\eta_2$};
\draw[fill] (eta1) circle [radius=0.06] node[left] {$i\eta_1$};
\draw[fill] (eta2c) circle [radius=0.06] node[below] {$-i\eta_2$};
\draw[fill] (eta1c) circle [radius=0.06] node[left] {$-i\eta_1$};
\draw[fill, gray] (eta2_up) circle [radius=0.06] node {};
\draw[fill, gray] (eta2_upc) circle [radius=0.06] node {};
\draw[fill, gray] (eta1_low) circle [radius=0.06] node {};
\draw[fill, gray] (eta1_lowc) circle [radius=0.06] node {};

\end{tikzpicture}
}
\caption{ The modified system of contours $\Gamma_S$ defining the lens opening transformation $\bm{T} \mapsto \bm{S}$ and the resulting jump matrix $\bm{J}_S$ for $(x,t) \in \mathcal{S}_R$ (so $\alpha = \eta_2$) when ${r(k) |k-i \eta_k|^{\pm1/2} = \bigo{1}}$.
The entries shown in gray in the jump matrices are all exponentially small.}
\label{fig:lenses.hard.edges}	
\end{figure}
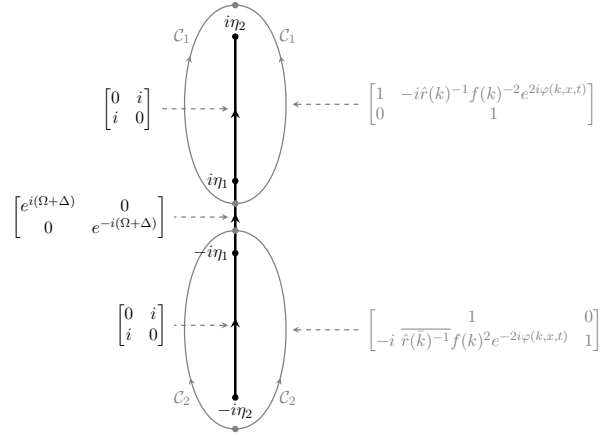
In the modified setting, the factorizations of the jump matrix $\bm{J}_T(k;x,t)$ defined by \eqref{Tjump.factor.1}--\eqref{Tjump.factor.2} remain valid. Note in particular, that property \ref{f5.5} in Proposition~\ref{prop:f} guarantees that the off-diagonal entries of each triangular factor are bounded at the endpoints $\pm i \eta_j$, $j=1,2$.  As before, we open a lens $\mathcal{C}_1$  away from $\Sigma_1$ (and its symmetric pair $\mathcal{C}_2$ away from $\Sigma_2$) with the modification that the lens detaches from $i \eta_1$ and fully encloses the endpoint. What happens at $i\alpha$ depends on whether it is modulating or fixed at $\alpha=i\eta_2$. 
When $(x,t) \in \mathcal{S}_R$, $\alpha = \eta_2$ and the lens $\mathcal{C}_1$ detaches at both endpoints and becomes a loop (see Figure~\ref{fig:lenses.hard.edges}). 
For $(x,t) \in \mathcal{S}_M$, $\alpha \in(\eta_1,  \eta_2)$ and the lens returns to $\alpha$ as usual. We then pick a fixed point in the interval $[i \alpha, \eta_2]$ and open a lens $\mathcal{C}_3$ from this point enclosing $i\eta_2$. Define $\mathcal{C}_4$ by symmetry (see \figurename~\ref{fig:lenses.mixed.edges}). 
For $(x,t) \in \mathcal{S}_R$ we use the \eqref{S} to again define the transformation $\bm{T} \mapsto \bm{S}$ modulo the change in the shape of the lenses. For $(x,t) \in \mathcal{S}_M$ we again use \eqref{S} but append the following extra factorizations inside $\mathcal{C}_3$ and $\mathcal{C}_4$. 
\begin{equation}
	\bm{S}(k;x,t) = \begin{dcases}
		\bm{T}(k;x,t) \begin{bmatrix} 1 & 0 \\ i \hat r(k) f(k)^2 e^{-2i\varphi(k;x,t)} & 1 \end{bmatrix}  & k \in \inside(\mathcal{C}_3) \cap \{\Re k <0 \}\ , \\
		\bm{T}(k;x,t) \begin{bmatrix} 1 & 0 \\ -i \hat r(k) f(k)^2 e^{-2i\varphi(k;x,t)} & 1 \end{bmatrix} & k \in \inside(\mathcal{C}_3) \cap \{\Re k >0 \} \ , \\
		\bm{T}(k;x,t) \begin{bmatrix} 1 & i\overline{ \hat r(\bar k) }f(k)^{-2} e^{2i\varphi(k;x,t)} \\ 0 & 1 \end{bmatrix} & k \in \inside(\mathcal{C}_4) \cap \{\Re k > 0 \} \ , \\
		\bm{T}(k;x,t) \begin{bmatrix} 1 & -i \overline{\hat r(\bar k)} f(k)^{-2} e^{2i\varphi(k;x,t)} \\ 0 & 1 \end{bmatrix}  & k \in \inside(\mathcal{C}_4) \cap \{\Re k < 0 \} \ , \\
		\text{defined by \eqref{S}} & \text{elsewhere.}
	\end{dcases}
\end{equation}
The key observation is that, due to the half-integer power behaviour of $r(k)$ at the endpoint we can define the extension $\hat r$ to satisfy \eqref{rhat.jump} so that the transformation $\bm{T} \mapsto \bm{S}$ does not introduce new jumps on the intervals $[i( \eta_1-\epsilon), i\eta_1]$ and $[i\eta_2, i(\eta_2+\epsilon)]$ . For example, when $(x,t) \in \mathcal{S}_R$ we have 
\begin{equation*}
	\bm{S}_-^{-1}(k) \bm{S}_+(k) = \begin{bmatrix} 1 & i ( \hat r_-(k)^{-1} + \hat r_+(k)^{-1}) f(k)^{-2} e^{2i \varphi(k;x,t)} \\ 0 & 1 \end{bmatrix} = \bm{I} 
\end{equation*}
for $k \in [i\eta_2, i(\eta_2+\epsilon)]$, 
\begin{equation*}
	\bm{S}_-^{-1}(k) \bm{S}_+(k) = \begin{bmatrix} e^{i(\Omega+\Delta)} & i \left( \hat r_+(k) +r_-(k) \right) f_+(k)^{-2} e^{i\varphi_+(k)} e^{i(\Omega+\Delta)} \\ 0 & e^{-i(\Omega+\Delta)} \end{bmatrix} = e^{i(\Omega+\Delta)\sigma_3} 
\end{equation*}
for $k \in [i(\eta_1-\epsilon), i\eta_1]$. The calculation to check the jump on the other boundaries are similar and left to the reader.

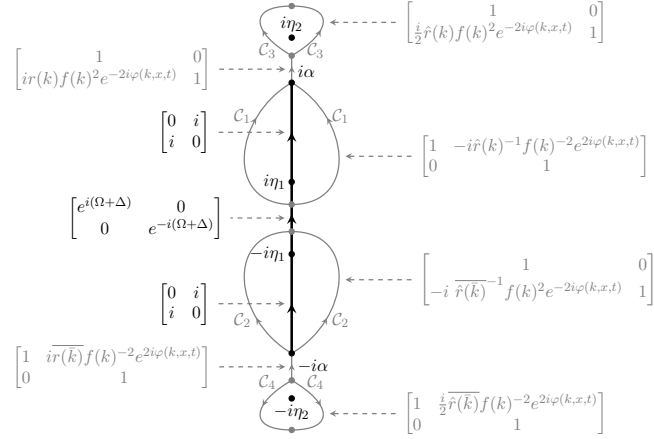
\begin{figure}
\centering
\scalebox{.6}{
\begin{tikzpicture}[>=stealth]op
\path (0,0) coordinate (O);

\coordinate (eta1) at (0,.8);    \coordinate (eta1c) at ($-1*(eta1)$);
\coordinate (eta2) at (0,4);       \coordinate (eta2c) at ($-1*(eta2)$);
\coordinate (alpha) at (0,3);   \coordinate (alphac) at ($-1*(alpha)$);

\coordinate (eta1_low) at (0,.3);    \coordinate (eta1_lowc) at ($-1*(eta1_low)$);
\coordinate (eta2_up) at (0,4.7);       \coordinate (eta2_upc) at ($-1*(eta2_up)$);
\coordinate (ic) at (0,3.6);         \coordinate (icc) at ($-1*(ic)$);

\draw[->- = .6,thick,gray] (alpha)--(ic) 
  node[pos=0.5, pin={[pin distance=15mm, pin edge={<-,dashed,thick}]180: $\begin{bmatrix} 1 & 0 \\ {\color{gray}i r(k) f(k)^2 e^{-2i \varphi(k,x,t)} } & 1 \end{bmatrix} $}] {}; 

\draw[->- = .5, ultra thick] (eta1)--(alpha) 
  node[pos=0.5, pin={[pin distance=15mm, pin edge={<-,dashed,thick}]180: $\begin{bmatrix} 0 & i \\ i & 0 \end{bmatrix}$}] {};  

\draw[->- = .5, ultra thick] (alphac)--(eta1c) 
  node[pos=0.5, pin={[pin distance=15mm, pin edge={<-,dashed,thick}]180: $\begin{bmatrix} 0 & i \\ i & 0 \end{bmatrix}$}] {};

\draw[->- = .55, ultra thick] (eta1c)--(eta1) 
  node[pos=0.5, pin={[pin distance=12mm, pin edge={<-,dashed,thick}]180: $\begin{bmatrix} e^{i({\Omega}+\Delta)} & 0 \\ 0 & e^{-i({\Omega}+\Delta)} \end{bmatrix} $}] {}; 

\draw[->- = .6, thick,gray] (icc)--(alphac)
   node[pos=0.5, pin={[pin distance=15mm, pin edge={<-,dashed,thick}]180: $\begin{bmatrix} 1 &  {\color{gray}i \overline{r(\bar k)} f(k)^{-2} e^{2i \varphi(k,x,t)}} \\ 0 & 1 \end{bmatrix}$}] {};

\draw[->- = .7, thick,gray] (eta1_low) .. controls + (0:1.5cm) and + (-30:1.5cm) .. (alpha) 
  node[pos=.75, right] {$\mathcal{C}_1$}
  node[pos=0.5, pin={[pin distance=15mm, pin edge={<-,dashed,thick}]0: $ \begin{bmatrix} 1 & {\color{gray} -i \hat{r}(k)^{-1} f(k)^{-2} e^{2i \varphi(k,x,t)} } \\ 0 & 1 \end{bmatrix}$ }] {};  
\draw[->- = .7,thick,gray] (eta1_low) .. controls + (180:1.5cm) and + (-150:1.5cm) .. (alpha)
  node[pos=.75, left] {$\mathcal{C}_1$};

\draw[->- = .3,thick,gray] (alphac) .. controls + (30:1.5cm) and + (0:1.5cm) .. (eta1_lowc) 
  node[pos=.25, right] {$\mathcal{C}_2$}
  node[pos=0.5, pin={[pin distance=15mm, pin edge={<-,dashed,thick}]0: $ \begin{bmatrix} 1 & 0   \\ {\color{gray}-i \ \overline{\hat{r}(\bar k)}^{-1} f(k)^2 e^{-2i \varphi(k,x,t)} } & 1 \end{bmatrix}$ }] {}; 
\draw[->- = .3,thick,gray] (alphac) .. controls + (150:1.5cm) and + (180:1.5cm) .. (eta1_lowc) 
  node[pos=.25, left] {$\mathcal{C}_2$};

\draw[->- = .45, thick,gray] (ic) .. controls + (30:1cm) and + (0:1cm) .. (eta2_up) node[pos=0.5, pin={[pin distance=15mm, pin edge={<-,dashed,thick}]0: $\begin{bmatrix} 1 & 0 \\ {\color{gray}\tfrac{i}{2} \hat{r}(k) f(k)^2 e^{-2i \varphi(k,x,t)} } & 1 \end{bmatrix} $}] {}
node[pos=.27, below ] {$\mathcal{C}_3$};

\draw[->-= .45, thick,gray] (ic) .. controls + (150:1cm) and + (180:1cm) .. (eta2_up)
node[pos=.27, below] {$\mathcal{C}_3$};;

\draw[->- = .45, thick,gray] (icc) .. controls + (-150:1cm) and + (180:1cm) .. (eta2_upc)
node[pos=.27, above, yshift=2 ] {$\mathcal{C}_4$};;
\draw[->- = .45, thick,gray] (icc) .. controls + (-30:1cm) and  + (0:1cm) .. (eta2_upc)
node[pos=0.5, pin={[pin distance=15mm, pin edge={<-,dashed,thick}]0: $\begin{bmatrix} 1 &  {\color{gray}\tfrac{i}{2} \overline{\hat{r}(\bar k)} f(k)^{-2} e^{2i \varphi(k,x,t)} }\\ 0 & 1 \end{bmatrix} $}] {}
node[pos=.27, above, yshift=2 ] {$\mathcal{C}_4$};

\draw[fill] (eta2) circle [radius=0.06] node[above] {$i\eta_2$};
\draw[fill] (eta1) circle [radius=0.06] node[left] {$i\eta_1$};
\draw[fill] (alpha) circle [radius=0.06] node[above right] {$i\alpha$};
\draw[fill] (eta2c) circle [radius=0.06] node[below] {$-i\eta_2$};
\draw[fill] (eta1c) circle [radius=0.06] node[left] {$-i\eta_1$};
\draw[fill] (alphac) circle [radius=0.06] node[below right] {$-i\alpha$};
\draw[fill, gray] (eta2_up) circle [radius=0.06] node {};
\draw[fill, gray] (eta2_upc) circle [radius=0.06] node {};
\draw[fill, gray] (eta1_low) circle [radius=0.06] node {};
\draw[fill, gray] (eta1_lowc) circle [radius=0.06] node {};
\draw[fill, gray] (ic) circle [radius=0.06] node {};
\draw[fill, gray] (icc) circle [radius=0.06] node {};
\end{tikzpicture}
}
\caption{ The modified system of contours $\Gamma_S$ defining the lens opening transformation $\bm{T} \mapsto \bm{S}$ and the resulting jump matrix $\bm{J}_S$ for $(x,t) \in \mathcal{S}_M$ when $r(k) |k-i \eta_j|^{\pm1/2} = \bigo{1}$.
The entries shown in gray in the jump matrices are all exponentially small.}
\label{fig:lenses.mixed.edges}	
\end{figure}
}

\section{The model problems  and proof of Theorem ~\ref{theorem1.3}}
\label{sec-model}

In this section we complete the proof of Theorem~\ref{theorem1.3}.  
In Theorem~\ref{thm:g}  we derived the behaviour of the endpoint $\alpha=\alpha(x/t)$ according to the Whitham modulation equations.
In this section we derive the solution of the mKdV equation given in  \eqref{eq:qrep} and \eqref{eq:bgwps} of Theorem~\ref{theorem1.3}. The leading order behaviour  of those expressions is  obtained by solving the  outer model RH problem \ref{rhp:W} (see below)
while  the error term is obtained by constructing the 
Airy parametrix at the points $\pm i\alpha$ and Bessel parametrix at the points $\pm i \eta_1$.

We start the section by constructing the  outer  model RH problem \ref{rhp:W}. 
Lemma~\ref{lem:lenses} shows that the jump matrix $\bm{J}_S(k)$ is exponentially near identity away from $[-i\alpha, i\alpha]$ as $t \to \infty$. Moreover, the convergence is uniform away from the four endpoints $k = \pm i \eta_1,\ \pm i \alpha$ where the lens contours return to imaginary axis. In order to construct a complete asymptotic description of the solution in the large time limit we introduce a set of model problems: an outer model to control the jumps which remain in the large-time limit, and  four local models to account for the locally non-uniform behavior at each of the endpoints. 

We are particularly interested in the outer model, as its solution will generate the leading order asymptotic behavior of the dynamics. The local models will produce sub-leading corrections to the dynamics which we are less interested in here, but could be computed fully in principle. 

\subsection{The outer model problem}\label{sec:outer.model}

Removing the jumps from the RH problem~\ref{rhp:S} which are near-identity as $t \to \infty$ results in the following model problem.  
\begin{RHP}\label{rhp:W}
Find a $2\times2$ matrix-valued function $\bm{W}(k; x,t)$ with the following properties 
\begin{enumerate}[label=\arabic*.]
	\item\label{rhpW1} $\bm{W}(k; x,t)$ is meromorphic for $k \in \C \setminus [- i\alpha, i\alpha]$.
	\item $\bm{W}(k; x,t) = \bm I + \bigo{k^{-1}}$ as $k \to \infty$. 
	\item For $k \in i(-\alpha, \alpha)$, the boundary values of $\bm{W}_\pm(k; x,t)$ satisfy the jump relation
	\begin{gather}\label{Wjumps}
		\bm{W}_+(k; x,t) = \bm{W}_-(k; x,t) \bm{J}_W(k; x,t), \\
		\bm{J}_W(k; x,t) = 
		\begin{dcases}
		   \begin{bsmallmatrix} 0 & i \\ i & 0 \end{bsmallmatrix}, & k \in \Sigma_{1,\alpha} \cup \Sigma_{2,\alpha}, \\
		   e^{i ({\Omega}+ \Delta) \sigma_3}, & k \in  (-i\eta_1, i \eta_1),
		\end{dcases}
	\end{gather}
	\item\label{rhpW4} For any endpoint $p \in \{ \pm i\eta_1, \pm i \alpha\}$, $\bm{W}(k; x,t) = \bigo{ (k-p)^{-1/4}}$, as $k \to p$. 
	\item $\bm{W}(k; x,t)$ has simple poles at $\pm i\kappa_0$ and no other poles. The poles satisfy the same residue conditions \eqref{Tres} as $\bm{T}(k;x,t)$. 
\end{enumerate}
\end{RHP}

The first four conditions above -- temporarily ignoring the poles at $\pm i \kappa_0$ -- define a well-known RH problem characterizing a finite-gap solution of mKdV, whose solution will be described below. If we let $\bm{W}^{(0)}$ denote the solution of the pole-free problem, then the solution of the RH problem~\ref{rhp:W} can by computed by introducing a Darboux transformation in the form
\begin{gather}\label{Darboux}
	\bm{W}(k; x,t) = \nonumber \\
	 \left( \bm I 
	   + \frac{i}{k-i\kappa_0} \begin{bmatrix} a(x,t) &  c(x,t) \\  b(x,t) & - d(x,t) \end{bmatrix} 
	   +\frac{i}{k+i\kappa_0} \begin{bmatrix}  d(x,t) &  b(x,t) \\ c(x,t) & -  a(x,t) \end{bmatrix} 
	   \right) \bm{W}^{(0)}(k;x,t) \ ,
\end{gather}
where the coefficients in the pre-factor can be computed directly in terms of entries of $\bm{W}^{(0)}(i\kappa_0;x,t)$ and from the residue conditions \eqref{Tres}. Before we describe this computation, let us first consider $\bm{W}^{(0)}$.

The explicit expression of the sectionally holomorphic function $\bm{W}^{(0)}$ which satisfies conditions \ref{rhpW1}--\ref{rhpW4} in the RH problem~\ref{rhp:W} is well known. 
On the genus-one Riemann surface $\mathfrak{X}$ defined by \eqref{surface} (see also \figurename~\ref{fig:homology}) let
\begin{equation}\label{eq:omegadiff}
	\omega = \left( 4 \int_0^{i\eta_1} \frac{\d k}{R(k)} \right)^{-1}   \frac{\d k}{R(k)}  =  \frac{\alpha}{4i K(m)} \frac{\d k}{R(k)}, 
\end{equation}
so that $\oint_{\mathcal A} \omega = 1$ and define the period
\begin{equation}\label{ratio}
	\tau := \oint_\mathcal{B} \omega =  \frac{i K(1-m)}{2K(m)}.
\end{equation}
Using $\omega$, define the integral
\begin{equation}\label{abel}
	A(k) = \int_{i\alpha}^k \omega, \quad  k \in \C \setminus [-i\alpha, i\alpha],
\end{equation}
where the path of integration is on any simple arc from $i\alpha$ to $k$ which does not intersect $[-i\alpha, i\alpha]$. We observe that
\begin{equation}\label{Abel.values}
	A(\infty) = - \frac{1}{4},  \quad 
	A_+(i\eta_1) = - \frac{\tau}{2}, \quad
	A_+(-i\eta_1) = - \frac{1}{2} - \frac{\tau}{2} , \quad
	A_+(-i\alpha) = - \frac{1}{2}  ,
\end{equation}
and
\begin{equation}
	\begin{aligned}
	&A_+(k) + A_-(k) = 0,\phantom{-} \quad k \in \Sigma_{1,\alpha},  
	&& &A_+(k) - A_-(k) = -\tau,  \quad k \in [-i\eta_1, i\eta_1], \\
	& A_+(k) + A_-(k) = -1, \quad k \in \Sigma_{2,\alpha}.
	\end{aligned}	
\end{equation}
Next, we introduce the Jacobi elliptic function
\begin{equation}\label{theta3}
	\theta_3(z;\tau) = \sum_{n \in \Z} e^{2\pi i n z + \pi n^2 i \tau},  \quad z \in \C,
\end{equation}
which is an even function of $z$, satisfies the periodicity relations
\begin{equation}
	\theta_3(z + h + k \tau; \tau ) = e^{-\pi i k^2 \tau - 2\pi i k z} \theta_3(z;\tau), \quad k,h \in \Z,
\end{equation}
and has a simple zero at the half period $\tfrac{1}{2} + \tfrac{\tau}{2}$.
Finally, define a function $\gamma$ analytic in $\C \setminus \{\Sigma_{1,\alpha} \cup \Sigma_{2,\alpha}\}$ by
\begin{equation} \label{gamma}
	\gamma(k) = \left( \frac{ k - i \alpha}{k- i \eta_1} \right)^{1/4} \left( \frac{k + i \eta_1}{ k + i \alpha} \right)^{1/4},
\end{equation}
and normalized such that $\gamma(k) \to 1$ as $k \to \infty$, so that 
\begin{align}
	\gamma_+(k) = i \gamma_-(k), \qquad k \in \Sigma_{1,\alpha} \cup \Sigma_{2,\alpha}.
\end{align}

Then the function $\bm{W}^{(0)}$ is given by the following formul\ae
\begin{equation}\label{W0.entries}
	\begin{aligned}
	&\bm{W}^{(0)}_{11}(k;x,t)  = \frac{1}{2} \left( \gamma(k) + \frac{1}{\gamma(k)} \right) 
	  \frac{ \theta_3( A(k) + \tfrac{1}{4} + \frac{{\Omega} + \Delta}{2\pi}; \tau)}{\theta_3(A(k) + \tfrac{1}{4}; \tau )} 
	  \frac{\theta_3(0; \tau )}{ \theta_3(\frac{{\Omega} + \Delta}{2\pi}; \tau)}\ , \\
	&\bm{W}^{(0)}_{12}(k;x,t)  = \frac{1}{2} \left( \gamma(k) - \frac{1}{\gamma(k)} \right)
	  \frac{ \theta_3( -A(k) + \tfrac{1}{4} + \frac{{\Omega} + \Delta}{2\pi}; \tau)}{\theta_3(-A(k) + \tfrac{1}{4}; \tau )} 
	  \frac{\theta_3(0; \tau )}{ \theta_3(\frac{{\Omega} + \Delta}{2\pi}; \tau)}\ , \\
	&\bm{W}^{(0)}_{21}(k;x,t)  = \frac{1}{2} \left( \gamma(k) - \frac{1}{\gamma(k)} \right) 
	  \frac{ \theta_3( A(k) - \tfrac{1}{4} + \frac{{\Omega} + \Delta}{2\pi}; \tau)}{\theta_3(A(k) - \tfrac{1}{4}; \tau )} 
	  \frac{\theta_3(0; \tau )}{ \theta_3(\frac{{\Omega} + \Delta}{2\pi}; \tau)}\ , \\
	&\bm{W}^{(0)}_{22}(k;x,t)  = \frac{1}{2} \left( \gamma(k) + \frac{1}{\gamma(k)} \right)
	  \frac{ \theta_3( -A(k) - \tfrac{1}{4} + \frac{{\Omega} + \Delta}{2\pi}; \tau)}{\theta_3(-A(k) - \tfrac{1}{4}; \tau )} 
	  \frac{\theta_3(0; \tau )}{ \theta_3(\frac{{\Omega} + \Delta}{2\pi}; \tau)}\ ,
	\end{aligned}
\end{equation}
with $\Omega$ and  $\Delta$ as in \eqref{omega.1} and \ref{Delta} respectively.

We're now ready to compute the coefficients in the Darboux transformation \eqref{Darboux}. Let
\begin{equation}
	\bm{W}^{(0)}(i \kappa_0;x,t) = 
	\begin{bmatrix} 
		w_{11}(x,t) & w_{12}(x,t) \\ w_{21}(x,t) & w_{22}(x,t)
	\end{bmatrix}
	=\begin{bmatrix} 
		w_{11} & w_{12} \\ w_{21} & w_{22}
	\end{bmatrix},
\end{equation}
where we have suppressed the $(x,t)$ dependence of the coefficients in the last equality for brevity in what follows. 

Suppose that $(x,t) \in \mathcal{S}_M^{^{(-)}} \cup \mathcal{S}_R^{^{(-)}}$.  
For $\bm{W}(k; x,t)$ given by \eqref{Darboux} to satisfy the residue conditions  \eqref{Tres-} requires that
\begin{equation}\label{Darboux.condition}
	\begin{bmatrix}  a &  c \\  b & - d \end{bmatrix} 
	\begin{bmatrix} w_{11} & w_{12} \\ w_{21} & w_{22} \end{bmatrix}
	\begin{bmatrix} 0 & 0 \\ 1 & 0 \end{bmatrix} = 
	\bm{0},
	\quad
	\text{implying}
	\quad
	c = - \frac{ w_{12} }{ w_{22}} a, \ \
	d = \frac{ w_{12} }{ w_{22}} b
\end{equation}
so that the limit on the right hand side exists. Defining
\begin{equation}\label{Q_minus_4.16}
	\mathcal{Q}^{^{(-)}}(x,t) := \frac{\bm{W}^{(0)}_{12}(i\kappa_0; x,t)}{\bm{W}^{(0)}_{22}(i\kappa_0; x,t)} = \frac{ w_{12}(x,t)}{w_{22}(x,t)},
	\qquad 
	\mathcal{Q}^{^{(-)}}_{\kappa_0}(x,t) := \od{}{\kappa_0}  \frac{\bm{W}^{(0)}_{12}(i\kappa_0 ; x,t)}{\bm{W}^{(0)}_{22}(i\kappa_0; x,t)} 
\end{equation}
and observing that $\det(\bm{W}^{(0)}) \equiv 1$, the remaining conditions are then equivalent to the system of equations
\begin{equation}\label{Darboux.system.-}
	\left\{
	\begin{gathered}
	  a X^{^{(-)}} + b Y^{^{(-)}} + \mathcal{Q}^{^{(-)}}(x,t) = 0, \\
	   -a Y^{^{(-)}} + b X^{^{(-)}} + 1 = 0,
	\end{gathered}
	\right.
\end{equation}
where
\begin{equation} \label{XY.sol.-}
\begin{aligned}
	&X^{^{(-)}}(x,t) = \frac{1}{ f(i\kappa_0)^2 w_{22}(x,t)^2 \chi e^{-2i\varphi(i\kappa_0; x,t) } }  +  \mathcal{Q}^{^{(-)}}_{\kappa_0}(x,t),
\\
	&Y^{^{(-)}}(x,t) = \frac{1}{2\kappa_0} \left[1 +\le( \mathcal{Q}^{^{(-)}}(x,t)\ri)^2\right].
	\end{aligned}
\end{equation}

For $(x,t) \in \mathcal{S}_M^{^{(+)}} \cup \mathcal{S}_R^{^{(+)}}$ the expansion \eqref{Darboux} must instead satisfy \eqref{Tres+}, and through a similar series of calculations one finds that $\alpha$ and $\beta$ satisfy a system of the same form as \eqref{Darboux.system.-} but with $\mathcal{Q}^{^{(-)}}, X^{^{(-)}}, Y^{^{(-)}}$ replaced by 
\begin{gather}
	\mathcal{Q}^{^{(+)}} (x,t):= \frac{w_{11}(x,t)}{w_{21}(x,t)} \ , \ \mathcal{Q}^{^{(+)}}_{\kappa_0}(x,t) := \od{}{\kappa_0}  \frac{\bm{W}^{(0)}_{11}(i\kappa_0 ; x,t)}{\bm{W}^{(0)}_{21}(i\kappa_0; x,t)} \ , 
	 \ c = - a \mathcal{Q}^{^{(+)}}, \ d = b \mathcal{Q}^{^{(+)}},   \\
	\label{XY.sol.+}
	X^{^{(+)}}(x,t) = \frac{ f'(i\kappa_0)^2  \chi e^{-2i\varphi(i\kappa_0; x,t) }}{w_{21}(x,t)^2} +  \mathcal{Q}^{^{(+)}}_{\kappa_0}(x,t),  
	\ 
	Y^{^{(+)}}(x,t) = \frac{1}{2\kappa_0} \left[ 1 + \le(\mathcal{Q}^{^{(+)}}(x,t)\ri)^2\right].
\end{gather}

Solving the system \eqref{Darboux.system.-} yields
\begin{equation}
	a(x,t) = \frac{ Y^{^{(\pm)}} - \mathcal{Q}^{^{(\pm)}} X^{^{(\pm)}}}{\le(X^{^{(\pm)}}\ri)^2+\le(Y^{^{(\pm)}}\ri)^2},
	\quad
	b(x,t) = -\frac{X^{^{(\pm)}} + \mathcal{Q}^{^{(\pm)}}Y^{^{(\pm)}}}{\le(X^{^{(\pm)}}\ri)^2+\le(Y^{^{(\pm)}}\ri)^2},
\end{equation}
for $(x,t) \in \mathcal{S}_M^{^{(\pm)}} \cup \mathcal{S}_R^{^{(\pm)}}$.

\subsection{The local models at the endpoints}\label{sect:local_problems}

The global model problem $\bm W$ is a good approximation of the original RH problem $\bm S$ everywhere in the complex plane, except at the endpoints $\pm i\eta_1, \pm \alpha$, where we will need to construct local parametrices. We will see that the local parametrix near $i\eta_1$ can be constructed in terms of the modified Bessel functions of index 0, and as for the point $i\alpha,$ we need to distinguish between the cases $\alpha < \eta_2$ and $\alpha = \eta_2$: in the case $\alpha = \eta_2,$ the parametrix is described in terms of the modified Bessel functions of index 0, and in the case $\alpha < \eta_2$ the parametrix is described in terms of Airy functions. The construction of the parametrices at $-i\eta_1$ and $-i\alpha$ will follow from the symmetric properties of the RH problem.

\subsubsection{Local parametrix at $k= i\alpha,$ case $\alpha < \eta_2$}\label{sect_loc_par_alpha}
\noindent \\
Inspection of the local behavior of the function $\varphi(k;x,t)$ 
\eqref{varphi.3b}, prompts to introduce a local variable $\lambda = \lambda(k; x/t)$ in a disk $U_\delta(i\alpha)$ (centered at $i\alpha$, of a sufficiently small radius $\delta>0$) as follows:
\[ 2 i \varphi(k; x,t) =: \frac43 t \lambda^{3/2}, \] 
so that 
\[ \lambda = \frac{k - i\alpha}{i}\cdot\(\frac{12(\alpha^2 - \mu^2)\sqrt{2\alpha} }{\sqrt{\alpha^2 - \eta_1^2}}\)^{\frac23}\(1 + \mathcal{O}(k - i\alpha)\), \qquad \text{as } k\to i\alpha,\]  
and the branch cut for $\lambda^{3/2},$ i.e., the half-line $\lambda<0,$ corresponds to $k\in(i\eta_1, i\alpha).$

\medskip
\noindent
Similarly as in, for example, \cite{D99, DKMcLVZ99}, we construct a function that solves exactly the same jump as $\bm{T}$ does in a small neighborhood of the point $i\alpha.$ 
Define
\[
{\bm \Psi}_{\rm Ai}(\lambda) = 
\begin{cases}
\begin{pmatrix}
v_1(\lambda) & i v_{-1}(\lambda)
\\
v'_1(\lambda) & i v'_{-1}(\lambda)
\end{pmatrix}  &\arg\lambda \in (\frac{2\pi}{3}, \pi),\\
\begin{pmatrix}
v_1(\lambda) & -v_0 (\lambda)
\\
v'_1(\lambda) & -v_0'(\lambda)
\end{pmatrix} & \arg\lambda \in (0, \frac{2\pi}{3}),
\\
\begin{pmatrix}
v_{-1}(\lambda) & -i v_{1}(\lambda)
\\
v'_{-1}(\lambda) & -i v'_{1}(\lambda)
\end{pmatrix} & \arg\lambda \in (-\pi, \frac{2\pi}{3}), \\
\begin{pmatrix}
v_{-1}(\lambda) & -v_0 (\lambda)
\\
v'_{-1}(\lambda) & -v_0'(\lambda)
\end{pmatrix} &\arg\lambda \in (\frac{-2\pi}{3}, 0),
\end{cases}
\]
where we denoted $v_0(\lambda) = \sqrt{2\pi}\, \mathrm{Ai}(\lambda),$ $v_{j} = \sqrt{2\pi}e^{\frac{-\pi j i}{6}}\, \mathrm{Ai}(\lambda e^{-\frac{2\pi j i}{3}})$, $j=\pm1,$ and $'$ means derivative with respect to $\lambda.$ The function ${\bm \Psi}_{\rm Ai}$ satisfies the jump conditions 
\[
{\bm \Psi}_{{\rm Ai}, +}(\lambda) = {\bm \Psi}_{{\rm Ai}, -}(\lambda)
\begin{dcases} 
\begin{bmatrix}0 & i \\ i & 0\end{bmatrix} & \lambda \in(-\infty, 0), \\
\begin{bmatrix}1 & -i \\ 0 & 1\end{bmatrix} & \lambda \in(\infty e^{\pm\frac{2\pi i}{3}}, 0), \\
\begin{bmatrix}1 & 0 \\ i & 1\end{bmatrix}, & \lambda \in(0, +\infty),
\end{dcases}
\]
 where the orientation of the segments is from the first mentioned point to the second one, and where $(\infty e^{i\beta}, 0)$ denotes a ray coming from infinity to $0$ at an angle $\beta\in\mathbb{R}$ (see \figurename~\ref{Fig_loc_par}, left). Besides, the function ${\bm \Psi}_{\rm Ai} $ satisfies the asymptotics
\[
{\bm \Psi}_{\rm Ai}(\lambda) = \lambda^{-\sigma_3/4}\frac{1}{\sqrt{2}}\begin{bmatrix}1 & -1 \\ 1 & 1\end{bmatrix} {\bm{\mathcal{E}}}_{\rm Ai}(\lambda)e^{\frac23\lambda^{3/2} \sigma_3},
\qquad  {\bm{\mathcal{E}}}_{\rm Ai}(\lambda) = \bm{ I} + \mathcal{O}\(\lambda^{-\frac32}\),\quad \lambda\to\infty.
\]
uniformly in $\arg\lambda \in[-\pi, \pi]$.
\medskip

Finally, define 
\[
\bm{P}_{\rm Ai}(k; x,t) = \bm{B}_{\rm Ai}(k; x,t) {\bm \Psi}_{\rm Ai}\(t^{\frac23} \lambda(k; x / t )\) e^{-i\varphi(k;x,t)\sigma_3}\(f(k)\sqrt{\widehat r(k)}\)^{\sigma_3}\ ,
\]
for $ |k-i\alpha| < \delta$, where
\[
\bm{B}_{\rm Ai}(k;x,t) = \bm{W}(k; x,t)
\(f(k)\sqrt{\widehat r(k)}\)^{-\sigma_3}\frac{1}{\sqrt{2}}\begin{bmatrix}1 & 1 \\ -1 & 1\end{bmatrix}\(t^{\frac23}\lambda\)^{\frac{\sigma_3}{4}}
\]
and $\bm{B}_{\rm Ai}$ is analytic in $U_{\delta}(i\alpha)$ (i.e., it does not have jumps across $(i\alpha-i\delta, i\alpha + i\delta)$).
The function $\bm{P}_{\rm Ai}$ satisfies exactly the same jumps inside $U_{\delta}(i\alpha)$ as $\bm{T}$ does, and on the boundary $\partial U_{\delta}(i\alpha)$ we have the following matching condition:
\begin{align*}
\bm{T}(k;x, t) \bm{P}_{\rm Ai}^{-1}(k;x,t)
 &= 
\bm{W}(k; x,t)
\(
f(k)\sqrt{\widehat r(k)}\)^{-\sigma_3}
{\bm{\mathcal{E}}}_{\rm Ai}(t^{\frac23}\lambda(k;x/t)) 
\(f(k)\sqrt{\widehat r(k)}\)^{\sigma_3}
\bm{W}(k; x,t)^{-1}
\\
&=
\bm{I} + \mathcal{O} \(t^{-1}\)
\end{align*}
as $t\to\infty.$ 
Here we used the fact that  $\alpha$ is not close to $\eta_2$ and $\eta_1,$ and hence $\bm{W}$ is bounded on $\partial U_{\delta}(i\alpha).$

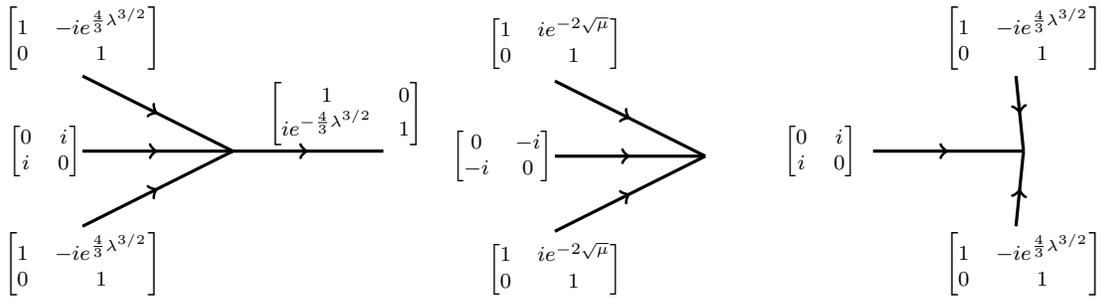
\begin{figure}[ht!]
\begin{tikzpicture}
\draw[very thick, decoration = {markings, mark = at position 0.25 with {\arrow{>}}},
decoration = {markings, mark = at position 0.75 with {\arrow{>}}}, 
postaction = decorate](-2,0) to (2,0);
\draw
[very thick, decoration = {markings, mark = at position 0.5 with {\arrow{>}}},
postaction = decorate]
(-2,1) to (0,0);
\draw[very thick, decoration = {markings, mark = at position 0.5 with {\arrow{>}}},
postaction = decorate](-2,-1) to (0,0);
\node at (-2, 1.5){\scriptsize $\begin{bmatrix}1 & -i e^{\frac43\lambda^{3/2}} \\ 0 & 1\end{bmatrix}$};
\node at (-2,-1.5){\scriptsize $\begin{bmatrix}1 & -i e^{\frac43\lambda^{3/2}} \\ 0 & 1\end{bmatrix}$};
\node at (1.5,0.5){\scriptsize $\begin{bmatrix}1 & 0 \\ i e^{-\frac43\lambda^{3/2}} & 1\end{bmatrix}$};
\node at (-2.5,0){\scriptsize $\begin{bmatrix}0 & i \\ i & 0 \end{bmatrix}$};
\end{tikzpicture}
\begin{tikzpicture}
\draw[very thick, decoration = {markings, mark = at position 0.5 with {\arrow{>}}},
postaction = decorate](-2,0) to (0,0);
\draw
[very thick, decoration = {markings, mark = at position 0.5 with {\arrow{>}}},
postaction = decorate]
(-2,1) to (0,0);
\draw[very thick, decoration = {markings, mark = at position 0.5 with {\arrow{>}}},
postaction = decorate](-2,-1) to (0,0);
\node at (-2, 1.5){\scriptsize $\begin{bmatrix}1 & i e^{-2\sqrt{\mu}} \\ 0 & 1\end{bmatrix}$};
\node at (-2,-1.5){\scriptsize $\begin{bmatrix}1 & i e^{-2\sqrt{\mu}} \\ 0 & 1\end{bmatrix}$};
\node at (-2.7,0){\scriptsize $\begin{bmatrix}0 & -i \\ -i & 0 \end{bmatrix}$};
\end{tikzpicture}
\qquad
\begin{tikzpicture}
\draw[very thick, decoration = {markings, mark = at position 0.5 with {\arrow{>}}},
postaction = decorate](-2,0) to (0,0);
\draw
[very thick, decoration = {markings, mark = at position 0.5 with {\arrow{>}}},
postaction = decorate]
(-0.1,1) to (0,0);
\draw[very thick, decoration = {markings, mark = at position 0.5 with {\arrow{>}}},
postaction = decorate](-0.1,-1) to (0,0);
\node at (0, 1.5){\scriptsize $\begin{bmatrix}1 & -i e^{\frac43\lambda^{3/2}} \\ 0 & 1\end{bmatrix}$};
\node at (0,-1.5){\scriptsize $\begin{bmatrix}1 & -i e^{\frac43\lambda^{3/2}} \\ 0 & 1\end{bmatrix}$};
\node at (-2.7,0){\scriptsize $\begin{bmatrix}0 & i \\ i & 0 \end{bmatrix}$};
\end{tikzpicture}
\caption{
Left: Jumps for $\bm{\Psi}_{\rm Ai}(\lambda)e^{-\frac23\lambda^{3/2}\sigma_3}.$
Middle: Jumps for $\bm{\Psi}_{\rm Bes}(\mu) e^{\sqrt{\mu} \sigma_3}.$
Right: Jumps for $(-i)^{\sigma_3}\bm{\Psi}_{\rm Bes}(\frac49\lambda^3) e^{\frac23\lambda^{3/2} \sigma_3} i^{\sigma_3}.$
}
\label{Fig_loc_par}
\end{figure}

\subsubsection{Local parametrix at $k = i\eta_1$}\label{sect_loc_par_eta1}
\noindent \\
Introduce a local variable $\mu = \mu(k; x/t)$ in a disk $U_{\delta}(i\eta_1)$ by formula 
\[ \varphi_{\mp}(k;x,t) =: \pm\frac12 {\Omega} + i t \sqrt{\mu},\] 
so that
\[ \mu(k; x/t) = \frac{576(\eta_1^2-\mu^2)^2(\alpha^2 - \eta_1^2)}{2\eta_1}\cdot\frac{k - i\eta_1}{-i}\(1 + \mathcal{O}(k - i\eta_1)\), \qquad k \to i\eta_1,\]
and the branch cut of $\sqrt{\mu}$ corresponds to $k\in(i\eta_1, i\alpha).$

\medskip
\noindent 
Define
\[
\bm{\Psi}_{\rm Bes}(\mu) = \begin{cases}
\begin{pmatrix}
\frac{1}{\sqrt{\pi}} K_0(\sqrt{\mu}) & \frac{-i}{\sqrt{\pi}} K_0(e^{-\pi i}\sqrt{\mu})
\\
\frac{1}{\sqrt{\pi}} \frac{\d}{\d\mu} K_0(\sqrt{\mu}) & \frac{-i}{\sqrt{\pi}} \frac{\d}{\d\mu} K_0(e^{-\pi i}\sqrt{\mu})
\end{pmatrix} & \arg\mu\in(\theta, \pi),
\\\\
\begin{pmatrix}
\frac{1}{\sqrt{\pi}} K_0(\sqrt{\mu}) & \frac{i}{\sqrt{\pi}} K_0(e^{\pi i}\sqrt{\mu})
\\
\frac{1}{\sqrt{\pi}} \frac{\d}{\d\mu} K_0(\sqrt{\mu}) & \frac{i}{\sqrt{\pi}} \frac{\d}{\d\mu} K_0(e^{\pi i}\sqrt{\mu})
\end{pmatrix} & \arg\mu\in(-\pi, -\theta),
\\\\
\begin{pmatrix}
\frac{1}{\sqrt{\pi}} K_0(\sqrt{\mu}) & \sqrt{\pi} I_0(\sqrt{\mu})
\\
\frac{1}{\sqrt{\pi}} \frac{\d}{\d\mu} K_0(\sqrt{\mu}) & 
\sqrt{\pi} \frac{\d}{\d\mu} I_0(\sqrt{\mu})
\end{pmatrix} & \arg\mu\in(-\theta, \theta),
\end{cases}
\]
where $I_0, K_0$ are the modified Bessel functions of index $0,$ and $\theta\in(0, \pi).$ The function $\bm{\Psi}_{\rm Bes}$ has jumps 
\[
\bm{\Psi}_{{\rm Bes}, +}(\mu) = \bm{ \Psi}_{{\rm Bes},-}(\mu) \begin{dcases}
\begin{pmatrix}
1 & i \\ 0 & 1\end{pmatrix} & \mu \in(\infty e^{\pm i\theta}, 0), \\
\begin{pmatrix}
0 & i \\ i & 0\end{pmatrix} & \mu \in(-\infty, 0). 
\end{dcases}
\]
 (see \figurename~\ref{Fig_loc_par}, middle). As $\mu\to\infty,$ the function $\bm{\Psi}_{\rm Bes}$ satisfies the asymptotics
\[
\bm{\Psi}_{\rm Bes} (\mu) = \mu^{-\sigma_3 / 4}\frac{1}{\sqrt{2}}\begin{bmatrix} 1 & 1 \\ -1 & 1 \end{bmatrix} {\bm{\mathcal{E}}}_{\rm Bes}(\mu) e^{-\sqrt{\mu}\sigma_3},
\qquad 
{\bm{\mathcal{E}}}_{\rm Bes}(\mu) = \bm{I} + \mathcal{O}(\mu^{-1/2}),\quad \mu\to\infty.
\]
 uniformly in $\arg\mu\in[-\pi, \pi]$.
\medskip
Finally, define
\[
\bm{P}_{{\rm Bes}, \eta_1}(k; x,t) = \bm{B}_{{\rm Bes}, \eta_1}(k; x,t) \bm{\Psi}_{\rm Bes}(t^2\mu(k; x/t)) e^{-i\varphi(k; x,t)\sigma_3} 
\(\sqrt{\widehat r(k)} f(k) \)^{\sigma_3}
,\quad |k - i\eta_1| < \delta, \]
where
\[
\bm{B}_{{\rm Bes}, \eta_1}(k; x,t) = \bm{W}(k; x,t) 
\(\sqrt{\widehat r(k)} f(k)e^{-i/2\,{\Omega} \Re k}\)^{-\sigma_3}
\frac{1}{\sqrt{2}}\begin{bmatrix}1 & -1 \\ 1 & 1 \end{bmatrix}(t^2\mu(k; x/t))^{\sigma_3/4}
\]
is analytic in $U_{\delta}(i\eta_1).$ The function $\bm{P}_{{\rm Bes}, \eta_1}$ has exactly the same jumps as $\bm{T}$ in $U_{\delta}(i\eta_1),$ and on the boundary $\partial U_{\delta}(i\eta_1)$ it matches with $\bm{T}$ as follows:
\begin{multline*}
\bm{T}(k; x, t) \bm{P}^{-1}_{{\rm Bes},\eta_1}(k; x,t) = \bm{W}(k; x, t) 
\(\sqrt{\widehat r(k)} f(k)e^{-i/2\,{\Omega} \Re k}\)^{-\sigma_3}
{\bm{\mathcal{E}}}_{\rm Bes}(t^2\mu(k; x/t))
\times
\\
\times
\(\sqrt{\widehat r(k)} f(k)e^{-i/2\,{\Omega} \Re k}\)^{\sigma_3}
\bm{W}^{-1}(k; x, t)
=
I + \mathcal{O}(t^{-1}),
\end{multline*}
where we again used boundedness of $\bm{W}$ on $\partial U_{\delta}(i\eta_1)$, and the fact that $i\eta_1$ is not close to $i\alpha.$

\subsubsection{Local parametrix at $k=i\alpha,$ the case $\alpha = \eta_2$}
\noindent \\
Here we introduce the local variable $\lambda = \lambda(k; x/t)$ as in Section \ref{sect_loc_par_alpha}, and use the function $\bm{\Psi}_{\rm Bes}$ from the Section \ref{sect_loc_par_eta1}, but with $\theta\in(\pi/2, 2\pi/3)$ (see also \figurename~\ref{Fig_loc_par}, right), and define a function that satisfies the same jumps as $\bm{T}$ in a disk $U_{\delta}(i\eta_2)$ as follows:
\begin{multline*}
\bm{P}_{{\rm Bes}, \eta_2}(k; x,t) = 
\bm{B}_{{\rm Bes}, \eta_2}(k; x,t) e^{-\pi i/2\, \sigma_3}
\bm{\Psi}_{\rm Bes}\(\frac49 t^2 \lambda^3(k; x/t)\) e^{\pi i/2\,\sigma_3} e^{i\varphi(k; x, t)\sigma_3}
\(\sqrt{\widehat r(k)} f(k)\)^{\sigma_3},
\\ |k - i\eta_2| < \sigma_3,
\end{multline*}
where
\[
\bm{B}_{{\rm Bes}, \eta_2}(k; x,t) = \bm{W}(k; x,t)
\(\sqrt{\widehat r(k)} f(k)\)^{-\sigma_3}
\frac{1}{\sqrt{2}}
\begin{bmatrix}
1 & -1 
\\ 1 & 1
\end{bmatrix}
\(\frac49\lambda(k; x/t)^3\)^{\sigma_3 / 4}.
\]
Then $\bm{P}_{{\rm Bes}, \eta_2}$ satisfies the same jumps as $\bm{T}$ inside the disk $U_{\delta}(i\eta_2),$ and the matching between them on $\partial U_{\delta}(i\eta_2)$ is as follows:
for $k\in\partial U_{\delta}(i\eta_2)$ as $t\to\infty,$
\begin{multline*}
\bm{P}_{{\rm Bes}, \eta_2}(k;x,t)
\bm{T}^{-1}(k;x,t)
=
\bm{W}(k; x,t)
\(\sqrt{\widehat r(k)} f(k)\)^{-\sigma_3}
(-i)^{\sigma_3}
{\bm{\mathcal{E}}}_{\rm Bes}\(\frac49 t^2 \lambda^3\)
i^{\sigma_3}
\(\sqrt{\widehat r(k)} f(k)\)^{\sigma_3}
\times
\\
\times
\bm{W}^{-1}(k; x,t)
=
\bm{I} + \mathcal{O}(t^{-1}).
\end{multline*}

\subsubsection{Local models when $\alpha= i \eta_2$ and $r(k) | k - i\eta_k|^{\pm1/2} = \bigo{1}$} 
When Assumption~\ref{Assumption:r2} is satisfied and $\alpha = \eta_2$, then the jumps of $\bm{T}(k)$ are as given in Figure~\ref{fig:lenses.hard.edges}. Because in this setting we can deform the lens counts away from all the endpoints the outer model $\bm{W}(k;x,t)$	is an exponentially accurate model
uniformly in $\C$. So local parametrices are not needed in this case. 

In the modulated elliptic region, that is, when $\alpha \in (\eta_1, \eta_2)$, the jumps of $\bm{T}$ are as given in Figure~\ref{fig:lenses.mixed.edges}. Here, a Bessel parametrix is not needed at $i\eta_1$ as the lens $\mathcal{C}_1$ remains bounded away from $i\eta_1$, but the lens must return to $
\Sigma_1$ at $i\alpha$. As such, a local Airy parametrix is required in $U_\delta(i\alpha)$. Consequently, there is no improvement to error estimates in the modulated region.

\subsection{Error analysis and conclusion of the proof of Theorem~\ref{theorem1.3}}\label{sect:error_analysis}

Define $\bm{P}(k; x,t)$ to be equal $\bm{P}_{{\rm Bes},\eta_1}(k; x, t)$ inside the disk $U_{\delta}(i\eta_1),$ to be equal $\bm{P}_{\rm Ai}(k; x,t)$ inside the disk $U_{\delta}(i\alpha)$ in the case $\alpha < \eta_2$ and $\bm{P}_{{\rm Bes},\eta_2}(k; x,t)$ inside the disk $U_{\delta}(i\alpha)$ in the case $\alpha = \eta_2.$ Furthermore, in the corresponding disks $|k+i\eta_1| < \delta,$ $|k + i\alpha| < \delta$ in the lower half plane, we define 
\[ \bm{P}(k) := \begin{bmatrix}0 & 1 \\ -1 & 0 \end{bmatrix} \bm{P}\ol{(\ol k; x,t)} \begin{bmatrix}0 & -1 \\ 1 & 0 \end{bmatrix},\]
 and define $\bm{P}(k;x,t)$ to be equal to $\bm{W}(k;x,t)$ elsewhere. Next, define the error function
\[\bm{R}(k; x,t) = \bm{S}(k; x,t) \bm{P}(k;x, t)^{-1}.\]
Lemma \ref{lem:lenses}, Proposition \ref{prop:left.estimate} and the matching properties of $\bm{P}_{\rm Ai},$ $\bm{P}_{{\rm Bes}, \eta_1}$, $\bm{P}_{{\rm Bes}, \eta_2}$ allow to conclude that 
\[ \bm{R}(k;x,t) = \bm{I} + \mathcal{O}(t^{-1})\]
 uniformly in $k$ as $t\to\infty.$  Tracing back the chain of transformations from the remainder problem $\bm{R}$ to the original RH problem~\ref{rhp:X} for $\bm{X},$ we find that 
\[
q(x,t) = \lim\limits_{k\to\infty} 2ik \bm{X}_{12}(k;x,t)
=
\lim\limits_{k\to\infty} 2ik \bm{T}_{12}(k;x,t)
=
\lim\limits_{k\to\infty} 2ik \bm{W}_{12}(k;x,t)
+
\lim\limits_{k\to\infty} 2ik \bm{R}_{12}(k;x,t)\ .
\]
The second term here is of the order $\mathcal{O}(t^{-1})$ in view of smallness of $\bm{R},$ and the first term decomposes further into a background part and a solitonic part,
\begin{equation}\label{q.leading}
	q(x,t) = \lim_{k \to \infty} 2ik \bm{W}_{12}(k; x,t) + \mathcal{O}(t^{-1}) = q_{\mathrm{bg}}(x,t) + q_{\mathrm{sol}}(x,t) + \mathcal{O}(t^{-1}),
\end{equation}
where
\begin{align}
	\label{q.bg.recovery}
	q_{\mathrm{bg}}(x,t) &=  \lim_{k \to \infty} 2ik \bm{W}^{(0)}_{12}(k;x,t)\ , 
	\\
	\label{q.sol.1}
	q_{\mathrm{sol}}(x,t)  &= -2(c + b ) = 2( a \mathcal{Q}^{^{(\pm)}} - b) 
	   = \frac{2 \le[ 1-\le(\mathcal{Q}^{^{(\pm)}}\ri)^2 \ri] X^{^{(\pm)}}+4\mathcal{Q}^{^{(\pm)}} Y^{^{(\pm)}}}{\le(X^{^{(\pm)}}\ri)^2+\le(Y^{^{(\pm)}}\ri)^2}\ .
\end{align}

Observe that $A(i\kappa_0) \in \R$, implying that $\mathcal{Q}^{^{(\pm)}} \in \R$ and thus $Y^{^{(\pm)}} > 0$. Similarly, $X^{^{(\pm)}} \in \R$, so that $a(x,t)$ and $b(x,t)$ are well-defined for any value of $(x,t) \in \mathcal{S}_L \cup \mathcal{S}_R$. Note also that the sets $\mathcal{S}^{^{(+)}}$ and $\mathcal{S}^{^{(-)}}$ in \eqref{before.after} are defined such that $X^{^{(\pm)}} \to \infty$ as $(x,t) \to \infty$ in $\mathcal{S}^{^{(\pm)}}$ -- non-tangentially to the boundary -- due to the exponential growth/decay of the term $\chi e^{2i\varphi(i\kappa_0; x,t)}$ while $\mathcal{Q}^{^{(\pm)}}$ and $Y^{^{(\pm)}}$ remain bounded as they have no dependence on the residue coefficients, therefore 
\begin{equation}
	q_{\mathrm{sol}}(x,t) \to 0 \quad \text{ as } (x,t)  \to \infty  \text{ in any relatively compact subset of } \mathcal{S}^{^{(\pm)}}.
\end{equation}
This justifies the interpretation of the term $q_\mathrm{bg}(x,t)$ as the background solution through which the localized soliton solution $q_{\mathrm{sol}}(x,t)$ is passing. 

In order to complete the proof of Theorem~\ref{theorem1.3}, we simplify the expression in \eqref{q.bg.recovery} of the background term $q_{\mathrm{bg}}$.

\paragraph{\bf Simplifying the formul\ae \ for $q_{\mathrm{bg}}$.} We will  resort to the following identities: 
\begin{equation}
   \begin{aligned}
     &\theta_3(0;\tau)  \theta_4(z;\tau) =  \theta_4(z; 2\tau)^2 + \theta_1(z;2\tau)^2, \\
     &\theta_4(0;\tau)  \theta_3(z;\tau) =  \theta_4(z; 2\tau)^2 - \theta_1(z;2\tau)^2. 
   \end{aligned}
\end{equation}

Using the large $k$ expansion of \eqref{gamma}, the first equality in \eqref{Abel.values}, and \eqref{W0.entries}, 
we have   that \eqref{q.bg.recovery} takes the form
\begin{gather}
	\begin{aligned} \label{q.bg1}
	q_{\mathrm{bg}}(x,t)
	&= (\alpha - \eta_1) 
	\frac{ \theta_3(0; \tau )}{\theta_3(\tfrac{1}{2}; \tau )} 
	  \frac{\theta_3( \tfrac{1}{2} + \frac{{\Omega} + \Delta}{2\pi}; \tau)}{ \theta_3(\frac{{\Omega} + \Delta}{2\pi}; \tau)} 
	= (\alpha - \eta_1) 
	\frac{ \theta_3(0; \tau )}{\theta_4(0; \tau )} 
	  \frac{\theta_4( \frac{{\Omega} + \Delta}{2\pi}; \tau)}{ \theta_3(\frac{{\Omega}+ \Delta}{2\pi}; \tau)} \\
	&= (\alpha - \eta_1) 
	  \frac{\theta_4( \frac{{\Omega}+ \Delta}{2\pi}; 2\tau)^2 + \theta_1( \frac{{\Omega}+ \Delta}{2\pi}; 2\tau)^2}
	     {\theta_4( \frac{{\Omega}+ \Delta}{2\pi}; 2\tau)^2 - \theta_1( \frac{{\Omega}+ \Delta}{2\pi}; 2\tau)^2}  \\
	&= (\alpha - \eta_1) 
	   \frac{1 + \sqrt{m} \sn^2 \left( \frac{K(m)}{\pi}({\Omega}+ \Delta) , m \right)}{1 - \sqrt{m} \sn^2 \left( \frac{K(m)}{\pi}({\Omega}+ \Delta), m \right) }  \\	
	&= (\alpha - \eta_1) 
		\nd \left( \frac{K(m_1)}{\pi} ({\Omega}+ \Delta),  m_1 \right) \\
	&= (\alpha - \eta_1) \frac{ 1 + \sqrt{m}}{1-\sqrt{m}} 
		\dn \left( \frac{K(m_1)}{\pi} ({\Omega}+ \Delta + \pi),  m_1 \right)	 \\
	&= (\alpha + \eta_1) 
		\dn \left( \frac{K(m_1)}{\pi} ({\Omega}+ \Delta + \pi),  m_1 \right),	 
	\end{aligned}
\shortintertext{so that}
	\label{q.bg}
	q_{\mathrm{bg}}(x,t) = (\alpha + \eta_1) 
		\dn \left( (\alpha + \eta_1) ( x - 2(\eta_1^2+\alpha^2) t - x^{^{(\pm)}}), m_1 \right),	
	\
	(x,t) \in \mathcal{S}_M^{^{(\pm)}} \cup \mathcal{S}_R^{^{(\pm)}},	
\end{gather}
where we've used the Landen transformation \cite{Akheizer} to simplify the above expressions: 
\begin{equation}
\label{Landen}
	m_1 = \frac{4\sqrt{m}}{(1+\sqrt{m})^2} = \frac{4 \alpha \eta_1}{(\alpha+ \eta_1)^2}, \qquad (1+\sqrt{m}) K(m) = K(m_1)
\end{equation}
and
\begin{equation}\label{x.shifts.background}
	x^{^{(\pm)}} = \frac{K(m)(\Delta^{^{(\pm)}} - \pi)}{\alpha \pi} = \frac{ K(m_1) \left( \displaystyle \frac{\Delta^{^{(\pm)}}}{\pi} - 1 \right)}{\alpha + \eta_1}.
\end{equation}

It's clear from the above formulae that the effect of the soliton passing through the background elliptic wave solution is to induce a phase change in the background solution which we can compute in terms of simple objects on the Riemann surface $\mathfrak{X}$ \eqref{surface}  which defines $q_{\rm bg}$. 

\begin{prop}\label{prop51}
	Fix $x/t > 4 \eta_1^2$, then the phase shift in the soliton gas induced by interaction with the trial soliton travelling with initial velocity $4\kappa_0^2$ is given by 
	\[
		x^{^{(-)}} - x^{^{(+)}} = -\frac{2K(m)}{\alpha} \le( 1+  4A(i\kappa_0) \ri) \in \left(-\frac{2 K(m_1)}{\alpha+\eta_1}, \ 0 \right),  
	\]
	where $A(k)$ is the Abel map on the Riemann surface given by \eqref{abel}. 
\end{prop}

\begin{proof}
	By fixing $x/t > 4\eta_1^2$, we consider a line in space-time where the background solution is an elliptic wave with fixed parameter $\alpha = \alpha(x/t)$.
	Starting from \eqref{x.shifts.background} we have
	\[
		x^{^{(-)}} - x^{^{(+)}} = \frac{(\Delta^{^{(-)}} - \Delta^{^{(+)}})K(m)}{\alpha \pi} = \frac{K(m)}{\alpha \pi} h(i\kappa_0),
	\]
	where, using \eqref{Delta}, the function $h(k)$ is given by
	\begin{equation}\label{h-int-form}	
		h(k) = -4i \int_{\Sigma_{1,\alpha} \cup \Sigma_{2,\alpha}} \ln \left( \frac{ k -  s}{k+s} \right)  \omega(s),
	\end{equation} 
	where $\omega(k)$ is the normalized holomorphic differential on the genus one Riemann surface $\mathfrak{X}$ introduced in \eqref{eq:omegadiff}.
	By differentiating in $k$, and evaluating the resulting integral using residues, one finds that
\[	
h'(k) = -8\pi \od{\omega}{k} = -8\pi A'(k) \ , \text{ i.e. } h(k) = h(\infty) + 8\pi \left(A(\infty) - A(k) \right)\ .
\]
	The result is completed by observing that $h(\infty) = 0$ from \eqref{h-int-form}  and that $A(\infty) = -\tfrac{1}{4}$ from \eqref{Abel.values}. 
\end{proof}
We conclude the proof of Theorem~\ref{theorem1.3} by calculating the behaviour of the soliton $q_\mathrm{sol}$ for $\kappa_0\gg\eta_2$.

\paragraph{\bf Simplifying $q_\mathrm{sol}$ in the limit as $\kappa_0 \to \infty$. }\label{sec:kappa.to.infty}

The expression for $q_\mathrm{sol}$ can be greatly simplified if we assume that $\kappa_0 \gg \eta_2$.
In order to get something nontrivial when $\kappa_0 \to \infty$ we use \eqref{x0} to write the norming constant in the form $\chi = 2\sgn(\chi) \kappa_0 e^{-2\kappa_0 x_0}$. Then using \eqref{phi.def} we have 
\begin{equation}
	\chi e^{-2i\varphi(i\kappa_0; x,t)} = 2\sgn(\chi) \kappa_0 e^{ 2\kappa_0( x- x_0 - 4 \kappa_0^2 t)} \left[1 - \frac{2 g_{-1}}{\kappa_0} + \bigo{\kappa_0^{-2}} \right],
\end{equation}
where 
\begin{equation}
	\begin{aligned}
	g_{-1} 
	  &= \frac{1}{2\pi i} \int\limits_{\Sigma_{1,\alpha } \cup \Sigma_{2, \alpha}} \frac{s^2(2xs + 8 t s^3)}{R_+(s)} \, \d s
	     + \frac{ {\Omega}}{2\pi i } \int_{-i\eta_1}^{i \eta_1} \frac{s^2 }{R(s)} \, \d s \\
	  & = \frac{\alpha^2 + \eta_1^2}{2} \left( x + 2(\alpha^2+\eta_1^2)^2 t \right) + \frac{(\alpha^2-\eta_1^2)^2}{2}t
	     + \frac{ {\Omega} \alpha}{\pi  } ( E(m) - K(m)) \\
	&= -\frac{\alpha^2 - \eta_1^2}{2}  \left( x - 3(3\alpha^2+ \eta_1^2) t \right) \ , 
	\end{aligned}
\end{equation}
where in the last equality we used \eqref{omega.1} and \eqref{whitham.00}.

From the expansions, 
\begin{equation}
\begin{aligned}
	&A(i \kappa_0) = -\frac{1}{4} + \frac{1}{4 K(m)} \frac{\alpha}{\kappa_0} + \bigo{\kappa_0^{-3}} \ ,\\
	&\gamma^2(i\kappa_0) = 1 - \frac{\alpha - \eta_1}{\kappa_0} + \bigo{\kappa_0^{-2}} \ ,
\end{aligned}
\end{equation}
it follows that  $\mathcal{Q}^{^{(-)}}(x,t)$ and  $Y^{^{(-)}}(x,t)$  in  \eqref{Q_minus_4.16}    and \eqref{XY.sol.-}  simplify to 
\begin{equation}
\begin{aligned}\label{Q.large.kappa}
	\mathcal{Q}^{^{(-)}}(x,t)  &= - \frac{\alpha - \eta_1}{2\kappa_0}  \frac{ \theta_3(0; \tau )}{\theta_3(\tfrac{1}{2}; \tau )} 
	  \frac{\theta_3( \tfrac{1}{2} + \frac{{\Omega} + \Delta}{2\pi}; \tau)}{ \theta_3(\frac{{\Omega} + \Delta}{2\pi}; \tau)}  + \bigo{\kappa_0^{-2}}\\
	  &= -\frac{1}{2\kappa_0} q_{\mathrm{bg}}(x,t) + \bigo{\kappa_0^{-2}} \ , \\
	 Y^{^{(-)}}(x,t) &= \frac{1}{2\kappa_0}\left[ 1 +  \frac{ q_\mathrm{bg}(x,t)^2}{4\kappa_0^2} + \bigo{\kappa_0^{-3}} \right]\ . \\
\end{aligned}
\end{equation}	
Similarly, 
\begin{equation}
\begin{aligned}\label{Qk.large.kappa}
	&\mathcal{Q}^{^{(-)}}_{\kappa_0} (x,t)
	=   \frac{(\alpha -\eta_1)}{\kappa_0^2} 
	   \frac{\theta_3(0; \tau )}{\theta_3( \tfrac{1}{2}; \tau )} 
	   \frac{ \theta_3(\tfrac{1}{2} + \frac{{\Omega} + \Delta}{2\pi}; \tau)} { \theta_3( \frac{{\Omega} + \Delta}{2\pi}; \tau)}
	   + \bigo{\kappa_0^{-3}}
	= \frac{q_\mathrm{bg}(x,t)}{\kappa_0^2} + \bigo{\kappa_0^{-3}} \ , \\
	&w_{22}(x,t) = 1 +  \frac{1}{4K(m)}\left[ \frac{ \theta_3'(\frac{{\Omega}+  \Delta}{2\pi}; \tau)}{\theta_3(\frac{{\Omega }+  \Delta}{2\pi}; \tau)}  - \frac{ \theta_3'(0;\tau) }{\theta_3(0;\tau)} \right] \frac{\alpha}{\kappa_0} + \bigo{\kappa_0^{-2}} 
\end{aligned}
\end{equation}
imply
\begin{equation}
	X^{^{(-)}}(x,t) =  \frac{\sgn(\chi) e^{-2\kappa_0(x-x_0-4\kappa_0 t)}}{2\kappa_0}\left[1 + \frac{X^{(1)}}{\kappa_0} \right]  + 
	\frac{q_\mathrm{bg}(x,t)}{\kappa_0^2} +  \bigo{\kappa_0^{-3}},
\end{equation}
where $X^{(1)}$ contains the corrections of order $\bigo{\kappa_0^{-1}}$  coming from $f(i\kappa_0)^2$,  $w_{22}^2$, and $e^{i \varphi(i\kappa_0)}$.

Including all terms at leading order we find that as $\kappa_0 \to \infty$
\begin{equation}
\begin{aligned}
	q_\mathrm{sol}(x,t)&= \frac{4 \sgn(\chi) \kappa_0 e^{-2\kappa_0(x-x_0 - 4\kappa_0^2 t)}}{1 + e^{-4\kappa_0(x-x_0 - 4\kappa_0^2 t)}} + \bigo{1}  \\
	&= 2 \sgn(\chi) \kappa_0 \sech(2\kappa_0(x-x_0 - 4\kappa_0^2 t)) + \bigo{1}\ ,
	\end{aligned}
\end{equation}
where generically the $\bigo{1}$ term is a complicated non-vanishing expression.

\section{Asymptotic description of the interaction dynamics  and proof of Theorem~\ref{theorem1.4}}
\label{sec-gas}

After establishing the asymptotic expansion of the soliton+gas solution in \eqref{q.leading} for large times, we now focus on better understanding the interaction dynamics taking place between the soliton gas and the large soliton passing through it.
We will   determine the speed of the soliton on an elliptic background and then we will  identify the location of the soliton peak up to corrections of order $\mathcal {O}(t^{-1})$.

\subsection{The speed of a soliton on a genus-1 background}\label{sec:speed_sol_genus1}

Once a large parameter is introduced --be it large time or large $|\chi|$-- the soliton must lie along a spacetime curve $(x(t),t)$ satisfying 
\[
	\varphi(i \kappa_0;x(t),t)  = \text{constant (of order $1$);}
\]
an analogous argument holds for the soliton gas. Indeed, if this is the case, implicit differentiation automatically gives the value of the phase velocity of the soliton (resp. soliton gas) 
\[
v_{\rm phase} := x'(t) = -\frac{\frac{\d \varphi}{\d t}}{\frac{\d \varphi}{\d x}}\ .
\]
Using the representation \eqref{varphi}, $\varphi(i \kappa_0; x,t)$ can be expressed in terms of elliptic integrals as 
\begin{equation}\label{phi.at.kappa}
	\varphi(i\kappa_0; x,t) = R(i \kappa_0) \left[ 4t i \kappa_0  + \frac{1}{i\kappa_0} \frac{\Pi\left(\frac{\eta_1^2}{\kappa_0^2}, \frac{\eta_1^2}{\alpha^2} \right)}{K\le(\frac{\eta_1^2}{\alpha^2}\ri)} \left( x - 2(\eta_1^2 + \alpha^2) t \right) \right]\ .
\end{equation}

So, the average soliton velocity is given by 
\begin{gather}
	\bar{v}_\mathrm{sol}(\kappa_0) 
	= -\frac{\partial_t \varphi(i\kappa_0;x,t) + \alpha_t\,  \partial_\alpha \varphi(i \kappa_0;x,t)}
	      {\partial_x \varphi(i\kappa_0;x,t) + \alpha_x \, \partial_\alpha \varphi(i \kappa_0;x,t)}
	  =    -\frac{ \varphi_2(i\kappa_0; x,t) }
	      { \varphi_0(i\kappa_0; x,t) } 
	      \end{gather}
since $\partial_\alpha \varphi(k; x,t) \equiv 0$ for any $\alpha$ satisfying the Whitham evolution equation (see Remark \ref{Dalpha=0}); thus, from \eqref{phi.at.kappa}, we have
\begin{gather}
	\bar{v}_\mathrm{sol} (\kappa_0)
	= 4 \kappa_0^2 \frac{K\le(\frac{\eta_1^2}{\alpha^2}\ri)}{\Pi \left(\frac{\eta_1^2}{\kappa_0^2},\frac{\eta_1^2}{\alpha^2} \right)} +  2(\eta_1^2+\alpha^2) \ , \qquad \text{for all } x > 4\eta_1^2 t\ ,
	\label{v.sol}
\end{gather}
where we can recognize the velocity of the elliptic background solution $v_{\rm bg} = 2(\eta_1^2+\alpha^2)$.

We also recall the expression of the function $\varphi$ in terms of the $g$-function (see formula \eqref{varphi.3b}): 
\[ 
\varphi(k;x,t) = g(k;x,t) + x k + 4 t k^3 =  \int\limits_{\Sigma_{1,\alpha} \cup \Sigma_{2,\alpha}} \ln (k-s) \, \rho(s;x,t)\, \d s  + xk + 4t k^3
\]
where 
\[
\rho(k;x,t) = \frac{1}{2\pi i} \frac{(k^2+\alpha^2) \le(x-6t(\alpha^2-\eta_1^2  - 2k^2)\ri)}{R(k)}\ .
\] 
By enforcing $\frac{\d}{\d t}\varphi(k;x(t),t) |_{k = i \kappa_0} =0$ (and recalling that the density $\rho$ vanishes at the end points), we obtain
\[
 \int\limits_{\Sigma_{1,\alpha} \cup \Sigma_{2,\alpha} } \ln (i\kappa_0-s) \, \Big( \rho_t(s) + \rho_x(s) x'(t)\Big)\, \d s  + i\kappa_0 x'(t) + 4 (i\kappa_0)^3 =0\, ,
\]
which, solving for $\bar{v}_{\rm sol}(\kappa_0) = x'(t)|_{k=i \kappa_0}$, yields
\begin{gather}
\bar v_{\rm sol}(\kappa_0) = 4\kappa_0^2 + \frac{1}{i \kappa_0}  \int_{\Sigma_{1,\alpha} }\ln \le| \frac{i\kappa_0-s}{i \kappa_0 +s}\ri| \, \le( v_{\rm group}(s) - \bar v_{\rm sol}(\kappa_0)\ri) \rho_x(s)\, \d s \, ,
\end{gather}
where $v_{\rm group} := -\frac{\rho_t}{\rho_x}$ is the group velocity of the genus-1 background wave and the integral term is real, as $\rho_x \in i\R$.

Furthermore, for $k \in \Sigma_{1,\alpha} $, we also have 
\[
g_+(k) + g_-(k) = 2\int\limits_{\Sigma_{1,\alpha}} \ln\le|  \frac{k-s}{k+s}\ri| \, \rho(s) \, \d s = -8k^3t - 2kx \ .
\]
By taking the derivative of the $g$ function with respect to $t$ and $x$ separately, 
\[
\int\limits_{\Sigma_{1,\alpha}} \ln\le|  \frac{k-s}{k+s}\ri| \, \rho_t(s) \, \d s = -4k^3 \ , \qquad \int\limits_{\Sigma_{1,\alpha}} \ln\le|  \frac{k-s}{k+s}\ri| \, \rho_x(s) \, \d s = - k
\]
and by simple algebra manipulations, we obtain that
\begin{gather}\label{group_vel_kinetic}
v_{\rm group}(k) = -4k^2 + \frac{1}{k}\int_{\Sigma_{1,\alpha}} \ln\le| \frac{k-s}{k+s}\ri| \le( v_{\rm group}(s) - v_{\rm group}(k)\ri) \rho_x(s) \, \d s\ ,
\end{gather}
where we notice that since $k\in i\R_+$ the term $-4k^2>0$ and the second term in the expression above is real ($\rho_x \in i\R$).

\subsection{Locating the soliton peak} \label{subsec:soliton.peak}
{
For simplicity, we will consider here the case of a trial soliton ($\chi >0$). The case of a trial anti-soliton ($\chi<0$) is analogous, but the roles of the critical points $X_1, X_2$ defined by \eqref{crit_points_X1,2} below are reversed.
Suppose  $x_{\rm peak}(t)$ is a function tracking the position of the (peak of the) trial soliton as it evolves in time.
Before the soliton enters the modulated elliptic region, we have $x_{\rm peak}(t) = x_0 + 4\kappa_0^2 t$, which coincides with the position of a trial soliton moving in the vacuum (i.e. without the presence of a gas).
As we want to analyze the situation where the trial soliton is interacting with the modulated elliptic wave, we consider the region of space-time such that
\begin{equation}\label{into_modulated_elliptic_region}
\begin{aligned}
&4 \eta_{1}^{2} + \varepsilon < \frac{x}{t} < v_{2} - \varepsilon, \\
& 4 \eta_{1}^{2} + \varepsilon < \frac{x_{0}}{t} + 4 \kappa_{0}^{2} ,
\end{aligned}
\end{equation}
The first condition defines the region of space time in which the gas behaves as a modulated elliptic wave  (here $\varepsilon$ is a positive constant ensuring that we do not encounter any delicate transitory phenomena occurring at the boundary of the modulated elliptic region) while the second ensures that enough time has passed that the soliton initially at position $x_0\ll -1$ has traversed the quiescent region and has entered the modulated wave region of the soliton gas.  Then we identify two times $t_1$ and $t_2$ characterized by the soliton peak entering and leaving the modulated wave region respectively:
\[ \Big( t_1 := \frac{-x_0}{4(\kappa_0^2-\eta_1^2)}, \ x_1 := 4\eta_1^2 t_1 \Big) 
\qquad \text{and} \qquad 
\Big( t_2, x_{\rm peak}(\, t_2\, )  \Big) \]
for some $ t_2 >  t_1$ such that 
\[ \frac{x_{\rm peak}( \, t_2 \,) }{ t_2} = v_2 \ ,\]
where $v_2$ is the solution of the Whitham equation for $\alpha=\eta_2^2$ (see \eqref{Whitham W}).
For $t>t_2$, the soliton enters the constant elliptic wave region. The same arguments that will be explained below still hold, but the quantity $\alpha$ will be constant for all values of $x/t > v_2$.

Computing the position of the soliton peak as it travels through the elliptic background is complicated by the existence of the many local extrema of the background; computing the maximum as a roots of $q_x = 0$ is inefficient as this also characterizes the (infinitely many) local extrema of the background. 
To account for this, we introduce the following change of coordinates which put us into a slowly varying reference frame, following the characteristic lines of the phase of the background wave. Evolving along these characteristics the background wave is slowly varying. 
This will allow us to identify the unique global maximum (i.e. the soliton peak) via straightforward calculations. 
For $x$ and $t$ satisfying inequalities \eqref{into_modulated_elliptic_region}, we introduce the following coordinate system: 
\begin{equation}\label{change_coords}
	\begin{aligned}
		s(x,t) = s  &:= {\Omega}(x,t) + \Delta(x/t)\ ,  \\ 
		\tau(t) = \tau &:= t.
	\end{aligned} 
\end{equation}

\begin{lemma}\label{lem:coords}
The mapping \eqref{change_coords} is well defined for $(x,t)$  satisfying \eqref{into_modulated_elliptic_region},  and it is invertible:
\[
\det \frac{\partial (s,\tau)}{\partial (x,t)} = \frac{\partial s}{\partial x} \neq 0 \ .
\]
\end{lemma}
\begin{proof}
The inequalities \eqref{into_modulated_elliptic_region} guarantee that ${\Omega}$ is well defined, therefore the change of coordinates makes sense. 
From \eqref{omega.phi.ointB} and the fact that $\partial_\alpha \varphi (k;x,t) \equiv 0$ for $\alpha$ satisfying the Whitham equation, we have $\partial_\alpha \Omega = 0$, therefore 
\begin{equation}
\begin{aligned}
\frac{\partial s}{\partial x} &= \frac{\partial }{\partial x} \le[ t \Omega_{2} (\alpha(v)) + x \Omega_{0}(\alpha(v)) + \Delta(v) \ri] \\
&=  \Omega_0(\alpha(v)) + \frac{1}{t}\Delta'(v) = \frac{\pi \alpha(v)}{K(m)}+ \frac{1}{t}\Delta'(v)\ .
\end{aligned}
\end{equation}
where $'$ refers to the derivative with respect to the variable $v := x/t$.
Notice now that $\Omega_0 = \frac{\pi \alpha}{K(m)}>0$ for $m \in [0,1)$ (see \eqref{omega.1}) and $\Delta(v)$ is bounded, therefore $\frac{\partial s}{\partial x} >0$, for large enough $t$ and $\alpha >\eta_1$. 
\end{proof}

At $t=t_1$ and for any $x> x_1$, we can identify a corresponding value of $ s = s(x, t_1)$ and follow the characteristic line $ s = $ constant. 
By keeping the value of $s$ fixed, we then seek the maximizer $\tau(s)$ for the leading order behavior \eqref{q.leading} for the solution $q(s,\tau)$.
It is clear that as we change $\tau$, the quantity $x$, as well as $\alpha = \alpha(x(s,\tau)/\tau)$, deform so as to hold $s$ fixed, therefore in the physical coordinates $(x,t)$, this setting will account for the background gas $q_{\rm bg}$ to be slowly varying, while the soliton $q_{\rm sol}$ is passing through it.
}

 The change of coordinates \eqref{change_coords} gives  
\begin{equation}\label{coordinate.derivatives}
\begin{aligned}
	&\pd{x}{s} = \frac{1}{\Omega_x + t^{-1} \Delta'}, && \pd{t}{s} =0 & \qquad && &\pd{}{s} = \frac{1}{{{\Omega}_x} + t^{-1} \Delta'} \frac{\partial}{\partial x} \ ,      \\
	&\pd{x}{\tau} = -\frac{\Omega_t - t^{-1} v \Delta'}{\Omega_x + t^{-1} \Delta'}, && \pd{t}{\tau} =1
	&\qquad && &\pd{}{\tau} = \pd{}{t} + \pd{x}{\tau} \pd{}{x}
\end{aligned}
\end{equation}
In particular, $\partial_\tau v =  \tau^{-1} \le( - v + x_\tau \ri) = \bigo{\tau^{-1}}$.

From the explicit formula of $q_{\rm bg}$ \eqref{q.bg.recovery}  and $q_{\rm sol}$ \eqref{q.sol.1}  and the RH problem \ref{rhp:W}, it is clear that the background and soliton components of \eqref{q.leading} can be written as
\begin{equation}\label{F_sol_bg}
	q_{\mathrm{bg}}(x,t) = F_\mathrm{bg}(s, \alpha(v))\ , 
	\qquad
	q_{\mathrm{sol}}(x,t) = F_\mathrm{sol} (s, \alpha(v), \hat \chi (x,t))\ ,
\end{equation}
where $\hat \chi (x,t) := \chi f_-(i\kappa_0;\alpha(v))^2 e^{-2i\varphi(i\kappa_0;x,t)}$.

Following similar calculations as in the Section above, we can simplify the expression \eqref{Q_minus_4.16} of $\mathcal{Q}^{^{(-)}}(x,t)$ appearing in the definition of $q_{\rm sol}$:
\begin{gather}
\begin{multlined}
	\mathcal{Q}^{^{(-)}}(x,t) 
	= \frac{ (\gamma(i \kappa_0)^2 -1) }{(\gamma(i \kappa_0)^2 + 1)} \frac{ \theta_3(A(i\kappa_0) - \tfrac{1}{4} - \tfrac{ {\Omega} + \Delta}{2\pi}; \tau) }{ \theta_4( A(i\kappa_0) - \tfrac{1}{4} - \tfrac{ {\Omega} + \Delta}{2\pi}; \tau)}  \frac{ \theta_3( A(i\kappa_0) + \tfrac{1}{4} ; \tau) }{ \theta_4( A(i\kappa_0) + \tfrac{1}{4} ; \tau ) } \\
	=  \frac{( \gamma(i \kappa_0)^2 -1 )}{(\gamma(i \kappa_0)^2 + 1)} \frac{ \theta_3(0;\tau)^2}{\theta_4(0;\tau)^2} 
	\frac{1 - \sqrt{m} \sn^2 \left( 2K(m)(A(i\kappa_0) - \tfrac{1}{4} - \tfrac{{\Omega}+\Delta}{2\pi}), m \right) }
	  {1 + \sqrt{m} \sn^2 \left( 2K(m)(A(i\kappa_0) - \tfrac{1}{4} - \tfrac{{\Omega}+\Delta}{2\pi}) , m \right)}  
	\frac{1 - \sqrt{m} \sn^2 \left( 2K(m)(A(i\kappa_0) + \tfrac{1}{4} ), m \right) }
	  {1 + \sqrt{m} \sn^2 \left( 2K(m)(A(i\kappa_0) + \tfrac{1}{4} ) , m \right)}  
\end{multlined} \nonumber 
\shortintertext{which simplifies to}
\label{Q- decent}
	\mathcal{Q}^{^{(-)}}(x,t) 
	=  \frac{( \gamma(i \kappa_0)^2 -1) }{(\gamma(i \kappa_0)^2 + 1)} \frac{( \alpha + \eta_1)}{(\alpha - \eta_1)} 
	\dn \Big( 2K(m_1) ( A(i\kappa_0) + \tfrac{1}{4} ), m_1 \Big)
	\dn \left( 2K(m_1) ( A(i\kappa_0) - \tfrac{1}{4} - \tfrac{ {\Omega} + \Delta}{2\pi} ), m_1 \right).
\end{gather}
The first elliptic function $\dn$ in the expression above is a slowly evolving envelope oscillation and since $-\tfrac{1}{4} < A(i\kappa_0)  \leq 0$,
$\dn \left( 2K(m_1) ( A(i\kappa_0) + \tfrac{1}{4} ), m_1 \right) \in \left( (1-m_1)^{1/4}, \ 1 \right)$. The second $\dn$ has a period of $\bigo{1}$, therefore $\dn \left( 2K(m_1) ( A(i\kappa_0) - \tfrac{1}{4} - \tfrac{ {\Omega} + \Delta}{2\pi} ), m_1 \right) \in \le[ \sqrt{1-m_1} , 1\ri]$.
Since $0 < \gamma(i\kappa_0) < 1$ and $\alpha > \eta_1$, it follows that $\mathcal Q^{^{(-)}}(x,t)$ is strictly negative; more precisely,
\begin{equation}
	\mathcal{Q}^{^{(-)}}(x,t) \in 
 \left( - \frac{ 1- \gamma(i \kappa_0)^2}{1+ \gamma(i \kappa_0)^2}  \left( \frac{ \alpha + \eta_1}{\alpha - \eta_1} \right), - \frac{ 1- \gamma(i \kappa_0)^2}{1+ \gamma(i \kappa_0)^2}  \sqrt{ \frac{ \alpha - \eta_1}{\alpha + \eta_1} }\, \right) \ .
	 \end{equation}
	Furthermore, we notice that while the right endpoint of the range of $\mathcal Q^{^{(-)}}$ is bounded within the interval $(-1,0)$ for $\kappa_0 >\alpha$ and $\alpha \in (\eta_1,\eta_2]$, the left endpoint is a strictly decreasing function of $\kappa_0$ for $\kappa_0> \alpha$, with absolute minimum $-\frac{ \alpha + \eta_1}{\alpha - \eta_1} <-1$ for $\kappa_0 \to \alpha$.
	Therefore, there exists a unique value $\kappa_{\rm crit}$ such that for any $\kappa_0 < \kappa_{\rm crit}$, the equation
	\[
		\mathcal Q^{^{(-)}}(x,t) = -1
	\]
	will have two solutions per period. Such a value can be computed explicitly:
	\begin{equation}\label{kappa.c.Q-1}
		\kappa_{\rm crit} = \alpha \left[ \frac{ 1 + m + \sqrt{ (1+m)^2 +4\sqrt{m}(1+\sqrt{m})^2} }{2(1+\sqrt{m}) } \right]\ .
	\end{equation}

From now on we will assume $\kappa_0 > \kappa_{\rm crit}$. Applying the change of variables \eqref{change_coords} to \eqref{F_sol_bg}, we have 
\begin{equation}
\begin{aligned}
	&\pd{q_{\mathrm{bg}} } {\tau} = \partial_\tau{v} \, \pd{\alpha}{v} \pd{F_{\mathrm{bg}} }{\alpha} = \bigo{\tau^{-1}}, \\
	&\pd{q_{\mathrm{sol}} } {\tau} = \partial_{\tau}v \, \pd{\alpha}{v} \pd{F_{\mathrm{sol}} }{\alpha} 
	   + \pd{\hat \chi}{\tau} \pd{F_{\mathrm{sol}} }{\hat \chi} = \pd{\hat \chi}{\tau} \pd{F_{\mathrm{sol}} }{\hat \chi} + \bigo{\tau^{-1}} \ ,
\end{aligned}
\end{equation}
and 
\begin{align}
	\pd{\hat \chi }{\tau} &= \partial_{\tau}v \, 2\chi f_-(i\kappa_0;\alpha(v))  e^{-2i\varphi(i\kappa_0;x,t)} \pd{\alpha}{v} \pd{f_-}{\alpha}  -2i \hat \chi \left( \varphi_t - \frac{\Omega_{t} - t^{-1}v \Delta' }{\Omega_{x} + t^{-1}\Delta'} \varphi_x \right)  \nonumber \\
	&=  -2i \hat \chi \left( \varphi_t - \frac{\Omega_{t}}{\Omega_{x}} \varphi_x \right) + \bigo{\tau^{-1}}
	=  -8 \kappa_0 R(i \kappa_0; x(s,\tau),\tau) \hat \chi + \bigo{\tau^{-1}},
\end{align}
where we have used \eqref{omega.1} and \eqref{phi.at.kappa}--\eqref{v.sol} to simplify the expression. 
Therefore, for all sufficiently large $\tau$, we have
\begin{equation}
	\begin{aligned}
	\pd{q}{\tau} 
	  &= -8 \kappa_0 R(i \kappa_0; x(s,\tau),\tau) \hat \chi \pd{ F_{\mathrm{sol}}}{ \hat \chi } + \bigo{\tau^{-1}} \\
	  &= -\frac{8 \kappa_0 R(i \kappa_0; x(s,\tau),\tau) }{w_{22}(x(s,\tau),\tau)^2 \hat \chi} 
	  \frac{2 \le[ 1-(\mathcal{Q}^{^{(-)}})^2 \ri] (X^{^{(-)}} - X_1)(X^{^{(-)}}-X_2)}{\le( (X^{^{(-)}})^2+(Y^{^{(-)}})^2 \ri)^2}+ \bigo{\tau^{-1}},
	\end{aligned}
\end{equation}
where
\begin{equation}\label{crit_points_X1,2}
	X_1 = \frac{ 1 - \mathcal{Q}^{^{(-)}}}{1+ \mathcal{Q}^{^{(-)}}} Y^{^{(-)}} \ , 
	\qquad
	X_2 = -\frac{ 1 + \mathcal{Q}^{^{(-)}}}{1- \mathcal{Q}^{^{(-)}}} Y^{^{(-)}} \ ,
\end{equation}
and $X^{^{(-)}}$ and $Y^{^{(-)}}$ are given in \eqref{XY.sol.-}. 
Thus, up to $\mathcal O (\tau^{-1})$ terms, we found two critical points of $q$, given by the implicit equations $X^{^{(-)}}=X_1$ and $X^{^{(-)}}=X_2$. 

By a similar calculation, the second derivative of $q(s,\tau)$ evaluated at $X_1$ and $ X_2$ is
\begin{equation}
	\begin{aligned}
	\pd[2]{q}{\tau} \bigg\vert_{X^{(-)}=X_\ell}  
	&= (-1)^{\ell} \left( \frac{8 \kappa_0 R(i \kappa_0; x(s,\tau),\tau) }{w_{22}(x(s,\tau),\tau)^2 \hat \chi}  \right)^2 \frac{ 2 \le[ 1-( \mathcal{Q}^{^{(-)}})^2 \ri] (X_1 -X_2)}{(X_\ell^2+{Y^{(-)}}^2)^2} + \bigo{\tau^{-1}} \\
	&= (-1)^{\ell} \left( \frac{8 \kappa_0 R(i \kappa_0; x(s,\tau),\tau) }{w_{22}(x(s,\tau),\tau)^2 \hat \chi}  \right)^2 \frac{ 8\kappa_0 (Y^{^{(-)}})^2}{\le( X_\ell^2+(Y^{^{(-)}})^2 \ri)^2} + \bigo{\tau^{-1}} \ ,
	\end{aligned}
\end{equation} 
$\ell = 1,2$, which is strictly nonzero. Thus $q$ has an (approximate) maximum at $X_1$ and minimum at $X_2$.  Such  solutions $X_1, X_2$ are still expressed in terms of the dependent variable $\mathcal Q^{^{(-)}}$ and in the moving frame $(s,\tau)$, however existence of a global maximum in the physical $(x,t)$-coordinates is established in  
Theorem~\ref{theorem1.4} that we are going to prove. For our  convenience,    we state  and prove   below Theorem~\ref{thm:xpeak}  that is a  more extended version of Theorem~\ref{theorem1.4}.

\begin{thm}\label{thm:xpeak}
For $\chi>0$ and $x_0 \ll -1$, there exists a $\tilde{\kappa}>  \kappa_{\rm crit}$ (see \eqref{kappa.c.Q-1}) , such that for all  \color{black} $\kappa_0 > \tilde{\kappa}$  \color{black} there exists a unique, continuous, global maximum $x_\mathrm{peak}(t)$ of the solution $q$ which identifies the position of the soliton peak for all $t>0$.
For any $(x,t)$ satisfying the inequalities \eqref{into_modulated_elliptic_region},  $x_\mathrm{peak}(t) = x^*(t) + \bigo{t^{-1}}$ where $x^*(t)$ is implicitly defined as the solution of 
\begin{equation}\label{x* implicit.eq}
	X^{^{(-)}}(x^*(t), t) =\frac{ 1 - \mathcal{Q}^{^{(-)}}(x^*(t),t)}{1+ \mathcal{Q}^{^{(-)}}(x^*(t),t)} \frac{1+\mathcal{Q}^{^{(-)}}(x^*(t),t)^2}{2\kappa_0}\ .
\end{equation}
where $\mathcal{Q}^{^{(-)}}$ is given in \eqref{Q- decent}. The amplitude of the solution at the maximum is given by 
\[
q(x_{\rm peak}(t),t) = q_{\rm bg}(x^*(t),t) + 2\kappa_0 \left[1 + \frac{ 2 \mathcal{Q}^{^{(-)}}(x^*(t),t)}{1+  \mathcal{Q}^{^{(-)}}(x^*(t),t)^2} \right] + \bigo{t^{-1}}.
\]

This global maximum, $x_{\rm peak}(t)$, is strictly increasing, and satisfies: 
\begin{itemize}
\item[(i)] for $t\in (0, t_1)$, $x_{\rm peak}(t) = x_0 + 4\kappa_0^2 t$; 
\item[(ii)] for $t >(1+ \epsilon)t_1$ (for some small positive $\epsilon$),
\begin{equation}
\begin{aligned}
& \dot x_{\rm peak}(t) =  - \frac{2\varphi_{2}(i \kappa_{0}) - \partial_t \ln\Psi(x,t; \kappa_0, \eta_1) } { 2\varphi_{0}(i \kappa_{0})- \partial_x \ln\Psi(x,t; \kappa_0, \eta_1) }\Big\vert_{x = x_{\rm peak}(t)} + \bigo{ \frac{1}{t}} \ .
\end{aligned}
\end{equation}
Here $\Psi$ is defined by \eqref{Psi.speed}.

Finally, for $t>(1+\epsilon)t_1$, let $T$ be the time it takes the soliton peak $x_{\mathrm{peak}}(t)$ to traverse one period of the elliptic background wave $q_{\mathrm{bg}}(x,t)$. Then the average velocity of the soliton peak over the period satisfies
\begin{equation}\label{eq.vsolbar}
	\frac{1}{T} \int_t^{t+T} \dot{x}_{\mathrm{peak}}(s)\,  \d s = 
	\frac{ x_{\mathrm{peak}}( t + T) - x_{\mathrm{peak}}(t)}{T} 
	= \bar{v}_\mathrm{sol}(\kappa_0) + \bigo{t^{-1}}, \quad  \bar{v}_\mathrm{sol}(\kappa_0):=- \frac{\varphi_2(i\kappa_0)}{\varphi_0(i\kappa_0)} \ .
\end{equation}	
Moreover, the leading order term for the averaged soliton velocity satisfies 
\begin{eqnarray*}
\bar{v}_{\rm sol} (\kappa_0)=  4 \kappa_0^{2} + \frac{1}{ \kappa_0} \int_{\eta_1}^{\alpha} \ln{
\left|
\frac{\kappa_0-s}{\kappa_0+s}
\right|
}(v_{{\rm group}}(s) - \bar{v}_{\rm sol} (\kappa_0) ) \, \partial_x \rho(is) \,  \d s \ .
\end{eqnarray*}
which we recognize as the kinetic equation for the \textit{average} velocity of an mKdV soliton analogous to those described in \cite{El03, ElKa05, Tovbis2} for other evolution equations.
\end{itemize}

\end{thm}

\begin{remark}
The condition $\chi>0$, in the above theorem is strictly for convenience. Physically, this means we are tracking the position of a trial soliton as it interacts with the soliton gas. The cases $\chi<0$, corresponds to an anti-soliton of mKdV, and an analogous result can be proved where the equation for the global minimum; specifically one can show $x_{\mathrm{min} }(t) = x_*(t) + \bigo{t^{-1}}$ where $x_*(t)$ is implicitly encoded in the condition $X^{^{(-)}}(x_*(t),t) = X_2(x_*(t),t)$ (c.f. \eqref{crit_points_X1,2}).
\end{remark}

\begin{remark}
The condition that $\kappa_0$ be sufficiently large in Theorem~\ref{thm:xpeak} is not merely a technical condition. 
For $i\kappa_0 \approx i \eta_2$, we see numerically that the position of the global maximum $x_{\mathrm{peak}}(t)$ is not continuous. One observes, that as the soliton attempts to crest a peak of background wave it loses amplitude while the incident peak of the background peak grows; if $\kappa_0$ is sufficiently small, then there exists a critical time at which two amplitudes---of the soliton and the incident peak of the background---are equal. Evolving past the critical time, the global maximum jumps forward emerging out of the incident peak, while the former maximum settles into the background wave. See Figure~\ref{fig:peakjump}. This phenomenon was already foreseen and analysed by Lax in \cite{Lax68}, for a simple two-soliton interaction.
\end{remark}

\begin{proof}[Proof of Theorem~\ref{thm:xpeak}]
The behavior of the peak for $x < t_1-\epsilon$ is trivial since $q(x,t) = 2\kappa_0 \sech(2\kappa_0(x-x_0 -4\kappa_0^2 t)) + \bigo{ e^{-c t}}$ for some constant $c>0$. 
For $(x,t)$ in the modulation region we use \eqref{chi0}  and \eqref{XY.sol.-}, to rewrite the equation $X^{^{(-)}}=X_1$ in the new $(s,\tau)$ coordinates as
\begin{eqnarray}\label{eq-5.26}
P(s,\tau) := 2 i \varphi(i \kappa_{0}; x(s,\tau),\tau) + 2\kappa_{0} x_{0}  - \ln \Psi(x(s,\tau), \tau; \kappa_0, \eta_1) = 0,
\end{eqnarray}
where 
\begin{equation}\label{Psi.speed}  
	\Psi(x,t) := f(i\kappa_0)^2  w_{22}(x,t)^2   \le[  \frac{ 1 - \mathcal{Q}^{^{(-)}}(x,t)}{1+ \mathcal{Q}^{^{(-)}}(x,t)} \left[1 + \mathcal{Q}^{^{(-)}}(x,t)^2\right] - 2 \kappa_{0}\mathcal{Q}^{^{(-)}}_{\kappa_0}(x,t)  \ \ri]  
\end{equation}
Using \eqref{phi.at.kappa}, we can expand 
\begin{multline}\label{P.2}
P(s,\tau)  =  2\tau  \sqrt{(\kappa_0^2-\alpha^2)(\kappa_0^2-\eta_1^2)} 
 \left[ 4  \kappa_0  - \frac{1}{\kappa_0} \frac{\Pi\left(\frac{\eta_1^2}{\kappa_0^2}, m \right)}{K(m)} \left( \frac{x(s,\tau)}{\tau} - 2(\eta_1^2 + \alpha^2)  \right) \right] \\
 + 2\kappa_0x_0 - \ln \Psi(s,\tau; \kappa_0, \eta_1) 
\end{multline}
Here the reason for the change of variables $(x,t) \mapsto(s,\tau)$ reveals itself: it's straightforward to check that $\partial_\tau \Psi(s,\tau) = \bigo{\tau^{-1}}$.
Furthemore from Remark~\ref{Dalpha=0}  we have $\partial_\alpha\varphi=0$ so that a short computation 
 shows that 
\[
\partial_x\varphi=-\frac{  \Omega_x }{ \Omega_t}(4i\kappa_0R(i\kappa_0)-\varphi_t)
\]
and  \eqref{coordinate.derivatives} and \eqref{omega.1} imply $\pd{x}{\tau} = - \Omega_t / \Omega_x + \bigo{t^{-1}} $
so that
\[
\partial_\tau P(s,\tau) = 2i[\partial_x\varphi\partial_\tau x+ \partial_t\varphi]-\frac{\partial_\tau\Psi}{\Psi}=8\kappa_0 \sqrt{(\kappa_0^2-\alpha^2)(\kappa_0^2-\eta_1^2)} + \bigo{t^{-1}} >0.
\]
As the first bracketed term in \eqref{P.2} is linear in $\tau$, while the remaining terms are bounded ensures that the equation $P(s,\tau) = 0$ has a unique solution $\tau^*(s)$ for each fixed $s$ in the modulation region. 

Now differentiating with respect to $s$ one finds that 
\[
	\partial_s P = -2 \frac{\sqrt{(\kappa_0^2-\alpha)(\kappa_0^2-\eta_1^2)}}{\pi \alpha \kappa_0 } \Pi \left( \frac{ \eta_1^2}{\kappa_0^2}, \, \frac{ \eta_1^2}{\alpha^2} \right) + \partial_s \ln \Psi(x(s,\tau),\tau;\kappa_0,\eta_1)
\]
where $\partial_s \Psi$ is quite complicated but  bounded. The first term in $\partial_s P$ is an increasing, asymptotically linear, function of $\kappa_0$, while the second term is a bounded function of $\kappa_0$. In fact, the calculation in Section~\ref{sec:kappa.to.infty} show that $\partial_s \ln \Psi = \bigo{\kappa_0^{-1}}$ as $\kappa_0 \to \infty$. So for all sufficiently large $\kappa$ we have $\partial_s P \neq 0$, so we can invert the unique solution $\tau^*(s)$ to get $s^*(\tau)$. The function $x^*(t)$ is then given by $x^*(t) = x( s^*(t),t))$.
This establishes the existence of the solution $x^*(t)$ of \eqref{x* implicit.eq}. 
The existence of the global maximum $x_{\rm peak}(t)$ for $q(x,t)$ then follows from the Implicit Function Theorem, as $\partial^2_{\tau} q$ does not vanish at $X^{^{(-)}}=X_1$.

We then compute an expansion for $\dot x_{\rm peak}$ by first observing that $\dot x_{\rm peak} = \dot x^* + \bigo{t^{-1}}$, and $\dot x^*$ can be computed by differentiating the condition $P(x^*(t),t) = 0$ with respect to $t$:
\[
	\dot x^* ( 2i\varphi_x - \partial_x \ln \Psi) + 2i\varphi_t - \partial_t \ln \Psi  =0.
\]
Solving for $\dot x^*$ gives the desired result. 

To complete the proof we consider, for $t > t_1+\epsilon$, the average of the soliton velocity $\dot{x}_{\mathrm{peak}}(t)$ over a period of the background wave $q_\mathrm{bg}(x,t)$. First we observe that as $x_{\mathrm{peak}}(t) = x^*(t) + \bigo{t^{-1}}$, the leading order term $\bar{v}_\mathrm{sol}(\kappa_0)$ in \eqref{eq.vsolbar} can be computed from the difference of $x^*(t)$ over a period of the background. 
From \eqref{eq-5.26} we have
\begin{equation}\label{avg.speed.constraint}
	2i \varphi(i\kappa_0; x^*(t), t) = -2\kappa_0 x_0 + \ln \Psi(x^*(t), t; \kappa_0, \eta_1)
\end{equation}
for all $t > t_1+\epsilon$. When $t > t_2$, the right hand side of this equation is periodic. This implies that $\varphi(i\kappa_0; x^*(t+T), t+T) = \varphi(i\kappa_0; x^*(t), t)$ or equivalently
\[
	 (t+T)\varphi_2(i\kappa_0)  + x^*(t+T) \varphi_0(i\kappa_0) = 
	 t\varphi_2(i\kappa_0)  + x^*(t) \varphi_0(i\kappa_0)
\]
solving this for the difference $x^*(t+T) -x^*(t)$ gives the result. 
For $t \in (t_1+\epsilon, t_2]$, $x^*(t)$ still satisfies \eqref{avg.speed.constraint}, but the right hand side is only quasi-periodic due to the slow modulation of the wave parameters as $\alpha(x^*(t)/t)$ varies for $t \in (t_1, t_2)$. However, the evolution of $\Psi$ with respect to the slow parameters only introduces changes of order $\bigo{t^{-1}}$ since $\partial_\tau \Psi = \bigo{t^{-1}}$. So treating $\Psi$ as exactly periodic in the modulation domain only introduces errors at the level of the first correction $\bigo{t^{-1}}$. 
\end{proof}

\begin{figure}
\begin{minipage}[c]{.6\textwidth}
\centering
\hspace*{\stretch{0.5}}
\begin{overpic}[scale=0.35]{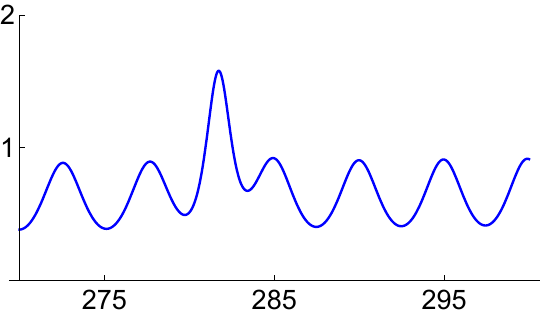}
\put(7,50){\tiny ${}^{\Delta t=-0.5}$ }
\put(99,1){\tiny $x$ }
\end{overpic}
\hspace*{\stretch{1}}
\begin{overpic}[scale=0.35]{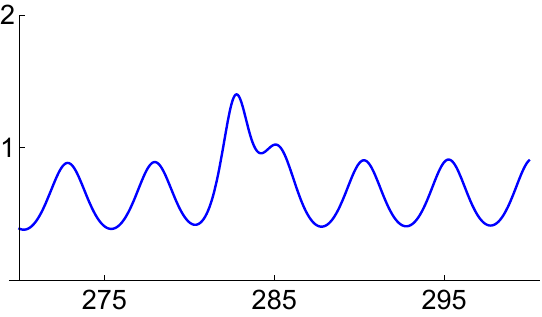}
\put(7,50){\tiny ${}^{\Delta t=-0.2}$ }
\end{overpic}
\hspace*{\stretch{1}}
\begin{overpic}[scale=0.35]{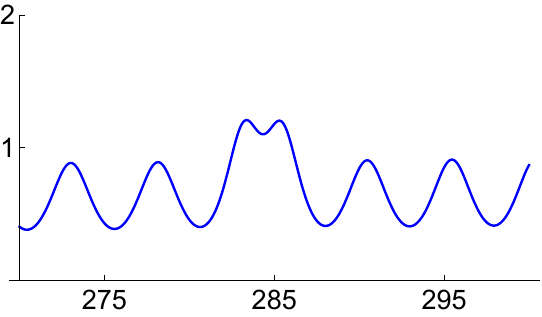}
\put(7,50){\tiny ${}^{\Delta t=0}$ }
\end{overpic}
\hspace*{\stretch{0.5}}
\\[1em]
\begin{overpic}[scale=0.35]{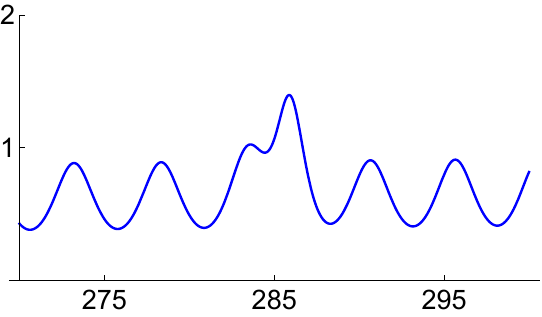}
\put(7,50){\tiny ${}^{\Delta t=0.2}$ }
\end{overpic}
\hspace*{1cm}
\begin{overpic}[scale=0.35]{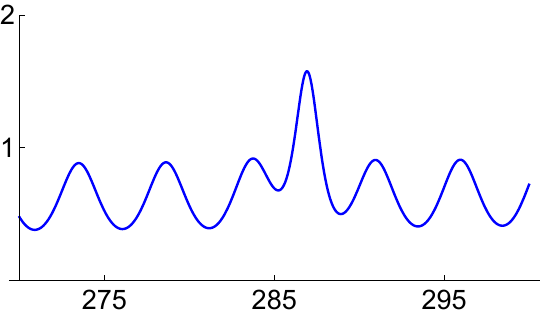}
\put(7,50){\tiny ${}^{\Delta t=0.5}$ }
\end{overpic} 
\end{minipage}
\begin{minipage}[c]{.35\textwidth}
\vspace{0pt}
\begin{overpic}[scale=0.65]{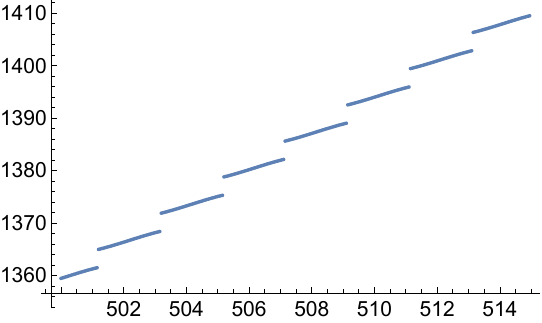}
\end{overpic}
\end{minipage}
\caption{Demonstration of discontinuities in the location of $x_{\mathrm{peak}}(t)$ when the soliton eigenvalue $i\kappa_0$ is near the soliton gas spectral band $[i\eta_1, i\eta_2]$. \textit{Left}: Evolution of $q(x,t)$ at times $t^*+\Delta t$ around a time $t^*=503.17$ at which the global max becomes multivalued. \textit{Right}: Computation of $x_{\mathrm{peak}}(t)$ as a function of time. In both figures we have $\eta_1=0.25, \eta_2=0.75, \kappa=0.8$ and $\chi=1.6e^{-320}$.}
\label{fig:peakjump}
\end{figure}

\appendix

\section{Solvability of the Riemann--Hilbert problem for {\bf \textit{X}}}
\label{appB_solvability}

\begin{thm}
Given a function $r \in L^2(\Sigma_1)$ Riemann--Hilbert problem \ref{rhp:X} is uniquely solvable for all $(x, t) \in \R^2$.
Moreover, the function $q(x,t)$ defined in \eqref{q_via_X} is a classical solution to the mKdV equation \eqref{mkdv}, which belongs to the class $C^\infty(\R_x \times \R_t)$.
\end{thm}

\begin{proof}
First we prove the existence of solutions to RHP~\ref{rhp:X}.
Let $\widehat\Sigma_1$ and $\mathcal{C}_{i\kappa_0}$ be non-intersecting, simple closed closed in $\C^+$ enclosing $\Sigma_1$ and $k=i\kappa_0$ respectively. Orient the contours positively. Let $\widehat\Sigma_2$ and $\mathcal{C}_{-i\kappa_0}$ be complex conjugate contours (with negative orientation) so that $\Gamma_Y := \widehat\Sigma_1 \cup \widehat\Sigma_2 \cup \mathcal{C}_{i\kappa_0} \cup \mathcal{C}_{-i\kappa_0}$ is Schwarz symmetric, $\overline{\Gamma}_Y = \Gamma_Y$. 

Given a function $r:\Sigma_1 \to \R$ in $L^2(\Sigma_1)$ define the analytic function
\begin{equation}
	F(k) = \frac{1}{2\pi i} \int_{\Sigma_1} \frac{ i r(\xi)}{\xi-k} \, \d \xi, \qquad k \in \C \setminus \Sigma_1.
\end{equation}
For $k \in \Sigma_1$, it follows from the Sokhotski--Plamelj formula that $F_+(k) - F_-(k) = r(k)$. Morover, since $r$ is real-valued we have
\begin{equation}
	\overline{F(\bar k)} = - \frac{1}{2\pi i} \int_{\Sigma_2} \frac{ i \overline{r(\bar \xi)}}{\xi-k} \, \d \xi , \qquad k \in \C \setminus \Sigma_2,
\end{equation}
where the minus sign comes from the fact that $\Sigma_2 = \overline{\Sigma}_1$ as a set but has opposite orientation. 

Using $F$, make the following transformation $\bm{X} \mapsto \bm{Y}$, which deforms the jumps of $\bm{X}$ onto the closed contour $\Gamma_Y$ where the jump matrices are analytic. Define
\begin{equation}
	\bm{Y}(k;x,t) = \begin{dcases}
		\bm{X}(k;x,t) \begin{bmatrix} 1 & 0 \\ -F(k) e^{-2i\theta(k;x,t)} & 1 \end{bmatrix} & k \in \inside \left(\widehat{\Sigma}_1 \right) \\
		\bm{X}(k;x,t) \begin{bmatrix} 1 & 0 \\ \frac{i\chi}{k-i\kappa_0}e^{-2i\theta(k;x,t)} & 1 \end{bmatrix} & k \in \inside \left(\mathcal{C}_{i\kappa_0} \right)\\
		\bm{X}(k;x,t) \begin{bmatrix} 1 &  \overline{F(\bar k)} e^{2i\theta(k;x,t)} \\ 0  & 1 \end{bmatrix} & k \in \overline{\inside \left(\widehat{\Sigma}_1 \right)} \\
		\bm{X}(k;x,t) \begin{bmatrix} 1 & \frac{i\chi}{k-i\kappa_0}e^{2i\theta(k;x,t)} \\ 0 & 1 \end{bmatrix} & k \in \overline{\inside \left(\mathcal{C}_{i\kappa_0}  \right)}\\
		\bm{X}(k;x,t) & \text{elsewhere} 
	\end{dcases}
\end{equation}

A straight forward calculation shows that the new unknown $\bm{Y}$ satisfies the following problem 
\begin{RHP}\label{rhp:Y}
Find a $2\times 2$ matrix-valued function $\bm{Y}(\, \cdot\,; x,t)$ with the following properties: 
	\begin{enumerate}[label=\arabic*.]
		\item $\bm{Y}(k;x,t)$ is holomorphic for $k \in \C \setminus \Gamma_Y$.
		\item $\bm{Y}(k;x,t) = \bm I + \bigo{k^{-1}}$ as $k \to \infty$,
		\item For $k \in \Gamma_Y$, the boundary values $\bm{Y}_\pm(k;x,t)$ satisfy $\bm{Y}_+(k) = \bm{Y}_-(k) \bm{J}_Y(k;x,t)$, where
		\begin{gather}\label{Yjumps}
		   \bm{J}_Y(k;x,t) =
		  \begin{dcases} 
		    \begin{bmatrix} 1 & 0 \\ F(k) e^{-2i \theta(k;x,t)} & 1 \end{bmatrix}, & k \in \widehat{\Sigma}_1, \\
		    \begin{bmatrix} 1 &  \overline{F(\bar k)} e^{2i \theta(k;x,t)} \\ 0 & 1 \end{bmatrix}, & k \in \overline{\widehat{\Sigma}}_1, \\
		    \begin{bmatrix} 1 & 0 \\ -\frac{i\chi}{k-i\kappa_0}e^{-2i\theta(k;x,t)} & 1 \end{bmatrix} & k \in \mathcal{C}_{i\kappa_0} \\ 
		    \begin{bmatrix} 1 & \frac{i\chi}{k-i\kappa_0}e^{2i\theta(k;x,t)} \\ 0 & 1 \end{bmatrix} & k \in \overline{\mathcal{C}_{i\kappa_0}}
		  \end{dcases}
		\end{gather}
	\end{enumerate}
\end{RHP}
The jump matrix $\bm{J}_Y(k;x,t)$ is analytic for $k \in \Gamma_Y$, and the symmetries $\Gamma_Y = \overline{\Gamma}_Y$ and $\bm{J}_Y(k;x,t) = \bm{J}(\bar k;x,t)^\dagger$ are clearly satisfied, so Zhou's vanishing lemma \cite[Theorem 9.3]{Zhou} can be applied to conclude that a unique solution of RH problem~\ref{rhp:Y} exists. Inverting the transformation from $\bm{X}$ to $\bm{Y}$ it follows that there exists a unique solution $\bm{X}(k;x,t)$ of RH problem~\ref{rhp:X} as well. 
Uniqueness of the solution follows from a standard Liouville type argument. 

To see that the solution $\bm{Y}(k;x,t)$ has derivatives of all orders in $x$ and $t$ one can differentiate the jump relation \eqref{Yjumps} to derive non-homogenous RH problems for the derivatives $(D^n \bm Y)(k;x,t)$ where $D$ denotes either $\pd{}{t}$ or $\pd{}{x}$. The resulting RH problem takes the from:
\begin{RHP}\label{rhp:DY}
Given $n \in \mathbb{N}$, find a $2\times 2$ matrix-valued function $D^n \bm{Y}(\, \cdot\,; x,t)$ with the following properties: 
\begin{enumerate}[label=\arabic*.]
	\item $\left( D^n \bm{Y} \right)(k;x,t)$ is holomorphic for $k \in \C \setminus \Gamma_Y$.
 	\item $\left( D^n \bm{Y} \right)(k;x,t) = \bigo{k^{-1}}$,  as $k \to \infty$ 
	\item For $k \in \Gamma_Y$, the boundary values of $\left( D^n \bm{Y} \right)(k;x,t)$ satisfy the jump relation
	\begin{equation}
	\begin{gathered}
	\left(D^n \bm Y\right)_+(k;x,t) = \left(D^n \bm Y\right)_-(k;x,t) \bm{J}_Y(k,;x,t) + {\bm{\mathcal{F}}}_n(k;x,t),  \\
	{\bm{\mathcal{F}}}_n(k;x,t) := \sum_{\ell=1}^n \binom{n}{\ell} \left(D^{n-j} \bm{Y} \right) (k;x,t) \left(D^{j} \bm{J}_Y\right) (k;x,t) 
	\end{gathered}
	\end{equation}
\end{enumerate}
\end{RHP}
It is a well-known result in the theory of RH problems that the inhomogenous problem  $\bm{M}_+(k) = \bm{M}_-(k) \bm{J}(k) + \bm{F}(k)$, $k \in \Gamma$, with $\bm{M}(k) \to 0$ as $k\to \infty$ has a unique solution in $L^p(\Gamma)$, $1 <p< \infty$ whenever $\bm{F} \in L^p(\Gamma)$ and there exists a unique solution of the associated homogenous problem $\widetilde{\bm{M}}_+(k) = \widetilde{\bm{M}}_-(k) \bm{J}(k)$, $k \in \Gamma$, with $\widetilde{\bm{M}}(k) \to \bm{I}$ as $k \to \infty$. In our setting, the associated homogenous problem is precisely the RH problem for $\bm{Y}$ for which we know a unique solution exists. Observing that $D^n \bm{J}_Y(k,x,t)$ is analytic for $k \in \Gamma_Y$, a simple induction argument shows $\bm{\mathcal{F}}^{(n)}(k;x,t)$ is analytic for any $n$. Since $\Gamma_Y$ is compact, analyticity of $\bm{\mathcal{F}}_n$ implies it is in $L^p$. It follows that $\bm{Y}$ has derivatives of all orders in $x$ and $t$. 
Finally since $\bm{Y}(k) = \bm{X}(k)$ for all sufficiently large $k$, it follows that 
\[
	q(x,t) = \lim_{k \to \infty} k \bm{Y}_{1,2}(k;x,t)
\]
lies in $C^\infty(\R_x \times \R_t)$.  Finally, one shows that $q(x,t)$ solves the mKdV equation \eqref{mkdv} by a standard Lax-pair argument. See for example step 4 of the  proof in \cite[Theorem 2.7]{GravaMinakov20}.
\end{proof}

\section{Fredholm determinant expression for the the KdV soliton + soliton gas}
\label{app:KdV_Fredholm}

The KdV soliton gas, with the possible presence of a separate soliton, is very similar to the mKdV case analyzed in the main body of this paper. The pure $N$-soliton solution of the KdV equation (\cite{Teschl}) is defined in terms of a $2$-dimensional row vector $\bm m$ such that
\begin{enumerate}[label=\arabic*.]
\item $\bm m(k;x,t)$ is meromorphic in $\mathbb{C}$, with simple poles at $\{ i \kappa_{j} \}_{j=1}^{N} \subset i\mathbb{R}_{+}$ 
and at the corresponding conjugate points $\{-i \kappa_{j} \}_{j=1}^{N} \subset i \mathbb{R}_{-}$;
\item $\bm m$ satisfies the residue conditions
\begin{equation}\label{residue_soliton}
\begin{gathered}
\Res_{k=i \kappa_{j}} \bm m (k) = \lim_{k \to i \kappa_{j}} \bm m(k) \begin{bmatrix} 0 & 0 \\ \displaystyle i \chi_j e^{2i\,  \theta(k; x,t) }  & 0 \end{bmatrix}\, ,  \\ 
\Res_{k=-i \kappa_{j}} \bm m(k) = \lim_{k \to -i {\kappa_{j}}} \bm m(k) \begin{bmatrix} 0 & \displaystyle -i \chi_j e^{-2i \, \theta(k; x,t) } \\ 0 & 0\end{bmatrix} \, ,
\end{gathered}
\end{equation}
where $\theta(k,x,t) = kx - 4tk^3$ and the norming constants $\chi_j\in \R_+$;
\item $\displaystyle \bm m(k)= \begin{bmatrix}1&1 \end{bmatrix} + \mathcal{O}\le(\frac{1}{k}\ri)$ as $k \rightarrow \infty$,
\item $\bm m$ satisfies the symmetry
\[
\bm m(-k)=\bm m(k)\begin{bmatrix}0&1\\1&0\end{bmatrix}\,.
\]
\end{enumerate}

The KdV solution then reads
\begin{gather}
q_N(x,t) = 2\partial_x \left( \lim_{k \to \infty} \frac{ k}{i} \left( \bm m(k;x,t)_1 - 1 \right) \right)  = 2 \ \partial_x \left( \sum_{k=1}^N  \alpha_k(x,t)\ri) \ ,
\end{gather}
where the second identity follows from the ansatz
\begin{eqnarray}
\label{M_soliton}
\bm m(k;x,t) = \left(  1 + \sum_{j=1}^{N} \frac{i \alpha_{j}(x,t)}{k- i \kappa_{j}} , \ 1 - \sum_{j=1}^{N} \frac{i \alpha_{j}(x,t)}{k+ i \kappa_{j}}  \right) \ .
\end{eqnarray}

Plugging the ansatz into the residue conditions gives a system of equations that yields
\begin{gather}
\sum_{k=1}^N \alpha_k(x,t)  =  -\Tr \le( (\bm I_N +\bm A)^{-1} \frac{\partial}{\partial x}\bm A \ri) = - \frac{\partial}{\partial x} \ln \det \le( \bm I_N  + \bm A\ri)
\end{gather}
where the matrix $\bm A$ has entries $$\bm A_{j\ell} = \frac{\sqrt{\chi_j}\sqrt{\chi_\ell} e^{i (\theta_j(x,t)+\theta_\ell(x,t))} }{\kappa_j+\kappa_\ell},$$
with $\theta_j = \theta(i \kappa_j; x,t)$.
Thus, 
\begin{gather}
q_N(x,t) = -2\frac{\partial^2}{\partial x^2}\ln \det  \le(\bm I_N  + \bm A \ri)\ . \label{eq9}
\end{gather}

As in the mKdV case, by rescaling $\chi_j \mapsto \frac{\chi_j}{{N}}$ and taking the limit $N \to \infty$ (under the same assumptions on the norming constants $\chi_j$'s), we obtain that the matrix determinant converges to a Fredholm determinant 
\[
q(x,t) = -2\frac{\partial^2}{\partial x^2}\ln \det  \Big( \Id_{L^2(\Sigma_1)} + \mathcal{K}\Big)
\]
with integral operator
\begin{align}
\mathcal{K}\le[f\ri](k)
&= \int_{\Sigma_1} \sqrt{r(k)} \sqrt{r(\xi)}   \frac{e^{i (\theta(k;x,t)+\theta(\xi; x,t))}}{k+\xi} \,  f(\xi)\,  \frac{\d \xi}{2\pi } 
\end{align}
where $\Sigma_1 = i( \eta_1,  \eta_2)$. Finally,  it is straightforward to verify that (see again \cite{BertolaCafasso11})
\begin{align*}
q(x,t) =  -2 \partial_x \le( \frac{1}{2}\int_{\Sigma_1\cup \Sigma_2} \Tr \le( {\bm \Gamma}_-^{-1} {\bm \Gamma}_-' \partial_x \bm J\bm J^{-1}\ri) \frac{\d k}{2\pi i} \ri)  
=2\, \partial_x  {\bm\Gamma}_{1;11}
\end{align*}
where ${\bm \Gamma}$ is the RH problem solution of the KdV soliton gas found in \cite{GirGraJenMcL}.

In the presence of an extra soliton, the solution can also be written as a Fredholm determinant. 
\begin{thm}
The function
\begin{gather}
q(x,t) =  -2\frac{\partial^2}{\partial x^2} \ln \det \Big( \Id_{L^2(\mathcal C)} + \mathcal{K}\Big) \label{FredDetKdVinfty}
\end{gather}
is the soliton + soliton gas solution of the KdV equation, where the integral operator $ \mathcal{K}$ has kernel
\begin{gather}
 {K}(k,z) :=    \frac{\sqrt{\widetilde r(k)} e^{i \, \theta(k; x,t)}  \, \sqrt{\widetilde r(z)} e^{i \theta(z; x,t)}}{2\pi (k+z)} , \quad k,z \in \mathcal C \ .
\end{gather}
where $\mathcal C = \Sigma_1 \cup \mathcal C_0$, with $\mathcal C_0$ being a small loop (oriented counterclockwise) circling the extra pole $i \kappa_0 \in i\R_+$ and not intersecting the real line nor $\Sigma_1$, and the function $\widetilde r(k)$ is defined as $\widetilde r(k) = r(k)$ for $k \in \Sigma_1$ (described in \eqref{cjasR1}) and $\widetilde r(k) = \frac{\chi_0}{k - i\kappa_0}$ for $k \in \mathcal C_0$.
\end{thm}

\addtocontents{toc}{\protect\setcounter{tocdepth}{0}}

\section*{Acknowledgments}

This manuscript was partially developed upon work supported by the National Science Foundation under Grant No. DMS-1928930 while TG, MG and KM participated in a program hosted by the Mathema\-tical Sciences Research Institute in Berkeley, California, during the Fall '21 semester  ``Universality and Integrability in Random Matrix Theory and Interacting Particle Systems". 
TG, MG, RJ, KM would also like to thank the Isaac Newton Institute for Mathematical Sciences in Cambridge, UK, for support (EPSRC grant No. EP/R014604/1) and hospitality during the programme ``Dispersive hydrodynamics: mathematics, simulation and experiments, with applications in nonlinear waves" (HYD2) in the Summer and Fall '22 semesters, where part of the work on this paper was undertaken. 

TG  acknowledges the  funding from the European Union’s H2020 research and innovation programme under the Marie Sk\l{}odowska–Curie grant No. 778010 IPaDEGAN,  the research project "Mathematical Methods in Non Linear Physics" (MMNLP), Gruppo 4-Fisica Teorica of INFN and the support of GNFM-INDAM group.   
KM was supported in part by the National Science Foundation under grant DMS-1733967.  MG acknowledges the support of the Natural Sciences and Enginee\-ring Research Council of Canada (NSERC) grant No. RGPIN-2022-04106 and the partial support of the Simons Foundation Fellowship while visiting the Isaac Newton Institute. RJ acknowledges the support of the Simons Foundation under grant 853620. 

The authors thank Iryna Egorova for interesting and useful discussions regarding decay rate of the potentials to the elliptic background for $x \to + \infty$, and Percy Deift for many insightful suggestions.

\printbibliography


\end{document}